\documentclass
[pre,nobibnotes,10pt,preprint,byrevtex,notitlepage,twocolumn]{revtex4}%
\usepackage{amsfonts}
\usepackage{amsmath}
\usepackage{amssymb}
\usepackage{graphicx}%
\providecommand{\U}[1]{\protect\rule{.1in}{.1in}}
\newtheorem{theorem}{Theorem}

\newtheorem{conclusion}[theorem]{Conclusion}

\newtheorem{corollary}[theorem]{Corollary}
\newtheorem{criterion}[theorem]{Criterion}
\newtheorem{definition}[theorem]{Definition}

\newtheorem{notation}[theorem]{Notation}

\newtheorem{proposition}[theorem]{Proposition}
\newtheorem{remark}[theorem]{Remark}

\newenvironment{proof}[1][Proof]{\noindent\textbf{#1.} }{\ \rule{0.5em}{0.5em}}
\begin{document}
\preprint{ }
\preprint{UATP/2105}
\title{A Review of the System-Intrinsic Nonequilibrium Thermodynamics in Extended
Space (MNEQT) with Applications }
\author{P.D. Gujrati,$^{1,2}$ }
\affiliation{$^{1}$Department of Physics, $^{2}$Department of Polymer Science, The
University of Akron, Akron, OH 44325}
\email{pdg@uakron.edu}

\begin{abstract}
The review deals with a \emph{novel approach} (MNEQT) to nonequilibrium
thermodynamics (NEQT) that is based on the concept of internal equilibrium
(IEQ) in an enlarged state space $\mathfrak{S}_{\mathbf{Z}}$\ involving\emph{
internal variables as additional state variables}. The IEQ\ macrostates are
unique in $\mathfrak{S}_{\mathbf{Z}}$\ and have no memory just as EQ
macrostates are in the EQ state space $\mathfrak{S}_{\mathbf{X}}%
\subset\mathfrak{S}_{\mathbf{Z}}$. The approach provides a clear strategy to
identify the internal variables for any model through several examples. The
MNEQT deals directly with system-intrinsic quantities, which are very useful
as they fully describe irreversibility. Because of this, MNEQT solves a
long-standing problem in NEQT of identifying a unique global temperature $T$
of a system, \emph{thus fulfilling Planck's dream of a global temperature for
any system}, even if it is not uniform such as when it is driven between two
heat baths; $T$ has the conventional interpretation of satisfying the Clausius
statement that the \emph{exchange macroheat }$d_{\text{e}}Q$\emph{ flows from
hot to cold}, and other sensible criteria expected of a temperature. The
concept of the generalized macroheat $dQ=d_{\text{e}}Q+d_{\text{i}}Q$ converts
the Clausius inequality $dS\geq d_{\text{e}}Q/T_{0}$ for a system in a medium
at temperature $T_{0}$ into the \emph{Clausius equality} $dS\equiv dQ/T$,
which also covers macrostates with memory, and follows from the extensivity
property. The equality also holds for a NEQ isolated system. The novel
approach is extremely useful as it also works when no internal state variables
are used to study nonunique macrostates in the EQ state space $\mathfrak{S}%
_{\mathbf{X}}$ at the expense of explicit time dependence in the entropy that
gives rise to memory effects. To show the usefulness of the novel approach, we
give several examples such as irreversible Carnot cycle, friction and Brownian
motion, the free expansion, etc.

\end{abstract}
\date{\today}
\maketitle

\section{Introduction\label{Sec-Introduction}}

Thermodynamics of a system out of equilibrium (EQ)
\cite{DeDonder,Prigogine71,deGroot,Bedeaux,Kuiken,Ottinger,Eu0,Evans} is far
from a complete science in contrast to the EQ thermodynamics based on the
original ideas of Carnot, Clapeyron, Clausius, Thomson, Maxwell, and many
others \cite{Prigogine,Reif,Landau,Waldram,Balian,Kestin,Fermi,Woods} that has
by now been firmly established in physics, thanks to Boltzmann
\cite{Boltzmann} and Gibbs \cite{Gibbs}. Therefore, it should not be a
surprise that there are currently many schools of nonequilibrium (NEQ)
thermodynamics (NEQT), among which are the most widely known schools of
local-EQ thermodynamics, rational thermodynamics, extended thermodynamics, and
GENERIC thermodynamics \cite{Muschik0,Jou0}. This pedagogical review and
various applications in different contexts deal with a recently developed
NEQT, which we have termed MNEQT, with M referring to a macroscopic treatment
in terms of \emph{system-intrinsic}\ (SI) quantities of the system $\Sigma$ at
each instant. These quantities are normally taken to be \emph{extensive}
SI-quantities, and are used as state variables to describe a macrostate
$\mathcal{M}$ of $\Sigma$. The MNEQT has met with success as we will describe
in this review so it is desirable to introduce it to a wider class of readers
and supplement it with many nontrivial applications.
\begin{figure}
[ptb]
\begin{center}
\includegraphics[
height=3.3209in,
width=3.1739in
]%
{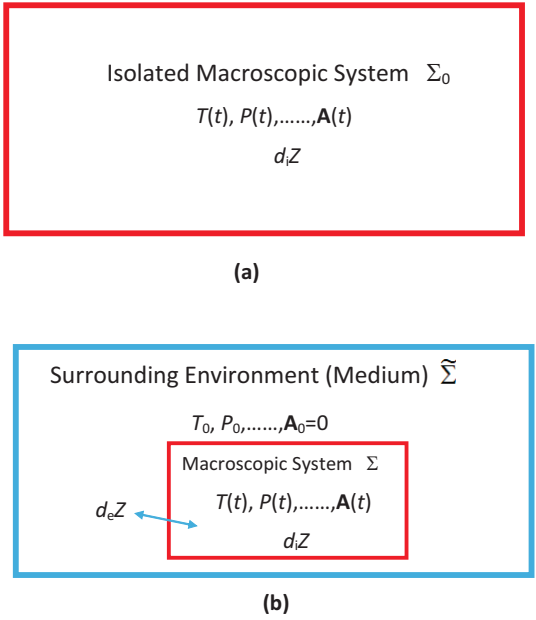}%
\caption{{}(a) An isolated nonequilibrium system $\Sigma_{0}$\ with internally
generated $d_{\text{i}}Z$ driving it towards equilibrium, during which its
SI-fields $T(t),P(t),\cdots,\mathbf{A}(t)$ continue to change to their
equilibrium values; $d_{\text{i}}Z_{k}$ denote the microanalog of
$d_{\text{i}}Z$. The sign of $d_{\text{i}}Z$ is determined by the second law.
(b) A nonequilibrium systen $\Sigma$ in a surrounding medium $\widetilde
{\Sigma}$, both forming an isolated system $\Sigma_{0}$. The macrostates of
the medium and the system are characterized by their fields $T_{0}%
,P_{0},...,\mathbf{A}_{0}=0$ and $T(t),P(t),...,\mathbf{A}(t)$, respectively,
which are different when the two are out of equilibrium. Exchange quantities
($d_{\text{e}}Z$) carry a suffix "e" and irreversibly generated quantities
($d_{\text{i}}Z$) within the system by a suffix "i" by extending the Prigogine
notation. Their sum $d_{\text{e}}Z+d_{\text{i}}Z$ is denoted by $dZ$, which is
a system-intrinsic quantity (see text). }%
\label{Fig_System}%
\end{center}
\end{figure}

We take $\Sigma$ as a \emph{discrete} system in that it is separated from its
surrounding medium $\widetilde{\Sigma}$ (if it exists) with which it
interacts; see Fig. \ref{Fig_System}. Such a system is also called a Schottky
system \cite{Muschik0,Schottky,Muschik-2020} Because of the use of
SI-quantities, the MNEQT differs from all other existing approaches to the
NEQT in that the latter invariably deal with exchange quantities with
$\widetilde{\Sigma}$, which are \emph{medium-intensive }(MI) quantities that
differ from SI-quantities in important ways in a NEQ process as we will see.
We will use \r{M}NEQT to refer to the latter approaches, with \r{M} referring
to the use of macroscopic exchange quantities. The corresponding NEQ
statistical mechanics of the MNEQT is termed $\mu$NEQT, in which $\mu$ refers
to the treatment of $\Sigma$ in terms of microstates, which form a countable
set $\left\{  \mathfrak{m}_{k}\right\}  $, with $k$ counting various
microstates. The existence of the $\mu$NEQT is possible only because of the
use of SI-quantities in the MNEQT. These quantities are easily associated with
$\left\{  \mathfrak{m}_{k}\right\}  $ as will become clear here. This ability
in the MNEQT immediately distinguishes it from the\ \r{M}NEQT as the latter
cannot lead directly to a statistical mechanical treatment with $\left\{
\mathfrak{m}_{k}\right\}  $. Therefore, we believe that the MNEQT and $\mu
$NEQT will prove very useful. All quantities pertaining to $\mathcal{M}$ are
called \emph{macroquantities}, while those pertaining to microstates contain
an index $k$ and are called \emph{microquantities} for simplicity in this review.

While most of the review deals with an isolated system or an interacting
system in a medium, we will occasionally also consider a system interacting
with two different media such as in Fig. \ref{Fig-Sys-TwoSources}, to study
driven and steady macrostates \cite{Oono,Sasa,Bejan} at $\tau\sim
\tau_{\text{st }}$for which there is no EQ macrostate having unique values of
the temperature, pressure, etc. as long as we do not allow the media to come
to EQ with each other, which takes much longer time $\tau_{\text{EQ }}%
>>\tau_{\text{st }}$. A steady or an unsteady macrostate always gives rise to
irreversible entropy generation so it truly belongs to the realm of the NEQT.
What makes the MNEQT a highly desirable approach is that it can also deal with
unsteady processes easily as we will do. \ \ \
\begin{figure}
[ptb]
\begin{center}
\includegraphics[
trim=0.601850in 0.802747in 0.603817in 0.803101in,
height=0.9971in,
width=2.7025in
]%
{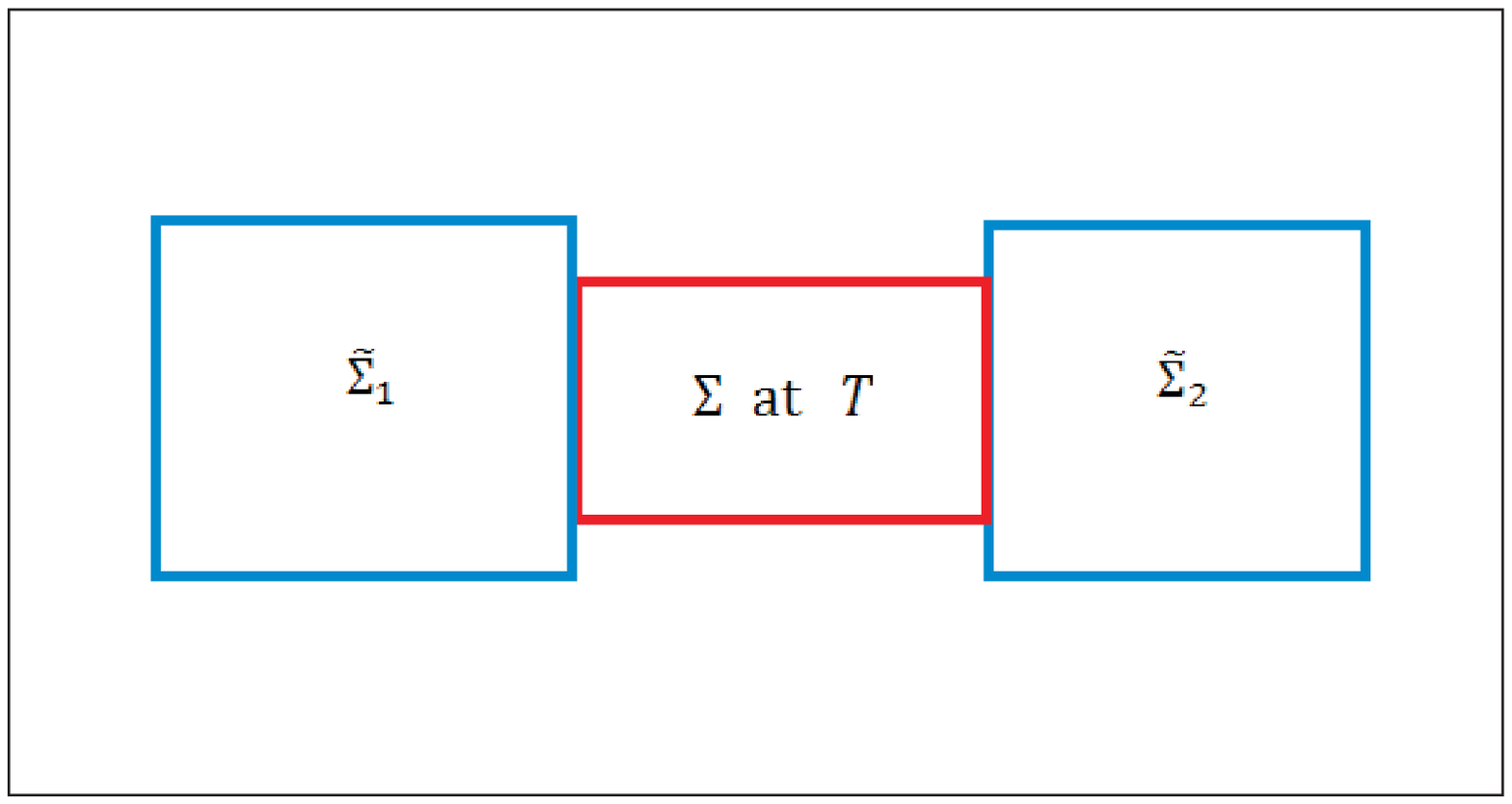}%
\caption{A system driven between two sources that are different in their
fields; see Fig. \ref{Fig_System}. If they are the same, the situation reduces
to that in Fig. \ref{Fig_System}(a). Later in Sec. \ref{Sec-Composite System}
we consider this situation between two heat sources $\widetilde{\Sigma
}_{\text{h}1}$ and $\widetilde{\Sigma}_{\text{h}2}$, where we treat $\Sigma$
as a composite system as an application of our approach.}%
\label{Fig-Sys-TwoSources}%
\end{center}
\end{figure}
\ \ \ \ \ \ \ \ 

\subsection{Unique Macrostates in Extended State Space}

The firm foundation of EQ statistical mechanics is accomplished by using the
concept of microstates $\mathfrak{m}_{k}$ of $\Sigma$ and their EQ
probabilities $p_{k\text{eq}}$. This is feasible as the EQ macrostate
$\mathcal{M}_{\text{eq}}$ is \emph{unique} in the EQ state space
$\mathfrak{S}_{\mathbf{X}}$ spanned by the set of observables \cite{Note}
$\mathbf{X}=(E,V,\cdots)$, where $E,V,\cdots$ are the energy, volume, etc.
using standard notation (we do not show the number of particles $N$ as we keep
it fixed throughout this review \cite{Note0}; see later, however). But the
same cannot be said about its extension to describe NEQT, since NEQ
macrostates $\mathcal{M}$ in $\mathfrak{S}_{\mathbf{X}}$ are not unique
\cite{Kestin} even if they appear in a process between two EQ macrostates,
which we will always denote by $\overline{\mathcal{P}}$ and use $\mathcal{P}$
for any general process including $\overline{\mathcal{P}}$. It is clear that
unless we can specify the microstates for $\mathcal{M}$\ uniquely, we cannot
speak of their probabilities $p_{k}$ in a sensible way, but this is precisely
what we need to establish a rigorous \emph{NEQ statistical mechanics} of
thermodynamic processes
\cite{Keizer-Book,Schuss,Coffee,Jarzynski,Sekimoto,Seifert,Gujrati-LangevinEq,
Stratonovich,Bochkov}. The system is usually surrounded by an external medium
$\widetilde{\Sigma}$, which we always take to be in EQ; see Fig.
\ref{Fig_System}(b). The combination $\Sigma_{0}$ as the union $\Sigma
\cup\widetilde{\Sigma}$ forms an isolated system, which we assume to be stationary.

The lack of uniqueness of $\mathcal{M}$\ is handled in the MNEQT by using a
well-established practice
\cite{deGroot,Coleman,Maugin,Gujrati-II,Langer,Prigogine,Pokrovskii,Gujrati-Hierarchy}
by considering a properly extended state space $\mathfrak{S}_{\mathbf{Z}}$
spanned by $\mathbf{Z}\doteq\mathbf{X}\cup\boldsymbol{\xi}$, by including a
set $\boldsymbol{\xi}$ of \emph{internal variables},\emph{ }in which NEQ
macrostates $\mathcal{M}$ and microstates of interest can be uniquely
specified during the entire process $\mathcal{P}$. Here, $\boldsymbol{\xi}$ is
internally generated within $\Sigma$ so it \emph{cannot} be controlled by the
observer \cite{Note}. The use of internal variables in glasses, prime examples
of NEQ systems, is well known, where they give rise to distinct relaxations of
the glassy macrostate \cite{Davies,Gutzow,Nemilov,Langer,Goldstein}. Their
justification is based on the ideas of chemical reactions \cite{Prigogine0},
and has been formalized recently by us \cite{Gujrati-Hierarchy} to any NEQ
macrostate $\mathcal{M}$. It is well known that internal variables contribute
to irreversibility in $\mathcal{P}$, which justifies their important role in
the NEQT. We give several examples for their need later in the review and a
clear strategy to identify them for computation under different conditions. In
$\mathfrak{S}_{\mathbf{Z}}$, the \emph{unique} $\mathcal{M}$'s are specified
by the collection $\left\{  \mathfrak{m}_{k},p_{k}\right\}  $ of two
\emph{independent} quantities, which form a probability space $\mathbb{P}$. We
can then pursue any $\mathcal{P}$ followed by $\mathcal{M}(t)$ as the latter
evolves in time $t$ to another (EQ or NEQ) unique macrostate. A major
simplification occurs when this independence is maintained at each instant so
that during the evolution, each microstate $\mathfrak{m}_{k}$ follows a
trajectory (such as a Brownian trajectory) $\gamma_{k}$ whose characteristics
do not depend on $p_{k}(t)$ as a function of time $t$ \cite[for example]%
{Gujrati-LangevinEq}; the latter, of course, determines the trajectory
probability $p_{\gamma_{k}}$. Thus, $\left\{  \gamma_{k},p_{\gamma_{k}%
}\right\}  $\ uniquely specifies $\mathcal{P}$ in $\mathbb{P}$. For the same
collection $\boldsymbol{\gamma}\doteq\left\{  \gamma_{k}\right\}  $, different
choices of $\left\{  p_{\gamma_{k}}\right\}  $ describe different processes.

\subsection{Layout}

The review is divided into two distinct parts. The first part consisting of
Sects. \ref{Sec-InternalVariables}-\ref{Sec-EntropyCalculation} deals with the
up-to-date foundation of the MNEQT for $\Sigma$, regardless of whether it is
isolated or interacting (in the presence of one or more a external sources).
We have tried to make the new concepts and their physics as clear as possible
so a reader can appreciate the foundation of the MNEQT, which can be complex
at times. The most important one is that of the NEQ temperature $T$ as
anticipated by Planck that is required to be defined globally over the system
so that it can satisfy the Clausius statement about macroheat flow from hot to
cold. The concepts of the \emph{generalized macroheat} $dQ$ and the
\emph{generalized macrowork} $dW$ are directly and uniquely defined in terms
of SI-quantities that pertain to the system alone. Thus, they are capable of
describing the irreversibilty in the system. A clear strategy to identify
internal variables is discussed for carrying out thermodynamic computation.
The other part consisting of Sects. \ref{Sec-Applications}%
-\ref{Sec-Free Expansion} deals with various applications of the MNEQT, many
of which cannot be studied within the \r{M}NEQT without imposing additional
requirements. This part provides an abundant evidence of successful
implementation of the MNEQT.

The layout of the paper is as follows. In the next section, we introduce our
notation and give some useful definitions and new concepts without any
explanation. This section is only for bookkeeping so that readers can come
back to it to refresh the concepts in the manuscript later when they are not
sure of their meanings. The next six sections deal with various new concepts
and theory behind the MNEQT. Sect. \ref{Sec-InternalVariables} introduces the
central concept of internal variables that are required for arbitrary NEQ
macrostates $\mathcal{M}$. Many examples are given to highlight their
importance for $\mathcal{M}$. They form the extended state space
$\mathfrak{S}_{\mathbf{Z}}$, which contains the state space $\mathfrak{S}%
_{\mathbf{X}}$ as a proper subspace. The internal variables are irrelevant for
EQ macrostates in $\mathfrak{S}_{\mathbf{X}}$. Sect. \ref{Sec-NEQ-S} is also
very important, where we introduce the concept of NEQ entropies based on the
original ideas of Boltzmann. In this sense, the derivation of this entropy is
thermodynamic in nature, and gives rise to an expression of $S$ that
generalizes the Gibbs formulation of the entropy to NEQ macrostates. Using
this formulation, we reformulate a previously given proof of the second law.
In Sect. \ref{Sec-HamiltonianTrajectories}, we formulate the statistical
mechanics of the MNEQT, and discuss the statistical significance of $dW$ and
$dQ$ that provide a reformulation of the first law in terms of SI-quantities
for any\emph{ arbitrary process} between any two arbitrary macrostates.\emph{
}The SI-quantities are determined by $\Sigma$ alone, even if it is interacting
with its exterior, and its usage has neither been noted nor has been
appreciated by other workers in the field. These generalized macroquantities
are different from exchange macrowork and macroheat. In this reformulation,
the first law includes the second law in that it contains all the information
of the irreversibility encoded in $\mathcal{M}$. This formulation applies
equally well to the exchange energy change $d_{\text{e}}E$ and the internally
generated energy change $d_{\text{i}}E$, which shows the usefulness of the
formulation. In Sect. \ref{Sec-UniqueMicrostates}, which is the most important
section for the foundation of the MNEQT, we discuss the conditions for
$\mathcal{M}$ to be uniquely specified in $\mathfrak{S}_{\mathbf{Z}}$, and
introduce the concept of the internal equilibrium (IEQ) to specify
$\mathcal{M}_{\text{ieq}}$ in $\mathfrak{S}_{\mathbf{Z}}$. A parallel is drawn
between $\mathcal{M}_{\text{ieq}}$ and $\mathcal{M}_{\text{eq}}$ so that many
results valid for $\mathcal{M}_{\text{eq}}$ also apply to $\mathcal{M}%
_{\text{ieq}}$, except that the latter has nonzero entropy generation
($d_{\text{i}}S\geq0$). The entropy of $\mathcal{M}_{\text{ieq}}$ is a state
function in $\mathfrak{S}_{\mathbf{Z}}$, while that of\ a macrostate
$\mathcal{M}_{\text{nieq}}$ that lies outside of $\mathfrak{S}_{\mathbf{Z}}$
is not a state function. see later. The entropy of $\mathcal{M}$ that lies
outside $\mathfrak{S}_{\mathbf{X}}$ is similarly not a state function of
$\mathbf{X}$. We show that the NEQ entropy in Sect. \ref{Sec-NEQ-S} reduces to
the thermodynamic EQ entropy for $\mathcal{M}_{\text{eq}}$ and to the
thermodynamic IEQ entropy for $\mathcal{M}_{\text{ieq}}$. We introduce the
concept of a NEQ\ thermodynamic temperature $T$ as an inverse entropy
derivative ($\partial S/\partial E$). We show that this concept satisfies
various sensible requirements (C1-C4) of a thermodynamic temperature, which is
global over the entire system even if it is inhomogeneous. This, we believe,
solves a long-standing problem of a NEQ\ temperature. In terms of $T$, we show
that the Clausius inequality in the \r{M}NEQT is turned into an equality in
the MNEQT as shown in Sect. \ref{Sec_Clausius_Equality}. In Sect.
\ref{Sec-EntropyCalculation}, which is the last section of the first part, we
use the idea of chemical equilibrium to show how entropy is generated in an
isolated system. We now turn to the second part of the review. In Sect.
\ref{Sec-Applications}, we consider various applications of the MNEQT ranging
from a simple system to composite systems under various conditions. This
section is very important in that we establish here that we can treat a system
either (i) as a "black box" $\Sigma_{\text{B}}$ of temperature $T$ but without
knowing anything about its interior, or (ii) as a composite system
$\Sigma_{\text{C}}$ for which we have a detailed information about its
interior inhomogeneity. Both realizations give the same irreversible entropy
generation. Thus, we can always treat a system as $\Sigma_{\text{B}}$ of
temperature $T$, whose study then becomes simpler. In Sect.
\ref{Sec-Tool-Narayan}, we apply our approach to a glassy system and derive
the famous Tool-Narayanaswamy equation for the glassy temperature $T$. In
Sect. \ref{Sec-CarnotCycle}, we apply the MNEQT to study an irreversible
Carnot cycle and determine its efficiency in terms of $\Delta_{\text{i}}S$. In
Sect. \ref{Sec-Friction}, we apply the MNEQT to a very important problem of
friction and the Brownian motion. In Sect. \ref{Sec-Free Expansion}, we
consider a classical and a quantum expansion. In the classical case, we study
the expansion in $\mathfrak{S}_{\mathbf{X}}$, where $\mathcal{M}$ is a non-IEQ
macrostate, with an explicit time-dependence, and in $\mathfrak{S}%
_{\mathbf{Z}}$, where $\mathcal{M}$ is a an IEQ macrostate, with no explicit
time-dependence, and show that we obtain the same result. The quantum
expansion is only studied in $\mathfrak{S}_{\mathbf{Z}}$. The last section
provides an extensive discussion of the MNEQT and draws some useful conclusions.

\section{Notation, Definitions and New Concepts\label{Sec-Notation}}

\subsection{Notation}

Before proceeding further, it is useful to introduce in this section our
notation to describe various systems and their behavior and new concepts for
their understanding without much or any explanation (that will be offered
later in the review where we discuss them) so that a reader can always come
back here to be reminded of their meaning in case of confusion. In this sense,
this section plays an important role in the review for the purpose of bookkeeping.

Even though $\Sigma$\ is macroscopic in size, it is extremely small compared
to the medium $\widetilde{\Sigma}$; see Fig. \ref{Fig_System}(b). The medium
$\widetilde{\Sigma}$ consists of two parts: a work source $\widetilde{\Sigma
}_{\text{w}}$ and a macroheat source $\widetilde{\Sigma}_{\text{h}}$, both of
which can interact with the system $\Sigma$ directly but not with each other.
This separation allows us to study macrowork and macroheat exchanges
separately. We will continue to use $\widetilde{\Sigma}$ to refer to both of
them together. The collection $\Sigma_{0}=\Sigma\cup\widetilde{\Sigma}$ forms
an isolated system, which we assume to be stationary. The system in Fig.
\ref{Fig_System}(a) is an isolated system, which we may not divide into a
medium and a system. Each medium in Fig. \ref{Fig-Sys-TwoSources}, although
not interacting with each other, has a similar relationship with $\Sigma$,
except that the collection $\Sigma_{0}=\Sigma\cup\widetilde{\Sigma}_{1}%
\cup\widetilde{\Sigma}_{2}$ forms an isolated system. In case they were
mutually interacting, they can be treated as a single medium. In the
following, we will mostly focus on Fig. \ref{Fig_System} to introduce the
notation, which can be easily extended to Fig. \ref{Fig-Sys-TwoSources}.

We will use the term "body" to refer to any of $\Sigma,\widetilde{\Sigma}$,
and $\Sigma_{0}$ in this review and use $\Sigma_{\text{b}}$ to denote it.
However, to avoid notational complication, we will use the notation suitable
for $\Sigma$ for $\Sigma_{\text{b}}$ if no confusion would arise in the
context. As the mechanical aspect of a body is described by the Hamiltonian
$\mathcal{H}$, whose value determines its macroenergy $E$, it plays an
important role in thermodynamics. Therefore, it is convenient to introduce
\begin{equation}
\mathbf{w}\doteq\mathbf{X}\backslash E=(V,\cdots),\mathbf{W}\doteq
\mathbf{Z}\backslash E=(V,\cdots,\boldsymbol{\xi}), \label{w-W-eq}%
\end{equation}
where $\backslash E$ means to delete $E$ from the set, and $\cdots$ refers to
the rest of the elements in $\mathbf{X}$ besides $V$. We use $\mathbf{x}$\ to
denote the collection of coordinates and momenta of the $N$ particles in the
phase space of $\Sigma$. The variable $\mathbf{W}$ appears as a parameter set
in the Hamiltonian $\mathcal{H}(\left.  \mathbf{x}\right\vert \mathbf{W})$ of
$\Sigma$\ that can be varied in a process with a concomitant change in
$\mathcal{H}$. As internal variables play no role in EQ, $\mathbf{W=w}$\ in
EQ. We will normally employ a discretization of the phase space in which we
divide it into cells $\delta\mathbf{x}$, centered at $\mathbf{x}$ and of some
small size, commonly taken to be $\left(  2\pi\hbar\right)  ^{3N}$. The cells
cover the entire phase space. To account for the identical nature of the
particles, the number of cells and the volume of the phase space is assumed to
be divided by $N!$ to give distinct arrangements of the particles in the
cells, which are indexed by $k$ $=1,2,\cdots$ and write them as $\left\{
\delta\mathbf{x}_{k}\right\}  $; the center of $\delta\mathbf{x}_{k}$ is at
$\mathbf{x}_{k}$. These cells represents the microstates $\left\{
\mathfrak{m}_{k}\right\}  $. The energy and probability of these cells are
denoted by $\left\{  E_{k},p_{k}\right\}  $ in which $E_{k}(\mathbf{W})$ is a
function of $\mathbf{W}$. Different choices of $\left\{  p_{k}\right\}  $ for
the same set $\left\{  \mathfrak{m}_{k},E_{k}\right\}  $ describes different
macrostates for a given $\mathbf{W}$, one of which corresponding to $\left\{
p_{k}^{\text{eq}}\right\}  $ uniquely specifies an EQ macrostate
$\mathcal{M}_{\text{eq}}$; all other states are called NEQ macrostates
$\mathcal{M}$. Among $\mathcal{M}$ are some special macrostates $\mathcal{M}%
_{\text{ieq}}$ that are said to be in internal equilibrium (IEQ); the rest are
nonIEQ macrostates $\mathcal{M}_{\text{nieq}}$. An arbitrary macrostate
$\mathcal{M}_{\text{arb}}$ refers to either an EQ or a NEQ macrostate.

We use a suffix $0$ to denote all quantities pertaining to $\Sigma_{0}$, a
tilde $(\widetilde{})$ for all quantities pertaining to $\widetilde{\Sigma}$,
and no suffix for all quantities pertaining to $\Sigma$ even if it is
isolated. Thus, the set of observables are denoted by $\mathbf{X}%
_{0},\widetilde{\mathbf{X}}$ and $\mathbf{X}$, respectively, and the set of
state variables by $\mathbf{Z}_{0},\widetilde{\mathbf{Z}}$ and $\mathbf{Z}$,
respectively, in the state space $\mathfrak{S}_{\mathbf{Z}}$; the set of
internal variables are $\boldsymbol{\xi}_{0},\widetilde{\boldsymbol{\xi}}$ and
$\boldsymbol{\xi}$, respectively. As $\widetilde{\Sigma}$ is taken to be in
EQ, weakly interacting with and is extremely large compared to $\Sigma$, all
its fields can be safely taken to be the fields associated with $\Sigma_{0}$
so can be denoted by using the suffix $0$.

In the discrete approach, $\Sigma$ and $\widetilde{\Sigma}$ are spatially
disjoint so%
\[
V_{0}=V+\widetilde{V}.
\]
They are weakly interacting so that their energies are \emph{quasi-additive}%
\[
E_{0}=E+\widetilde{E}+E_{\text{int}}\simeq E+\widetilde{E},
\]
where $E_{\text{int}}$ is the weak interaction energy between $\Sigma$ and
$\widetilde{\Sigma}$ and can be neglected to a good approximation. We also
take them to be \emph{quasi-independent} \cite{Gujrati-II} so that their
entropies also become quasiadditive:%
\begin{align}
S_{0}(\mathbf{X}_{0}\mathbf{,}t)  &  =S(\mathbf{X}(t)\mathbf{,}t)+\widetilde
{S}(\widetilde{\mathbf{X}}(t))+S_{\text{corr}}(t)\nonumber\\
&  \simeq S(\mathbf{X}(t)\mathbf{,}t)+\widetilde{S}(\widetilde{\mathbf{X}%
}(t)); \label{Entropy-Additivity}%
\end{align}
here, $S_{\text{corr}}(t)$ is a negligible contribution to the entropy due to
quasi-independence between $\Sigma$ and $\widetilde{\Sigma}$, and can also be
neglected to a good approximation. The entropy $\widetilde{S}$ has no explicit
time dependence as $\widetilde{\Sigma}$ is always assumed to be in
equilibrium, and $\mathbf{X}_{0}$ remains constant for the isolated system
$\Sigma_{0}$. The discussion of quasi-independence and its distinction from
weak interaction has been carefully presented elsewhere \cite[$S_{\text{corr}%
}$ was called $S_{\text{int}}$ there; however, $S_{\text{corr}}$ seems to be
more appropriate]{Gujrati-II} for the first time, which we summarize as
follows. The concept of quasi-independence is determined by the thermodynamic
concept of \emph{correlation length }$\lambda_{\text{corr}}$, which is a
property of macrostates, and can be much larger than the interaction length
between particles. A simple well-known example is of the correlation length of
a nearest neighbor Ising model, which can be extremely large near a critical
point than the nearest neighbor distance between the spins. This distinction
is usually not made explicit in the literature. For quasi-independence between
$\Sigma$ and $\widetilde{\Sigma}$, we require their sizes to be larger than
$\lambda_{\text{corr}}$. Throughout this review, we will think of the above
\emph{approximate equalities} as equalities to make the energies to be
additive by neglecting the interaction energy between $\Sigma$ and
$\widetilde{\Sigma}$, which is a standard practice in the field, but also
assuming quasi-independence between them to make the entropies to be additive,
which is not usually mentioned as a requirement in the literature.

For a reversible process, the entropy of each macrostate $\mathcal{M}%
_{\text{eq}}(t)\in\mathfrak{S}_{\mathbf{X}}$ of a body along the process is a
state function of $\mathbf{X}(t)$, but not for an irreversible process for
which $\mathcal{M}(t)\mathfrak{\notin S}_{\mathbf{X}}$. Their entropies are
written as $S(\mathbf{X}(t)\mathbf{,}t)$
\cite{Gujrati-Entropy1,Gujrati-Entropy2} with an explicit time dependence. In
general \cite{Gujrati-Symmetry,Landau,Gujrati-Entropy1,Gujrati-Entropy2},
\begin{equation}
S(\mathbf{X}(t)\mathbf{,}t)\leq S(\mathbf{X}(t));\text{ fixed }\mathbf{X}%
(t)\mathbf{.} \label{EntropyBound}%
\end{equation}

The equilibrium values of various entropies are always denoted with no
explicit time dependence such as by $S_{0}(\mathbf{X}_{0})$ for $\Sigma_{0}$.
These entropies represent the maximum possible values of the entropies of a
body as it relaxes and comes to equilibrium for a given set of observables.
Once in equilibrium, the body\ will have no memory of its original macrostate.
The set $\mathbf{X}_{0}$, which includes its energy $E_{0}$ among others,
remains constant for $\Sigma_{0}$ as it relaxes \cite{Note}. This notion is
also extended to a body in internal equilibrium.

\begin{notation}
\label{Notation}We use modern notation \cite{Prigogine,deGroot} and its
extension, see Fig. \ref{Fig_System}, that will be extremely useful to
understand the usefulness of our novel approach. Any infinitesimal and
extensive \emph{system-intrinsic} quantity $dY(t)$ during an arbitrary process
$d\mathcal{P}$ can be partitioned as%
\begin{equation}
dY(t)\equiv d_{\text{e}}Y(t)+d_{\text{i}}Y(t), \label{Y-partition}%
\end{equation}
where $d_{\text{e}}Y(t)$ is \emph{the change caused by exchange (e) with the
medium} and $d_{\text{i}}Y(t)$ is its \emph{change due to internal or
irreversible (i) processes going on within the system}.
\end{notation}

\subsection{Some Definitions and New Concepts}

\begin{definition}
\label{Def-Observables-InternalVariables-StateVariables}Observables
$\mathbf{X}=(E,V,N,\cdots)$ of a system are quantities that can be controlled
from outside the system, and internal variables $\boldsymbol{\xi}=(\xi_{1}%
,\xi_{2},\xi_{3},\cdots)$ are quantities that cannot be controlled. Their
collection $\mathbf{Z}=\mathbf{X}\cup\boldsymbol{\xi}$ is called the set of
state variables in the state space $\mathfrak{S}$.
\end{definition}

\begin{definition}
\label{Def-SystemIntrisic}A \emph{system-intrinsic }quantity is a quantity
that pertains to the system alone and can be used to characterize the system.
A \emph{medium--intrinsic }quantity is a quantity that is solely determined by
the medium alone and can be used to characterize the exchange between the
system and the medium.
\end{definition}

\begin{definition}
A macrostate in $\mathfrak{S}_{\mathbf{X}}$ or $\mathfrak{S}_{\mathbf{Z}}$ is
a collection $\left\{  \mathfrak{m}_{k},p_{k}\right\}  $ of microstates
$\mathfrak{m}_{k}$ and their probabilities $p_{k},k=1,2,\cdots$. In general,
$p_{k}$ are functions of $\mathbf{X}$ or $\mathbf{Z}$, depending on the state
space. They are implicit function of time $t$ through them; they may also
depend explicitly on time $t$ if not unique in the state space.. For an EQ or
an IEQ macrostate, $p_{k}$ have no explicit dependence on $t$. For EQ states,
$p_{k\text{ }}$have no time-dependence. It is through the microstate
probabilities that thermodynamics gets its stochastic nature.
\end{definition}

\begin{definition}
The collection $\left\{  \mathfrak{m}_{k},p_{k}\right\}  $ provides a complete
microscopic or statistical mechanical description of thermodynamics for
$\mathcal{M}_{\text{arb}}$ in some state space $\mathfrak{S}$ in which one
deals with macroscopic or ensemble averages, see Definition
\ref{Def-EnsembleAverage}, over $\left\{  \mathfrak{m}_{k}\right\}  $ of
microstate variables. The same collection $\left\{  \mathfrak{m}_{k}%
,p_{k}\right\}  $ also provides a microscopic description of a microstate and
its probability in any arbitrary process.
\end{definition}

\begin{definition}
\label{Def-NEQ-States}The nonequilibrium macrostates can be classified into
two classes:
\end{definition}

\begin{enumerate}
\item[(a)] \emph{Internal-equilibrium macrostates }(IEQ): The nonequilibrium
entropy $S(\mathbf{X,}t)$ for such a macrostate is a state function
$S(\mathbf{Z})$ in the larger nonequilibrium state space $\mathfrak{S}%
_{\mathbf{Z}}$ spanned by $\mathbf{Z}$; $\mathfrak{S}_{\mathbf{X}}$ is a
proper subspace of $\mathfrak{S}_{\mathbf{Z}}$: $\mathfrak{S}_{\mathbf{X}%
}\subset\mathfrak{S}_{\mathbf{Z}}$. As there is no explicit time dependence,
there is no memory of the initial macrostate in IEQ macrostates.

\item[(b)] \emph{Non-internal-equilibrium macrostates }(NIEQ): The
nonequilibrium entropy for such a macrostate is not a state function of the
state variable $\mathbf{Z}$. Accordingly, we denote it by $S(\mathbf{Z},t)$
with an explicit time dependence. The explicit time dependence gives rise to
memory effects in these NEQ macrostates that lie outside the nonequilibrium
state space $\mathfrak{S}_{\mathbf{Z}}$. A NIEQ macrostate in $\mathfrak{S}%
_{\mathbf{Z}}$ becomes an IEQ\ macrostate in a larger state space
$\mathfrak{S}_{\mathbf{Z}^{\prime}},\mathbf{Z}^{\prime}\supset\mathbf{Z}$,
with a proper choice of $\mathbf{Z}^{\prime}$.
\end{enumerate}

\begin{definition}
\label{Def-ArbitrayState}An \emph{arbitrary macrostate} \emph{(ARB) }of a
system\emph{ }refers to \emph{all possible thermodynamic states}, which
include EQ macrostates, and NEQ macrostates with and without the memory of the
initial macrostate. We denote an arbitrary macrostate by $\mathcal{M}%
_{\text{arb}}$, NEQ macrostates by $\mathcal{M}$, EQ macrostates by
$\mathcal{M}_{\text{eq}}$, and IEQ macrostates by $\mathcal{M}_{\text{ieq}}$.
\end{definition}

\begin{definition}
\label{Def-ThermodynamicEntropy}Thermodynamic entropy $S$ is defined by the
Gibbs fundamental relation for a macrostate.
\end{definition}

\begin{definition}
\label{Def-StatisticalEntropy}Statistical entropy $S$ for $\mathcal{M}%
_{\text{arb}}$ is defined by its microstates by Gibbs formulation.
\end{definition}

\begin{definition}
\label{Def-Finite-Infinitesimal-EQ} Changes in quantities such as
$S,E,V,\cdots$ in an infinitesimal processes $\delta\mathcal{P}$ are denoted
by $dS,dE,dV,\cdots$; changes during a finite process $\mathcal{P}$ are
denoted by $\Delta S,\Delta E,\Delta V,\cdots$.
\end{definition}

\begin{definition}
\label{Def-Path-trajectories}The path $\gamma_{\mathcal{P}}$ of a macrostate
$\mathcal{M}$ is the path it takes in $\mathfrak{S}$\ during a process
$\mathcal{P}$. The trajectory $\gamma_{k}$ is the trajectory a microstate
$\mathfrak{m}_{k}$ takes in time in $\mathfrak{S}$ during the process
$\mathcal{P}$.
\end{definition}

As $\mathfrak{m}_{k}$ evolves due to Hamilton's equations of motion for given
$\mathbf{W}$, the variation of $\mathbf{x}_{k}$ has no effect on $E_{k}$.
Therefore, we will no longer exhibit $\mathbf{x}$ and simply use
$\mathcal{H}(\mathbf{W})$ for the Hamiltonian. The microenergy $E_{k}$ changes
isentropically as $\mathbf{W}$ changes without changing $p_{k}$
\cite{Gujrati-GeneralizedWork}. Accordingly, the generalized macrowork $dW$
does not generate any stochasticity. The latter is brought about by the
generalized macroheat $dQ$, which changes $p_{k}$ but without changing $E_{k}%
$. In the MNEQT,%
\begin{subequations}
\begin{equation}
dQ\equiv TdS \label{dQ-dS}%
\end{equation}
in terms of the temperature
\begin{equation}
T=\partial E/\partial S \label{temperature}%
\end{equation}
and $dS$\ of $\Sigma$. The Eq. (\ref{dQ-dS}) is a general result in the MNEQT.

It is convenient to introduce $\boldsymbol{\varphi}=(S,\mathbf{Z})$ as the set
of all thermodynamic macrovariables, which takes the microvalue
$\boldsymbol{\varphi}_{k}$ on $\mathfrak{m}_{k}$.
\end{subequations}
\begin{definition}
\label{Def-Macropartition}Macropartition: As suggested in Fig.
\ref{Fig_System} and Notation \ref{Notation}, the change
\begin{subequations}
\begin{equation}
d\boldsymbol{\varphi}\mathbf{\doteq}d_{\text{e}}\boldsymbol{\varphi}%
\mathbf{+}d_{\text{i}}\boldsymbol{\varphi} \label{Macro-Partition}%
\end{equation}
in the SI-macrovariable $\boldsymbol{\varphi}$ of $\Sigma$\ consists of two
parts: the MI-change $d_{\text{e}}\boldsymbol{\varphi}$\ is the change due to
exchange with $\widetilde{\Sigma}$, and $d_{\text{i}}\boldsymbol{\varphi}$ is
the \emph{irreversible }change occurring within $\Sigma$.
\end{subequations}
\end{definition}

This is an extension of the standard partition for the entropy change
\cite{Prigogine,deGroot}
\begin{equation}
dS\mathbf{\doteq}d_{\text{e}}S\mathbf{+}d_{\text{i}}S.
\label{Entropy-Partition}%
\end{equation}
For $E$ and $V$, the partitions are%
\begin{equation}
dE\mathbf{\doteq}d_{\text{e}}E\mathbf{+}d_{\text{i}}E,dV\mathbf{\doteq
}d_{\text{e}}V\mathbf{+}d_{\text{i}}V, \label{E-V-Partition}%
\end{equation}
except that
\begin{equation}
d_{\text{i}}E\equiv0,d_{\text{i}}V\equiv0, \label{E-and-V-Partition}%
\end{equation}
for the simple reason that internal processes cannot change $E$ and $V$,
respectively. For $N$, the partition is
\[
dN\mathbf{\doteq}d_{\text{e}}N\mathbf{+}d_{\text{i}}N,
\]
with $d_{\text{i}}N$ present when there is chemical reaction. We will find the
shorthand notation
\begin{equation}
d_{\alpha}=(d,d_{\text{e}},d_{\text{i}}) \label{infinitesimal changes}%
\end{equation}
quite useful in the following for the various infinitesimal contributions.
These linear operators satisfy%
\begin{equation}
d\equiv d_{\text{e}}+d_{\text{i}}. \label{d-operator-rule}%
\end{equation}

\begin{definition}
\label{Def-EnsembleAverage}Ensemble Average: In NEQT, any thermodynamic
macroquantity $\boldsymbol{\varphi}$ is obtained by the instantaneous
\emph{ensemble average}%
\begin{equation}
\boldsymbol{\varphi}\equiv\left\langle \boldsymbol{\varphi}\right\rangle =%
{\textstyle\sum\nolimits_{k}}
p_{k}\boldsymbol{\varphi}_{k}, \label{ensemble-average}%
\end{equation}
where $\boldsymbol{\varphi}$ takes microvalues $\boldsymbol{\varphi}_{k}$ on
$\mathfrak{m}_{k}$ at that instant with probability $p_{k}$.
\end{definition}

We have used the standard convention to write $\boldsymbol{\varphi}$
for\ $\left\langle \boldsymbol{\varphi}\right\rangle $. \ For example, the
internal energy $E$ is given by%
\begin{equation}
E\equiv\left\langle E\right\rangle =%
{\textstyle\sum\nolimits_{k}}
p_{k}E_{k}, \label{E-ensemble-average}%
\end{equation}
while the statistical entropy, often called the Gibbs entropy, is given by%
\begin{equation}
S\equiv\left\langle S\right\rangle =%
{\textstyle\sum\nolimits_{k}}
p_{k}S_{k}=-%
{\textstyle\sum\nolimits_{k}}
p_{k}\ln p_{k} \label{S-Gibbs}%
\end{equation}
where the \emph{microentropy }$S_{k}$ is
\begin{equation}
S_{k}\equiv-\eta_{k}\doteq-\ln p_{k}; \label{Microentropy}%
\end{equation}
in terms of Gibbs' \emph{index of probability} $\eta_{k}\doteq\ln p_{k}$
\cite[p. 16]{Gibbs}.

\begin{definition}
\label{Def-Micropartition}Micropartition:The macropartition in Eq.
(\ref{Macro-Partition}) is extended to microvariable $\mathbf{\varphi}_{k}$:%
\begin{subequations}
\begin{equation}
d\mathbf{\varphi}_{k}\mathbf{\doteq}d_{\text{e}}\mathbf{\varphi}_{k}%
\mathbf{+}d_{\text{i}}\mathbf{\varphi}_{k}. \label{Micro-Partition}%
\end{equation}
Thus,%
\begin{equation}
dE_{k}\mathbf{\doteq}d_{\text{e}}E_{k}\mathbf{+}d_{\text{i}}E_{k}%
,dS_{k}\mathbf{\doteq}d_{\text{e}}S_{k}\mathbf{+}d_{\text{i}}S_{k}.
\label{Micro-Partition-E-S}%
\end{equation}
The micropartition also applies to $dp_{k}$:
\end{subequations}
\begin{subequations}
\begin{equation}
dp_{k}\mathbf{\doteq}d_{\text{e}}p_{k}\mathbf{+}d_{\text{i}}p_{k},
\label{dpk-Partition}%
\end{equation}
We define%
\begin{equation}
d_{\alpha}\eta_{k}\doteq\frac{d_{\alpha}p_{k}}{p_{k}}. \label{detak}%
\end{equation}

\end{subequations}
\end{definition}

In a process, $\boldsymbol{\varphi}$ undergoes infinitesimal changes
$d_{\alpha}\boldsymbol{\varphi}_{k}$ at fixed $p_{k}$, or infinitesimal
changes $d_{\alpha}p_{k}$ at fixed $\boldsymbol{\phi}_{k}$. The changes result
in two distinct ensemble averages or process quantities.

\begin{definition}
\label{Def-Macroaverage Partition}Infinitesimal macroquantities $\left\langle
d_{\alpha}\boldsymbol{\varphi}\right\rangle $ are ensemble averages
\begin{subequations}
\begin{equation}
d_{\alpha}\boldsymbol{\varphi}_{\text{m}}\equiv\left\langle d_{\alpha
}\boldsymbol{\varphi}\right\rangle =%
{\textstyle\sum\nolimits_{k}}
p_{k}d_{\alpha}\boldsymbol{\varphi}_{k} \label{ensemble-average-d-alpha-mech}%
\end{equation}
at fixed $\left\{  p_{k}\right\}  $ so they are isentropic. We identify them
as \emph{mechanical} macroquantity and write it as $d_{\alpha}%
\boldsymbol{\varphi}_{\text{m}}$. Infinitesimal macroquantities%
\begin{equation}
d_{\alpha}\varphi_{\text{s}}\doteq\left\langle \varphi_{k}d_{\alpha}%
\eta\right\rangle \label{ensemble-average-d-alpha-stoch}%
\end{equation}
that are ensemble averages involving $\left\{  d_{\alpha}p_{k}\right\}  $ are
identified as \emph{stochastic} macroquantities and written as $d_{\alpha
}\boldsymbol{\varphi}_{\text{s}}$. Together, they determine the change
$d_{\alpha}\boldsymbol{\varphi}$:%
\end{subequations}
\begin{equation}
d_{\alpha}\boldsymbol{\varphi}\equiv d_{\alpha}\left\langle
\boldsymbol{\varphi}\right\rangle \doteq d_{\alpha}\boldsymbol{\varphi
}_{\text{m}}+d_{\alpha}\boldsymbol{\varphi}_{\text{s}}.
\label{dph-im-dphi-s-identity}%
\end{equation}

\end{definition}

We must carefully distinguish $d_{\alpha}\left\langle \boldsymbol{\varphi
}\right\rangle $ and $\left\langle d_{\alpha}\boldsymbol{\varphi}\right\rangle
$. For $E$, we will use instead the following notation:%
\begin{equation}
d_{\alpha}Q=d_{\alpha}E_{\text{s}},d_{\alpha}W=-d_{\alpha}E_{\text{m}},
\label{Macro-heat-work-alpha}%
\end{equation}
from which follows%
\begin{equation}
d_{\alpha}E=d_{\alpha}Q-d_{\alpha}W. \label{FirstLaw-alpha}%
\end{equation}
Using Eq. (\ref{E-and-V-Partition}) for $d_{\text{i}}E$, we have the following
thermodynamic identity:%
\begin{equation}
d_{\text{i}}Q\equiv d_{\text{i}}W. \label{diQ-diW-EQ}%
\end{equation}
For $d_{\alpha}=d,d_{\text{e}}$, we have the following SI- and MI- formulation
of the first law:
\begin{subequations}
\label{FirstLaw}%
\begin{align}
dE  &  =dQ-dW\label{FirstLaw-SI}\\
dE  &  =d_{\text{e}}Q-d_{\text{e}}W, \label{FirstLaw-MI}%
\end{align}
where we have used the identity $dE=d_{\text{e}}E$. The top equation is also
known as the \emph{Gibbs fundamental relation}.

We can use the operator identity in Eq. (\ref{d-operator-rule}) to introduce
the following important identities following Notation \ref{Notation}%
\end{subequations}
\begin{align}
dW  &  =d_{\text{e}}W+d_{\text{i}}W,dQ=d_{\text{e}}Q+d_{\text{i}%
}Q,\label{Macro-work-heat-partition}\\
dW_{k}  &  =d_{\text{e}}W_{k}+d_{\text{i}}W_{k},dQ_{k}=d_{\text{e}}%
Q_{k}+d_{\text{i}}Q_{k}, \label{Micro-work-heat-partition}%
\end{align}
that will be very useful in the MNEQT. For an isolated system, $d_{\text{e}%
}W\equiv0,d_{\text{e}}Q\equiv0$. Note that $dW,dQ$, etc. do not represent
changes in any SI-macrovariable.

\begin{definition}
\label{Def-Heat-Work}We simply call $dQ$ and $dW$ macroheat and macrowork,
respectively, unless clarity is needed and use exchange macroheat for
$d_{\text{e}}Q$ and exchange macrowork for $d_{\text{e}}W$, irreversible
macroheat for $d_{\text{i}}Q$ and irreversible macrowork for $d_{\text{i}}W$, respectively.
\end{definition}

Manipulating $\mathbf{w}$ such as the "volume" $V$\ from the outside through
$\widetilde{\Sigma}_{\text{w}}$ requires some external "force" $\mathbf{F}%
_{\text{w}0}$, such as the external pressure $P_{0}$ to do some "exchange
macrowork" $d\widetilde{W}$ on $\Sigma$. We have $d\widetilde{\mathbf{w}%
}=d_{\text{e}}\widetilde{\mathbf{w}}=-d_{\text{e}}\mathbf{w}$, and
\begin{equation}
d\widetilde{W}\doteq\mathbf{F}_{\text{w}0}\cdot d\widetilde{\mathbf{w}%
}=-d_{\text{e}}W\doteq-\mathbf{F}_{\text{w}0}\cdot d_{\text{e}}\mathbf{w},
\label{ExchangeWork}%
\end{equation}
where $\mathbf{F}_{\text{w}0}=(P_{0},...,\mathbf{A}_{0}=0)$; see Fig.
\ref{Fig_System}. We use $\mathbf{F}_{\text{w}0}=(\mathbf{f}_{\text{w}%
0},\mathbf{A}_{0}=0)$.

In a NEQ system, the generalized force $\mathbf{F}_{\text{w}}$ in $\Sigma$
differs from $\mathbf{F}_{\text{w}0}$. The resulting macrowork done by
$\Sigma$ is
\begin{equation}
dW\doteq\mathbf{F}_{\text{w}}\cdot d\mathbf{W}. \label{GeneralizedWork-W}%
\end{equation}
This is the SI-macrowork and differs from the MI-macrowork $d\widetilde{W}$
$=-d_{\text{e}}W$. Here,
\begin{equation}
\mathbf{F}_{\text{w}}\doteq-\partial E/\partial\mathbf{W=}(P(t),...,\mathbf{A}%
(t))=(\mathbf{f}(t),\mathbf{A}(t)); \label{GeneralizedForce-W}%
\end{equation}
see Fig. \ref{Fig_System}. The SI-affinity $\mathbf{A}$ corresponding to
$\boldsymbol{\xi}$ \cite{Prigogine,Prigogine0} is nonzero, except in EQ, when
it vanishes: $\mathbf{A}_{\text{eq}}\equiv\mathbf{A}_{0}=0=0$
\cite{deGroot,Prigogine}. The " SI-macrowork" $dW_{\boldsymbol{\xi}}$ done by
$\Sigma$ as $\boldsymbol{\xi}$ varies is
\begin{equation}
dW_{\boldsymbol{\xi}}\doteq\mathbf{A\cdot}d\boldsymbol{\xi}.
\label{GeneralizedForce-Xi}%
\end{equation}
Even for an isolated NEQ system, $dW_{\boldsymbol{\xi}}$ will not vanish; it
vanishes only in EQ, since $\boldsymbol{\xi}$ does no work when $\mathbf{A}%
_{0}=0$; however, $\mathbf{F}_{\text{w}0},d\widetilde{W}$ and $d_{\text{e}}W$
are unaffected by the presence of $\boldsymbol{\xi}$.

The macroforce \emph{imbalance} is the difference%
\begin{subequations}
\begin{equation}
\Delta\mathbf{F}^{\text{w}}\doteq(\mathbf{F}_{\text{w}}-\mathbf{F}_{\text{w}%
0})=(\mathbf{f}_{\text{w}}-\mathbf{f}_{\text{w}0},\mathbf{A}).
\label{work macroforce}%
\end{equation}
In general, $\mathbf{A}$ controls the behavior of $\boldsymbol{\xi}$ in
$\mathcal{M}$ \cite{Prigogine,Prigogine0} and vanishes when EQ is reached
\cite{deGroot,Prigogine}.\ Here, we will take a more general view of
$\mathbf{A}$, and extend its definition to $\mathbf{X}$ also. In particular,
$\Delta F^{\text{h}}\doteq T_{0}-T$ also plays the role of an affinity
\cite{Gujrati-I} so we can include it with $\Delta\mathbf{F}^{\text{w}}$ to
form set of \emph{thermodynamic macroforces} or of \emph{macroforce
imbalance}:%
\begin{equation}
\Delta\mathbf{F}\doteq(T_{0}-T,\mathbf{f}_{\text{w}}-\mathbf{f}_{\text{w}%
0},\mathbf{A}). \label{thermodynamic force}%
\end{equation}
The same reasoning also shows that $\Delta\mathbf{F}$ plays the role of an activity.

\bigskip The \emph{irreversible macrowork }$d_{\text{i}}W\doteq dW-d_{\text{e}%
}W\equiv dW+d\widetilde{W}$ is given by%
\end{subequations}
\begin{equation}
d_{\text{i}}W\doteq(\mathbf{f}_{\text{w}}-\mathbf{f}_{\text{w}0})\cdot
d_{\text{e}}\mathbf{w}+\mathbf{f}_{\text{w}}\cdot d_{\text{i}}\mathbf{w}%
+\mathbf{A\cdot}d\boldsymbol{\xi}\geq0. \label{Irrev-Work}%
\end{equation}
For the sake of clarity, we will take $V$ as a symbolic representation of
$\mathbf{X}$, and a single $\xi$ as an internal variable in many examples.
Then $\mathbf{W}=(V,\xi)$ is the macrowork parameter. In this case, we have
\begin{subequations}
\label{SI-MI-Work}%
\begin{align}
dW  &  =PdV+Ad\xi,d_{\text{e}}W=P_{0}dV,\label{SI-Work}\\
d_{\text{i}}W  &  =(P-P_{0})dV+Ad\xi, \label{MI-Work}%
\end{align}
provided $d_{\text{i}}V=0$.

The microanalogue of $\Delta\mathbf{F}^{\text{w}}$ is the internal microforce
imbalance%
\end{subequations}
\begin{equation}
\Delta\mathbf{F}_{k}^{\text{w}}\doteq(\mathbf{f}_{\text{w}k}-\mathbf{f}%
_{\text{w}0},\mathbf{A}_{k}), \label{thermodynamic microforce}%
\end{equation}
which determines the internal microwork%
\begin{equation}
d_{\text{i}}W_{k}\doteq(\mathbf{f}_{\text{w}k}-\mathbf{f}_{\text{w}0})\cdot
d_{\text{e}}\mathbf{w}+\mathbf{f}_{\text{w}k}\cdot d_{\text{i}}\mathbf{w}%
+\mathbf{A}_{k}\mathbf{\cdot}d\boldsymbol{\xi}, \label{InternalMicrowork}%
\end{equation}
as the exchange microwork is%
\begin{equation}
d_{\text{e}}W_{k}\doteq\mathbf{f}_{\text{w}0}\cdot d_{\text{e}}\mathbf{w}%
=d_{\text{e}}W\mathbf{,\forall}k. \label{ExchangeMicrowork}%
\end{equation}

\section{Internal Variables\label{Sec-InternalVariables}}

We should emphasize that the concept of internal variables and their
usefulness in NEQT has a long history. We refer the reader to an excellent
exposition of this topic in the monograph by Maugin \cite[see Ch. 4]{Maugin}.
We consider a few simple examples to justify why internal variables are needed
to uniquely specify a $\mathcal{M}$, and how to identify them for various systems.

It should be stated that in order to capture a NEQ process, internal variables
are usually \emph{necessary}. Another way to appreciate this fact is to
realize that

\begin{remark}
\label{Remark-IsolatedSystem}For an isolated system, all the observables in
$\mathbf{X}_{0}$ are fixed so if the entropy is a function of $\mathbf{X}_{0}$
only, it \emph{cannot} change
\cite{Gujrati-I,Gujrati-II,Gujrati-Entropy1,Gujrati-Entropy2} even if the
system is out of EQ.
\end{remark}

Thus, we need additional independent variables to ensure the law of increase
of entropy for a NEQ isolated system. A point in $\mathfrak{S}_{\mathbf{X}}$
represents $\mathcal{M}_{\text{eq}}$, but a point $\mathfrak{S}_{\mathbf{Z}}$
represents $\mathcal{M}$. In EQ, internal variables are no longer independent
of the observables. Consequently, their affinities (see later) vanish in EQ.
It is common to define the internal variables so their EQ values vanish. We
now discuss various scenarios where they are needed for a proper consideration.

\subsection{A Two-level System}

Consider a NEQ system of $N$ particles such as Ising spins, each of which can
be in two levels, forming an isolated system $\Sigma_{0}$ of volume $V$. Let
$\rho_{l}$ and $e_{l}(V),l=1,2$ denote the probabilities and energies of the
two levels of a particle in a NEQ macrostate so that $\rho_{1},\rho_{2}$ keep
changing. We have assumed that $e_{l}(V)$ depends on the observable $V$ only,
which happens to be constant for $\Sigma_{0}$. We have $e=\rho_{1}%
e_{1}(V)+\rho_{2}e_{2}(V)$ for the average energy per particle, which is also
a constant for $\Sigma_{0}$, and
\[
d\rho_{1}+d\rho_{2}=0
\]
as a consequence of $\rho_{1}+\rho_{2}=1$. Using $de=0$, we get
\[
d\rho_{1}+d\rho_{2}e_{2}/e_{1}=0,
\]
which, for $e_{1}\neq e_{2}$, is inconsistent with the first equation (unless
$d\rho_{1}=0=d\rho_{2}$, which corresponds to EQ). Thus, $e_{l}(V)$ cannot be
treated as constant in evaluating $de$. In other words, there must be an extra
dependence in $e_{l}$ so that%
\[
e_{1}d\rho_{1}+d\rho_{2}e_{2}+\rho_{1}de_{1}+\rho_{2}de_{2}=0,
\]
and the inconsistency is removed. This extra dependence must be due to
\emph{independent} internal variables that are not controlled from the outside
(isolated system) so they continue to relax in $\Sigma_{0}$ as it approaches
EQ. Let us imagine that there is a single internal variable $\xi$ so that we
can express $e_{l}$ as $e_{l}(V,\xi)$ in which $\xi$ continues to change as
the system comes to equilibrium. The above equation then relates $d\rho_{1}$
and $d\xi$; they both vanish simultaneously as EQ is reached. We also see that
without any $\xi$, the isolated system cannot equilibrate; see Remark
\ref{Remark-IsolatedSystem}.

\subsection{A Many-level System}

The above discussion is easily extended to a $\Sigma$ with many energy levels
of a particle with the same conclusion that at least a single internal
variable is required to express $e_{l}=e_{l}(V,\xi)$ for each level $l$. We
can also visualize the above system in terms of microstates. A microstate
$\mathfrak{m}_{k}$ refers to a particular distribution of the $N$ particles in
any of the levels with energy $E_{k}=%
{\textstyle\sum\nolimits_{l}}
N_{l}e_{l}$, where $N_{l}$ is the number of particles in the $l$th level, and
is obviously a function of $N,V,\xi$ so we will express it as $E_{k}(V,\xi)$;
we suppress the dependence on $N$. This makes the average energy of the system
also a function of $V,\xi$, which we express as $E(V,\xi)$.

\subsection{Disparate Degrees of freedom}

In classical statistical mechanics, the kinetic and potential energies $K$ and
$U$, respectively, are functions of independent variables. Only their sum
$K+U=E$ can be controlled from the outside, but not individually. Thus, one of
them can be treated as an internal variable. In a NEQ macrostates, each term
can have its own temperature. Only in EQ, do they have the same temperature.

This has an important consequence for glasses, where the vibrational degrees
of freedom (dof$_{\text{v}}$) come to EQ with the heat bath at $T_{0}$ faster
than the configurational degrees of freedom (dof$_{\text{c}}$), which have a
different temperature than $T_{0}$. The disparity in dof$_{\text{v}}$ and
dof$_{\text{c}}$ cannot be controlled by the observer so it plays the role of
an internal variable. A well-known equation, the Tool-Narayanaswamy equation
is concerned with this disparity and is discussed in Sect.
\ref{Sec-Tool-Narayan}.

Consider a collection of semiflexible polymers in a solution on a lattice. The
interaction energy $E$ consists of several additive terms as discussed in
\cite[Eq. (40)]{Gujrati-II}: the interaction energy $E_{\text{ps}}$\ between
the polymer and the solvent, the interaction energy $E_{\text{ss}}$\ between
the solvent, the interaction energy $E_{\text{pp}}$\ between polymers. Only
the total $E$ can be controlled from the outside so the remaining terms
determine several internal variables.

In the examples above, the internal variables are not due to spatial
inhomogeneity. An EQ system is uniform. Thus, the presence of $\xi$ suggests
some sort of nonuniformity in the system. To appreciate its physics, we
consider a slightly different situation below as a possible example of
nonuniformity.%
\begin{figure}
[ptb]
\begin{center}
\includegraphics[
trim=0.798533in 0.000000in 0.799189in 0.499309in,
height=1.3033in,
width=3.0035in
]%
{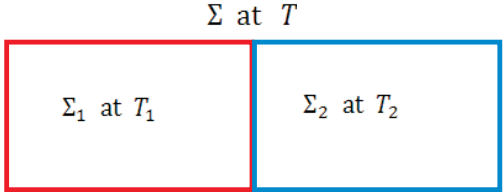}%
\caption{A composite system $\Sigma$ consisting of two identical subsystems
$\Sigma_{1}$ at temperature $T_{1}$ and $\Sigma_{2}$ at temperature $T_{2}$.
It will be seen later in Sec. \ref{Sec-Composite System} that the
thermodynamic temperature of $\Sigma$ can be defined as $T$ given by Eq.
(\ref{beta-A-Composite}). The irreversibility in $\Sigma$ requires one
internal variable $\xi$ given in Eq. (\ref{Internal Variable-1}).}%
\label{Fig-CompositeSystem}%
\end{center}
\end{figure}

\subsection{Nonuniformity \label{Sec-Nonuniformity}}

(a) We consider as a simple NEQ example a composite isolated system $\Sigma$,
see Fig. \ref{Fig-CompositeSystem}, consisting of two subsystems $\Sigma_{1}$
and $\Sigma_{2}$ of identical volumes and numbers of particles but at
different temperatures $T_{1}$ and $T_{2}$ at any time $t<\tau_{\text{eq}}$
before EQ is reached at $t=\tau_{\text{eq}}$ so the subsystems have different
time-dependent energies $E_{1}$ and $E_{2}$, respectively. We assume a
diathermal wall separating $\Sigma_{1}$ and $\Sigma_{2}$. Treating each
subsystem in EQ at each $t$, we write their entropies as $S_{1}(E_{1}%
,V/2,N/2)$ and $S_{2}(E_{2},V/2,N/2)$, which we simply show as $S_{1}(E_{1})$
and $S_{2}(E_{2})$ as we will not let their volumes and particles numbers
change. The entropy $S\doteq S_{1}(E_{1})+S_{2}(E_{2})$ of $\Sigma$ is a
function of $E_{1}$ and $E_{2}$. Obviously, $\Sigma$ is in a NEQ macrostate at
each $t<\tau_{\text{eq}}$. As $E_{1}$ and $E_{2}$ do not refer to $\Sigma$, we
form two independent combinations from $E_{1}$ and $E_{2}$%

\begin{equation}
E=E_{1}+E_{2},\xi=E_{1}-E_{2}, \label{Internal Variable-1}%
\end{equation}
that refer to $\Sigma$ so that we can express the entropy as $S(E,\xi)$ for
$\Sigma$ treated as a blackbox $\Sigma_{\text{B}}$; we do not need to know
about its interior (its inhomogeneity) anymore. Here, $\xi$ plays the role of
an internal variable, which continues to relax\ towards zero as $\Sigma$
approaches EQ. For given $E$ and $\xi$,\ $S(E,\xi)$ has the \ maximum possible
values since both $S_{1}$ and $S_{2}$ have their maximum value. As we will see
below, this is the idea behind the concept of \emph{internal equilibrium} in
which $S(E,\xi)$ is a state function of state variables and continues to
increase as $\xi$ decreases and vanishes in EQ. In this macrostate,
$S(E,\xi=0)$ has the maximum possible value for fixed $E$ so it becomes a
state function; see Definition \ref{Def-InternalEQ}. This case and its various
extensions are investigated in MNEQT in Sec. \ref{Sec-Composite System}.

(b) We can easily extend the model to include four identical subsystems of
fixed and identical volumes and numbers of particles, but of different
energies $E_{1},E_{2},E_{3}$, and $E_{4}$. Instead of using these $4$
independent variables, we can use the following four independent combinations%
\begin{align}
E  &  =E_{1}+E_{2}+E_{3}+E_{4}=\text{constant},\nonumber\\
\xi &  =E_{1}+E_{2}-E_{3}-E_{4},\nonumber\\
\text{ }\xi^{\prime}  &  =E_{1}-E_{2}+E_{3}-E_{4},\nonumber\\
\xi^{\prime\prime}  &  =E_{1}-E_{2}-E_{3}+E_{4}, \label{Internal Variable-2}%
\end{align}
to express the entropy of $\Sigma$ as $S(E,\xi,\xi^{\prime},\xi^{\prime\prime
})$. The pattern of extension for this simple case of energy inhomogeneity. is evident.

(c) We make the model a bit more interesting by allowing the volumes $V_{1}$
and $V_{2}$ to also vary as $\Sigma$ equilibrates. Apart from the internal
variable $\xi$, we require another internal variable $\xi^{\prime}$ to form
two independent combinations
\begin{equation}
V=V_{1}+V_{2}=\text{constant},\text{ }\xi^{\prime}=V_{1}-V_{2}
\label{Internal Variable-3}%
\end{equation}
so that we can use $S(E,V,\xi,\xi^{\prime})\doteq S_{1\text{eq}}(E_{1}%
,V_{1})+S_{2\text{eq}}(E_{2},V_{2})$ for the entropy of $\Sigma$ in terms of
the entropies of $\Sigma_{1}$ and $\Sigma_{2}$.

(d) In the above examples, we have assumed the subsystems to be in EQ. We now
consider when the subsystems are in IEQ. We consider the simple case of
two\ subsystems $\Sigma_{1}$ and $\Sigma_{2}$ of identical volumes and numbers
of particles. Each subsystem is in different IEQ macrostates described by
$E_{1},\xi_{1}$ and $E_{2},\xi_{2}$. We now construct four independent
combinations%
\begin{align}
E  &  =E_{1}+E_{2}=\text{constant},\xi=E_{1}-E_{2},\nonumber\\
\xi^{\prime}  &  =\xi_{1}+\xi_{2},\xi^{\prime\prime}=\xi_{1}-\xi_{2},
\label{Internal Variable-4}%
\end{align}
which can be used to express the entropy of $\Sigma$ as $S(E,\xi,\xi^{\prime
},\xi^{\prime\prime})$.

(e) The example in (a) can be easily extended to the case of expansion and
contraction by replacing $E,E_{1}$, and $E_{2}$ by $N,N_{\text{L}}$, and
$N_{\text{R}}$, see Fig. \ref{Fig_Expansion}, to describe the diffusion of
particles \cite{Vilar-Rubi}. The role of $\beta$ and $E$, etc. are played by
$\beta\mu$ and $N$, etc.

\subsection{Relative Motion in Piston-Gas System \label{Sec-Piston-Spring}}

We now consider the motion of the piston in Fig. \ref{Fig_Piston-Spring}(a)
because of the pressure difference across it. The discussion also shows how
the Hamiltonian becomes dependent on internal variables, and how the system is
maintained \emph{stationary} despite motion of its parts.

Let $\mathbf{P}_{\text{p}}$ denote the momentum of the piston. The gas, the
cylinder and the piston constitute the system $\Sigma$. We have a gas of mass
$M_{\text{g}}$ in the cylindrical volume $V_{\text{g}}$, the piston of mass
$M_{\text{p}}$, and the rigid cylinder (with its end opposite to the piston
closed) of mass $M_{\text{c}}$. However, we will consider the composite
subsystem $\Sigma_{\text{gc}}=\Sigma_{\text{g}}\cup\Sigma_{\text{c}}$ so that
with $\Sigma_{\text{p}}$ it makes up $\Sigma$. The Hamiltonian $\mathcal{H}$
of the system is the sum of $\mathcal{H}_{\text{gc}}$ of the gas and cylinder,
$\mathcal{H}_{\text{p}}$ of the piston, the interaction Hamiltonian
$\mathcal{H}_{\text{int}}$ between the two subsystems $\Sigma_{\text{gc}}%
\ $and $\Sigma_{\text{p}}$, and the interaction Hamiltonian $\mathcal{H}%
_{\text{sm}}$ between $\Sigma$ and $\widetilde{\Sigma}$. As is customary, see
the discussion in Sect. \ref{Sec-Notation}, we will neglect $\mathcal{H}%
_{\text{sm}}$ here. We assume that the centers-of-mass of $\Sigma_{\text{gc}}$
and $\Sigma_{\text{p}}$ are moving with respect to the medium with linear
momentum $\mathbf{P}_{\text{gc}}$ and $\mathbf{P}_{\text{p}}$, respectively.
We do not allow any rotation for simplicity. We assume that
\begin{equation}
\mathbf{P}_{\text{gc}}+\mathbf{P}_{\text{p}}=0,
\label{Stationary_Momentum_Condition}%
\end{equation}
so that $\Sigma$ is at rest with respect to the medium. Thus,\
\begin{equation}
\mathcal{H}(\left.  \mathbf{x}\right\vert V,\mathbf{P}_{\text{gc}}%
,\mathbf{P}_{\text{p}})=%
{\textstyle\sum\nolimits_{\lambda}}
\mathcal{H}_{\lambda}(\left.  \mathbf{x}_{\lambda}\right\vert V_{\lambda
},\mathbf{P}_{\lambda})+\mathcal{H}_{\text{int}}, \label{TotalHamiltonian}%
\end{equation}
where $\lambda=$gc,p, $\mathbf{x}_{\lambda}\mathbf{=(r}_{\lambda}%
\mathbf{,p}_{\lambda}\mathbf{)}$ denotes a point in the phase space
$\Gamma_{\lambda}$ of $\Sigma_{\lambda}$; $V_{\lambda\text{ }}$is the volume
of $\Sigma_{\mathbf{\lambda}}$, and $V=V_{\text{gc}}+V_{\text{p}}$ is the
volume of $\Sigma$. We do not exhibit the number of particles $N_{\text{g}%
},N_{\text{c}},N_{\text{p}}$ as we keep them fixed. We let $\mathbf{x}$
denotes the collection ($\mathbf{x}_{\text{gc}},\mathbf{x}_{\text{p}}$). Thus,
$\mathcal{H}(\left.  \mathbf{x}\right\vert V,\mathbf{P}_{\text{gc}}%
,\mathbf{P}_{\text{p}})$ and the average energy $E$ depend on the parameters
$V,\mathbf{P}_{\text{gc}},\mathbf{P}_{\text{p}}$. As the relative motion
cannot be controlled from the outside, one of the momenta plays the role of an
internal variable.
\begin{figure}
[ptb]
\begin{center}
\includegraphics[
height=1.689in,
width=3.2785in
]%
{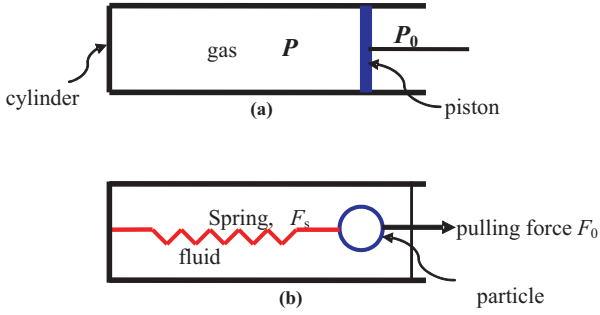}%
\caption{We schematically show a system of (a) gas in a cylinder with a
movable piston under an external pressure $P_{0\text{ }}$controlling the
volume $V$ of the gas, and (b) a particle attached to a spring in a fluid
being pulled by an external force $F_{0}$, which causes the spring to stretch
or compress depending on its direction. In an irreversible process, the
internal pressure $P$ (the spring force $F_{\text{s}}$) is different in
magnitude from the external pressure $P_{0}$ (external force $F_{0}$).}%
\label{Fig_Piston-Spring}%
\end{center}
\end{figure}

We discuss the example of the spring in Fig. \ref{Fig_Piston-Spring}(b) will
be discussed in Sect. \ref{Sec-Friction}.

\subsection{Extended State Space}

It should be clear from above that we can identify the entropy

If we divide $\Sigma$ into many subsystems $\left\{  \Sigma_{i}\right\}  $ so
that they are \emph{all} quasi-independent, then the entropy additivity gives%
\[
S(\mathbf{X}(t)\mathbf{,}t)=%
{\textstyle\sum\nolimits_{i}}
S_{i}(\mathbf{X}_{i}(t)\mathbf{,}t).
\]

As we will be dealing with the Hamiltonian of the system, it is useful to
introduce the notation in Eq. (\ref{w-W-eq}) with $\mathbf{W}=(\mathbf{w,}%
\boldsymbol{\xi})$. Then, $E$ and $E_{k}$\ become a function of $\mathbf{W}$
as we will show in Sec. \ref{Sec-HamiltonianTrajectories}. Here, $\mathbf{W}$
appears as a \emph{parameter} in the Hamiltonian, which we will write as
$\mathcal{H}(\left.  \mathbf{x}\right\vert \mathbf{W})$, where $\mathbf{x}$ is
a point (collection of coordinates and momenta of the particles) in the phase
space $\Gamma(\mathbf{W})$ specified by $\mathbf{W}$. As an example,
$V,\mathbf{P}_{\text{gc}},\mathbf{P}_{\text{p}}$ are the parameters in Sec.
\ref{Sec-Friction}. When the system moves about in the phase space
$\Gamma(\mathbf{W})$, $\mathbf{x}$ changes but $\mathbf{W}$ as a parameter
remains fixed in a state subspace $\mathfrak{S}_{\mathbf{W}}\subset
\mathfrak{S}_{\mathbf{Z}}$; see the discussion of Eq. (\ref{HamiltonianChange}).

It is important to draw attention to the following important distinction
between the Hamiltonian $\mathcal{H}$ and the ensemble average energy $E$; see
Eq. (\ref{Ensemble_Average}). While $E$ accounts for the stochasticity through
microstate probabilities, the use of the Hamiltonian is going to be restricted
to a particular microstate. In other words, the Hamiltonian depends on
$\mathbf{x}$ and $\mathbf{W}$ but the energy depends on the entropy $S$ and
$\mathbf{W}$. The energy $E_{k}$ of $\mathfrak{m}_{k}$, on the other hand,
depends only on $\mathbf{W}$ and denotes the value of $\mathcal{H}$ for
$\mathfrak{m}_{k}$. In the following, we will always treat Hamiltonians and
microstate energies as equivalent description, which does not depend on
knowing $\{p_{k}\}$; the average energies depend on $\{p_{k}\}$ for their
definition; see Eq. (\ref{E-ensemble-average}).

\section{NEQ\ Entropy \label{Sec-NEQ-S}}

\subsection{Determination of $S$}

The uniqueness issue about the NEQ macrostate says nothing about the entropy
of an arbitrary (so it may be nonunique) macrostate $\mathcal{M}:\left\{
\mathfrak{m}_{k},p_{k}\right\}  $, which is \emph{always} given by the Gibbs
entropy in Eq. (\ref{S-Gibbs}); see also \cite{Shannon}. The ensemble
averaging implies that the entropy is a statistical concept, as is the energy
$E=\left\langle E\right\rangle $, Eq. (\ref{E-ensemble-average}).

We now justify the Gibbs' statistical formulation of $S$ for any arbitrary
$\mathcal{M}$ in thermodynamics. The demonstration follows a very simple
combinatorial argument \cite{Gujrati-Entropy2} using Boltzmann concept of
thermodynamic entropy. In the demonstration, $\mathcal{M}$ is not required to
be uniquely identified. This entropy satisfies the \emph{law of increase of
entropy} as is easily seen by the discussion by Landau and Lifshitz
\cite{Landau} for a NEQ ideal gas \cite{Note-Landau} in $\mathfrak{S}%
_{\mathbf{X}}$ to derive the equilibrium distribution. Thus, the form in Eq.
(\ref{S-Gibbs}) is not restricted to only uniquely identified $\mathcal{M}$'s.
Hopefully, this will become clear below.

\begin{proposition}
\label{Prop-SecondLaw}\textbf{The Second Law} The NEQ Gibbs entropy
$S_{0}(\mathbf{X}_{0}\mathbf{,}t)$ of an isolated system $\Sigma_{0}$ is
bounded above by its equilibrium entropy $S_{0}(\mathbf{X}_{0})$ and
continuously increases towards it so that \cite{Landau}
\begin{equation}
dS_{0}(\mathbf{X}_{0}\mathbf{,}t)/dt\geq0. \label{SecondLaw0}%
\end{equation}

\end{proposition}

\subsection{General Formulation of the Statistical
Entropy\label{Marker_NonEq-S}\ \ \ }

We focus on a macrostate $\mathcal{M}(t)\ $of some body $\Sigma$\ at a given
instant $t$, which refers to the set $\mathbf{m}=\left\{  \mathfrak{m}%
_{k}\right\}  $\ of microstates and their probabilities $\mathbf{p}=\left\{
p_{k}\right\}  $. The microstates are specified by $(E_{k}(t),\mathbf{W}(t))$,
and may not uniquely specify the macrostate $\mathcal{M}(t)$. Thus, even the
set $\mathbf{m}$ need not be uniquely specified. In the following, we will use
the set $\mathbf{Z}(t)=(E(t),\mathbf{W}(t))$ for the set $\mathbf{m}$ for
simplicity. We will also denote $\mathbf{Z}(t)$ by $\overline{\mathbf{Z}}$ so
that we can separate out the explicit variation due to $t$. For simplicity, we
suppress $t$ in $\mathcal{M}$ in the following. For the computation of
combinatorics, the probabilities are handled in the following abstract way. We
consider a large number $\mathcal{N=C}W(\overline{\mathbf{Z}})$ of independent
\emph{replicas} or \emph{samples} of $\Sigma$, with $\mathcal{C}$\ some large
integer constant and $W(\overline{\mathbf{Z}})$\ the number of distinct
microstates $\mathfrak{m}_{k}$. We will see that $W(\overline{\mathbf{Z}}%
)$\ is determined by $\mathfrak{m}_{k}$'s having nonzero probabilities. We
will call them \emph{available} microstates. The samples should be thought of
as identically prepared experimental samples \cite{Gujrati-Symmetry}.

Let $\Gamma(\overline{\mathbf{Z}})$ denote the sample space spanned by
$\left\{  \mathfrak{m}_{k}\right\}  $, and let $\mathcal{N}_{k}(t)$ denote the
number of $k$th samples (samples specified by $\mathfrak{m}_{k}$) so that%
\begin{equation}
0\leq p_{k}(t)=\mathcal{N}_{k}(t)/\mathcal{N}\leq1;\ \
{\textstyle\sum\limits_{k=1}^{W(\overline{\mathbf{Z}})}}
\mathcal{N}_{k}(t)=\mathcal{N}. \label{sample_probability}%
\end{equation}
The above sample space is a generalization of the \emph{ensemble} introduced
by Gibbs, except that the latter is restricted to an equilibrium body, whereas
$\Gamma(\overline{\mathbf{Z}})$ refers to the body in any arbitrary macrostate
so that $p_{k}$ may be time-dependent, and need not be unique. The\emph{
ensemble average} of some quantity $\mathcal{Z}$\ over these samples is given
by Eq. (\ref{ensemble-average}). Thus,%
\begin{equation}
\left\langle \mathcal{Z}\right\rangle \equiv%
{\textstyle\sum\limits_{k=1}^{W(\overline{\mathbf{Z}})}}
p_{k}(t)\mathcal{Z}_{k},\ \
{\textstyle\sum\limits_{k=1}^{W(\overline{\mathbf{Z}})}}
p_{k}(t)\equiv1, \label{Ensemble_Average}%
\end{equation}
where $\mathcal{Z}_{k}$ is the value of $\mathcal{Z}$ in $\mathfrak{m}_{k}$.

The samples are, by definition, \emph{independent} of each other so that there
are no correlations among them. Because of this, we can treat the samples
$\left\{  \mathfrak{m}_{k}\right\}  $ to be the outcomes of some random
variable, the macrostate $\mathcal{M}(t)$. This independence property of the
outcomes is crucial in the following. They may be equiprobable but not
necessarily. The number of ways $\mathcal{W}$ to arrange the $\mathcal{N}$
samples into $W(\overline{\mathbf{Z}})$ distinct microstates is%
\begin{equation}
\mathcal{W\equiv N}!/%
{\textstyle\prod\limits_{k}}
\mathcal{N}_{k}(t)!. \label{Combinations}%
\end{equation}
Taking its natural log, as proposed by Boltzmann, to obtain an \emph{additive}
quantity per sample as
\begin{equation}
\mathcal{S}\equiv\ln\mathcal{W}/\mathcal{N},
\label{Ensemble_entropy_Formulation}%
\end{equation}
and using Stirling's approximation, we see easily that it can be written as
the average of the negative of Gibbs' index of probability\emph{:}
\begin{equation}
\mathcal{S}(\overline{\mathbf{Z}},t)\equiv-\left\langle \eta(t)\right\rangle
\equiv-%
{\textstyle\sum\limits_{k=1}^{W(\overline{\mathbf{Z}})}}
p_{k}(t)\ln p_{k}(t), \label{Gibbs_Formulation}%
\end{equation}
where we have also shown an explicit time-dependence, which merely reflects
the fact that it is not a state function in $\mathfrak{S}_{\overline
{\mathbf{Z}}}$, a reflection of the fact that $\mathcal{M}$ is not uniquely
specified in $\mathfrak{S}_{\overline{\mathbf{Z}}}$. We have put $t$ back
above for clarity. Thus, Eq. (\ref{Gibbs_Formulation}) is nothing but Eq.
(\ref{S-Gibbs}) in form, and thus justifies it for an arbitrary $\mathcal{M}$.

The above derivation is based on fundamental principles and does not require
the body to be in equilibrium; therefore, it is always applicable for any
arbitrary macrostate $\mathcal{M}(t)$. To the best of our knowledge, even
though such an expression has been extensively used in the literature for
NEQ\ entropy, it has been used by simply appealing to the information entropy
\cite{Shannon}.

The distinction between the Gibbs' statistical entropy $\mathcal{S}$ and the
thermodynamic entropy $S$ should be emphasized. The latter appears in the
Gibbs fundamental relation that relates the energy change $dE$ with the
entropy change $dS$ as is well known in classical thermodynamics, and as we
will also demonstrate below; see also Eq. (\ref{FirstLaw-SI}). The concept of
microstates is irrelevant for this, which is a purely thermodynamic relation.
On the other hand, $\mathcal{S}$ is solely determined by $\left\{
\mathfrak{m}_{k}\right\}  $ so its a statistical quantity. It then becomes
imperative to show their equivalence, mainly because $\mathcal{S}$ is based on
the Boltzmann idea. This equivalence has been justified elsewhere
\cite{Gujrati-Entropy1,Gujrati-Entropy2}, and will be briefly summarized below.

\begin{remark}
Because of this equivalence, we will no longer make any distinction between
the statistical Gibbs entropy and the thermodynamic entropy and will use the
standard notation $S$\ for both of them.
\end{remark}

\begin{remark}
The Gibbs entropy $\mathcal{S}$ appears as an instantaneous ensemble average,
see Definition \ref{Def-EnsembleAverage}. This average should be contrasted
with a \emph{temporal average} in which a macroquantity ${\varphi}$ is
considered as the average over a long period $\tau_{0}$ of time%
\[
{\varphi}=\frac{1}{\tau_{0}}%
{\textstyle\int\nolimits_{0}^{\tau_{0}}}
{\varphi}(t)dt,
\]
where ${\varphi}(t)$ is the value of ${\varphi}$ at time $t$ \cite{Landau}.
For an EQ macrostate $\mathcal{M}_{\text{eq}}$, both definitions give the same
result provided ergodicity holds. The physics of this average is that
$\mathbf{\varphi}(t)$ at $t$ represents a microstate of $\mathcal{M}%
_{\text{eq}}$. As $\mathcal{M}_{\text{eq}}$\ is invariant in time, these
microstates belong to $\mathcal{M}_{\text{eq}}$, and the time average is the
same as the ensemble average if ergodicity holds. However, for a NEQ
macrostate $\mathcal{M}(t)$, which continuously changes with time, the
temporal average is not physically meaningful as the microstate at time $t$
corresponds to $\mathcal{M}(t)$ and not to $\mathcal{M}(t=0)$ in that the
probabilities and $\mathbf{Z}$ are different in the two macrostates. Only the
ensemble average makes any sense at any time $t$ as was first pointed out in
\cite{Gujrati-ResidualEntropy}. Because of this, we only consider ensemble
averages in this review.
\end{remark}

The maximum possible value of $S(t)$ for given $\overline{\mathbf{Z}}%
\in\mathfrak{S}_{\overline{\mathbf{Z}}}$ occurs when $\mathfrak{m}_{k}$ are
\emph{uniquely }specified in $\mathfrak{S}_{\overline{\mathbf{Z}}}$. This
makes $S(t)$ a state function of $\overline{\mathbf{Z}}(t)$ with no explicit
time dependence, which we write as $S(\overline{\mathbf{Z}})$. Thus,
\begin{equation}
\left.  S_{\text{max}}(\overline{\mathbf{Z}},t)\right\vert _{\overline
{\mathbf{Z}}\text{ fixed}}=S(\overline{\mathbf{Z}}). \label{S_Boltzmann0}%
\end{equation}
The simplest way to understand the physical meaning is as follows: Consider
$\overline{\mathbf{Z}}\in\mathfrak{S}_{\mathbf{Z}}$ at some time $t$. As
$S(t)$ may not be a unique function of $\overline{\mathbf{Z}}$, we look at all
possible entropy functions for this $\overline{\mathbf{Z}}$. These entropies
correspond to all possible sets of $\left\{  p_{k}(t)\right\}  $ for a fixed
$\overline{\mathbf{Z}}$, and define different possible macrostates $\left\{
\mathcal{M}\right\}  $. We pick that particular $\overline{\mathcal{M}}%
\in\left\{  \mathcal{M}\right\}  $ among these that has the \emph{maximum
possible value} of the entropy, which we denote by $S(\overline{\mathbf{Z}})$
or $S(\mathbf{Z}(t))$ without any explicit $t$-dependence. This entropy is a
\emph{state function} $S(\overline{\mathbf{Z}})$. For a macroscopic system,
this occurs when the corresponding microstate probabilities for $\overline
{\mathcal{M}}$ are
\begin{subequations}
\begin{equation}
\overline{p}_{k}(t)=1/W(\overline{\mathbf{Z}})>0,\text{~~\ \ }\forall
\overline{m}_{k}\in\Gamma(\overline{\mathbf{Z}}), \label{Equiprobable-Z}%
\end{equation}
so that
\begin{equation}
S(\overline{\mathbf{Z}})=\ln W(\overline{\mathbf{Z}}). \label{S_BoltzmannZ}%
\end{equation}
We wish to point out the presence of nonzero probabilities in Eq.
(\ref{Equiprobable-Z}) that explains the comment above of available
microstates. Including microstates with zero probabilities will not correcting
account for the number of microstates with given $\overline{\mathbf{Z}}$.

There is an alternative to the above picture in which we can imagine the
$\Sigma$ for which $\overline{\mathbf{Z}}$ has been fixed, which essentially
"isolates" $\Sigma$ and converts it into a $\Sigma_{0}$. Then, as $t$ varies,
its entropy increases until it reaches its maximum value $S(\overline
{\mathbf{Z}})$ in accordance with Proposition \ref{Prop-SecondLaw}.
\end{subequations}
\begin{remark}
\label{Remark-FlatDistribution}We emphasize that $\overline{\mathbf{Z}%
}=(E,\mathbf{W})$ so $p_{k}$ above in Eq. (\ref{Equiprobable-Z}) is determined
by the average energy $E$ and not by the microstate energy $E_{k}$ as derived
later in Sect. (\ref{Sec-MicrostateProbabilities}). The $p_{k}$ in Eq.
(\ref{Equiprobable-Z}) basically replaces the actual probability distribution
in Eq. (\ref{microstate probability}) by a \emph{flat distribution} of height
$1/W(\overline{\mathbf{Z}})$ and width $W(\overline{\mathbf{Z}})$, a common
practice in the thermodynamic limit of statistical mechanics \cite{Landau}.
Despite this modification, the entropy has the same value for a macroscopic
system, for which $\beta$ and $\mathbf{F}_{\text{w}}$ are given by Eqs.
(\ref{T-beta}) and (\ref{Fw}), respectively; see also Sect.
\ref{Sec-MicrostateProbabilities}.
\end{remark}

Let us consider a different formulation of the entropy for a macrostate
$\mathcal{\bar{M}}(t)\in$ $\mathfrak{S}_{\overline{\mathbf{X}}}$\ specified by
some $\overline{\mathbf{X}}=\mathbf{X}(t)\mathbf{\subset Z}$ at some instance
$t$. This macrostate provides a more incomplete specification than in
$\mathfrak{S}_{\overline{\mathbf{Z}}}$. Applying the above formulation to
$\mathcal{\bar{M}}\in$ $\mathfrak{S}_{\overline{\mathbf{X}}}$, and consisting
of microstates $\left\{  \overline{m}_{k}\right\}  ,$ forming the set
$\overline{\mathbf{m}}\equiv\mathbf{m}(\overline{\mathbf{X}}),$ with
probabilities $\left\{  \overline{p}_{k}(t)\right\}  $, we find that
\begin{equation}
S(\overline{\mathbf{X}},t)\equiv-%
{\textstyle\sum\limits_{k=1}^{W(\overline{\mathbf{X}})}}
\overline{p}_{k}(t)\ln\overline{p}_{k}(t),%
{\textstyle\sum\limits_{k=1}^{W(\overline{\mathbf{X}})}}
\overline{p}_{k}(t)\equiv1, \label{Conventional_Entropy0}%
\end{equation}
is the entropy of $\mathcal{\bar{M}}$; here $W(\overline{\mathbf{X}})$ is the
number of distinct microstates $\overline{m}_{k}$. It should be obvious that%
\[
W(\overline{\mathbf{X}})\equiv%
{\textstyle\sum\nolimits_{\boldsymbol{\xi}(t)}}
W(\overline{\mathbf{Z}}).
\]
Again, under the equiprobable assumption
\[
\overline{p}_{k}(t)\rightarrow\overline{p}_{k\text{,eq}}=1/W(\overline
{\mathbf{X}}),\text{~~\ \ }\forall\overline{m}_{k}\in\Gamma(\overline
{\mathbf{X}}),
\]
$\Gamma(\overline{\mathbf{X}})$\ denoting the sample space spanned by
$\overline{\mathbf{m}}=\left\{  \overline{m}_{k}\right\}  $, the above entropy
takes its maximum possible value%
\begin{equation}
S_{\text{max}}(\overline{\mathbf{X}},t)=S(\overline{\mathbf{X}})=\ln
W(\overline{\mathbf{X}}), \label{S_Boltzmann}%
\end{equation}
which is the well-known value of the Boltzmann entropy for a body in
equilibrium%
\begin{equation}
S(\overline{\mathbf{X}})=\ln W(\overline{\mathbf{X}}), \label{Boltzmann_S}%
\end{equation}
and provides a statistical definition of, and hence connects it with the,
thermodynamic entropy of the body proposed by Boltzmann. The maximization
again has the same implication as in Eq. (\ref{S_Boltzmann0}):\ For given
$\overline{\mathbf{X}}$, we look for the maximum entropy at all possible
times. It is evident that%
\begin{equation}
S(\overline{\mathbf{Z}},t)\leq S(\overline{\mathbf{Z}})\leq S(\overline
{\mathbf{X}}). \label{Entropy_Inequalities0}%
\end{equation}
Thus, the NEQ entropy $S(\overline{\mathbf{Z}},t)$ as $t\rightarrow
\tau_{\text{eq}}$, the equilibration time, reduces to $S(\overline{\mathbf{X}%
})$ in EQ, as expected. Before equilibration, $S(\overline{\mathbf{Z}})$ in
$\mathfrak{S}_{\overline{\mathbf{Z}}}$ remains a nonstate function
$S(\overline{\mathbf{X}},t)$ in $\mathfrak{S}_{\overline{\mathbf{X}}}$\ where
we do not invoke $\boldsymbol{\xi}$. It is the variation in $\boldsymbol{\xi}$
that is responsible for the time variation in $S(\overline{\mathbf{X}},t)$. A
simple proof of this conclusion is given in Sect. \ref{Sec-M_nieq}; see Remark
\ref{Remark-NonIEQ-S_X} also. We can summarize this conclusion as

\begin{conclusion}
The variation in time in $S(\overline{\mathbf{X}},t)$ in $\mathfrak{S}%
_{\overline{\mathbf{X}}}$ is due to the missing set of internal variables
$\boldsymbol{\xi}$.
\end{conclusion}

We now revert back to the standard use of $\mathbf{X}$, and $\mathbf{Z}$. Let
us consider a body $\Sigma$, which we take to be isolated and out of
equilibrium so that its macrostate $\mathcal{M}$ spontaneously relaxes towards
$\mathcal{M}_{\text{eq}}$ at fixed $\mathbf{X}$. Its entropy $S(\mathbf{X,}t)$
in $\mathfrak{S}_{\mathbf{X}}$\ has an explicit time dependence, which
continue to increase towards $S(\mathbf{X})$. For such NEQ states, the
explicit time dependence in $S(\mathbf{X,}t)$ is explained by introducing
$\boldsymbol{\xi}$ to make their entropies a state function in an
appropriately chosen larger state space $\mathfrak{S}_{\mathbf{Z}}$
\cite{Gujrati-II}. It is also shown there that a NIEQ macrostate with a
nonstate function entropy $S(\mathbf{Z,}t)$ may be converted to an IEQ
macrostate with a state function entropy $S(\mathbf{Z}^{\mathbf{\prime}})$ by
going to an appropriately chosen larger state space $\mathfrak{S}%
_{\mathbf{Z}^{\prime}}$ spanned by $\mathbf{Z}^{\prime}$ with $\mathfrak{S}%
_{\mathbf{Z}}$ its proper subspace. Therefore, in most cases of interest here,
we would be dealing with a state function and usually write it as
$S(\mathbf{Z})$, unless a choice for $\mathbf{Z}$ has been made based on the
experimental setup. In that case, we must deal with a pre-determined state
space $\mathfrak{S}_{\mathbf{Z}}$ so that some NEQ states that lie outside
$\mathfrak{S}_{\mathbf{Z}}$ can become a state function in some $\mathfrak{S}%
_{\mathbf{Z}^{\prime}}\supset\mathfrak{S}_{\mathbf{Z}}$.

We have discussed above that the explicit time dependence in a NEQ macrostate
with a nonstate function entropy $S_{\text{neq}}(t)\doteq S(\mathbf{X,}t)$ is
due to additional state variables in $\mathbf{\xi}$ and that this NEQ
macrostate may be converted into an IEQ macrostate with a macrostate function
entropy $S_{\text{ieq}}(\mathbf{Z})$ by going from $\mathfrak{S}_{\mathbf{X}}%
$\ to an appropriately chosen larger state space $\mathfrak{S}_{\mathbf{Z}}$.
Similarly, it has been shown \cite{Gujrati-II} that a NIEQ macrostate
$\mathcal{M}_{\text{nieq}}$ in $\mathfrak{S}_{\mathbf{Z}}$\ with a nonstate
function entropy $S_{\text{nieq}}(t)\doteq S(\mathbf{Z,}t)$ may be converted
to an IEQ macrostate $\mathcal{M}_{\text{ieq}}^{\prime}$ in an appropriately
chosen larger state space $\mathfrak{S}_{\mathbf{Z}^{\prime}}$ with a state
function entropy $S_{\text{ieq}}(\mathbf{Z}^{\prime})$.The additional internal
variables $\mathbf{\xi}^{\prime}$ that are over and above $\mathbf{\xi}$ in
$\mathbf{Z}^{\mathbf{\prime}}$ give rise to additional entropy generation as
they relax for fixed $\mathbf{Z}$. This results in the following inequality:%
\begin{equation}
S_{\text{ieq}}(\mathbf{Z})\geq S_{\text{ieq}}(\mathbf{Z}^{\prime
})=S_{\text{nieq}}(\mathbf{Z,}t). \label{EntropyBound-IEQ}%
\end{equation}
However, if the choice for $\mathbf{Z}$ has been made based on the
experimental setup and the observation time $\tau_{\text{obs}}$, see Sect.
\ref{Sec-Choice-S_Z}, we must restrict our discussion to $\mathfrak{S}%
_{\mathbf{Z}}$ so that we must consider $\mathcal{M}_{\text{nieq}}$ in
$\mathfrak{S}_{\mathbf{Z}}$\ the following. This will be done in Sect.
\ref{Sec-M_nieq}; see Remarks \ref{Remark-IEQ-ARB-Macrostate} and
\ref{Remark-NonIEQ-S_X}.

\subsection{A Proof of The Second Law\label{Sec-SecondLawProof}}

The second law has been proven so far under different assumptions \cite[among
others]{Tolman,Rice,Jaynes,Gujrati-ResidualEntropy,Gujrati-Symmetry}. Here, we
provide a simple proof of it based on the postulate of the flat distribution;
see Remark \ref{Remark-FlatDistribution}. The current proof is an extension of
the proof given earlier see also \cite[Theorem 4]{Gujrati-Symmetry}. We
consider an isolated system $\Sigma_{0}$ for which the second law is expressed
by Eq. (\ref{SecondLaw0}). However, for simplicity, we will suppress the
subscript $0$ from all the quantities in this section. As the law requires
considering the instantaneous entropy as a function of time, we need to focus
on the sample space at each instant to determine its entropy $S$ as a function
of time. At each instance, it is an ensemble average over the instantaneous
sample space $\Gamma(t)$ formed by the instantaneous set $\mathbf{m}(t)$ of
available microstates, see Eq. (\ref{S-Gibbs}) or (\ref{Gibbs_Formulation}).
We will use the flat distributions for the microstates at each instance, see
Remark \ref{Remark-FlatDistribution}, so that the entropy is given by Eq.
(\ref{S_BoltzmannZ}).

To prove the second law, see Proposition \ref{Prop-SecondLaw}, we proceed in
steps by considering a sequence of sample spaces belonging to $\Gamma$\ as
follows \cite{Gujrati-ResidualEntropy,Gujrati-Symmetry}. At a given instant, a
system happens to be in some microstate. We start at $t=t_{1}=0$, at which
time $\Sigma$ happens to be in a microstate, which we label $\mathfrak{m}_{1}%
$. It forms a sample space $\Gamma_{1}$ containing $\mathfrak{m}_{1}$ with
probability $p_{1}^{(1)}=1$, with the superscript denoting the sample space.
We have $S^{(1)}=0$. At some $t=t_{2}$, the sample space is enlarged from
$\Gamma_{1}$ to $\Gamma_{2}$, which contains $\mathfrak{m}_{1}$ and
$\mathfrak{m}_{2}$, with probabilities $p_{1}^{(2)}$ and $p_{2}^{(2)}$. Using
the flat distribution, the entropy is now $S_{2}=\ln2$. We just follow the
system in a sequence of time so that at $t=t_{n}$, we have a sample space
$\Gamma_{n}$ with $\mathfrak{m}_{1},\mathfrak{m}_{2},\cdots,\mathfrak{m}_{n}$
so that $S_{n}=\ln n$. Continuing this until all microstates in $\Gamma$ have
appeared, we have $S_{\text{max}}=\ln W$.

Thus, we have proven that the entropy continues to increase until it reaches
its maximum in accordance with Proposition \ref{Prop-SecondLaw}.

\section{Hamiltonian Trajectories in $\mathfrak{S}_{\mathbf{Z}}$%
\label{Sec-HamiltonianTrajectories}}

\subsection{Generalized Microforce and Microwork for $\Sigma$}

Traditional formulation of statistical thermodynamics
\cite{Landau,Gibbs,Gujrati-Symmetry} takes a mechanical approach in which
$\mathfrak{m}_{k}$ follows its classical or quantum mechanical evolution
dictated by its SI-Hamiltonian $\mathcal{H}(\left.  \mathbf{x}\right\vert
\mathbf{W})$. The quantum microstates are specified by a set of good quantum
numbers, which we have denoted by $k$ above as a single quantum number for
simplicity; we take $k\in\mathbb{N},\mathbb{N}$ denoting the set of natural
numbers. We will see below that $k$\ does not change as $\mathbf{W}$ changes.
In the classical case, we use a small cell $\delta\mathbf{x}_{k}$ around
$\mathbf{x}_{k}=\mathbf{x}$ as discussed above as the microstate
$\mathfrak{m}_{k}$. The Hamiltonian gives rise to a purely mechanical
evolution of individual $\mathfrak{m}_{k}$'s, which we will call the
\emph{Hamiltonian evolution}, and suffices to provide their mechanical
description. The change in $\mathcal{H}(\left.  \mathbf{x}\right\vert
\mathbf{W})$ in a process is%
\begin{subequations}
\begin{equation}
d\mathcal{H}=\frac{\partial\mathcal{H}}{\partial\mathbf{x}}\cdot
d\mathbf{x}+\frac{\partial\mathcal{H}}{\partial\mathbf{W}}\cdot d\mathbf{W}.
\label{HamiltonianChange}%
\end{equation}
The first term on the right vanishes identically due to Hamilton's equations
of motion for any $\mathfrak{m}_{k}$. Thus, for fixed $\mathbf{W}$, the energy
$E_{k}=\mathcal{H}_{k}\doteq\mathcal{H}(\left.  \mathbf{x}_{k}\right\vert
\mathbf{W})$ remains constant as $\mathfrak{m}_{k}$ moves about in
$\Gamma(\mathbf{W})$. Only the variation $d\mathbf{W}$ in $\mathfrak{S}%
_{\mathbf{Z}}$ generates any change in $E_{k}$. Consequently, we do not worry
about how $\mathbf{x}_{k}$\ changes in $\mathcal{H}(\left.  \mathbf{x}%
\right\vert \mathbf{W})$ in the phase space, and focus, instead, on the state
space $\mathfrak{S}_{\mathbf{Z}}$, in which\ can write%
\begin{equation}
dE_{k}=\frac{\partial E_{k}}{\partial\mathbf{W}}\cdot d\mathbf{W}=-dW_{k},
\label{HamiltonianChange-StateSpace}%
\end{equation}
where $dW_{k}$ denotes the\emph{ generalized microwork} produced by the
\emph{generalized microforce} $\mathbf{F}_{\text{w}k}$:
\begin{equation}
dW_{k}=\mathbf{F}_{\text{w}k}\cdot d\mathbf{W},\ \mathbf{F}_{\text{w}k}%
\doteq-\partial E_{k}/\partial\mathbf{W}. \label{Generalized force-work}%
\end{equation}
For the case $\mathbf{W}=(V,\xi)$, the corresponding microforce $\mathbf{F}%
_{\text{w}k}$ is ($P_{k},A_{k})$, where%
\end{subequations}
\begin{equation}
P_{k}=-\partial E_{k}/\partial V,A_{k}=-\partial E_{k}/\partial\xi.
\label{GeneralizedMicroforce-EX}%
\end{equation}
The corresponding microwork is%
\begin{equation}
dW_{k}=P_{k}dV+A_{k}d\xi. \label{GeneralizedMicrowork-Ex}%
\end{equation}

\subsection{Statistical Significance of $dW$ and $dQ$\label{Sec_Stat_Concepts}%
}

Before proceeding further, let us see how the generalized macrowork and
macroheat could be understood from a statistical point of view so that we can
identify them using the Hamiltonian. Once $\mathbf{W}$ has been identified,
the Hamiltonian must be expressed in terms of it. Thus, $\mathfrak{m}_{k}$ and
$E_{k}$ are functions of $\mathbf{W}$ in $\mathfrak{S}_{\mathbf{Z}}$. We now prove

\begin{theorem}
\label{Theorem_1}$E(t)$ is a function of $\mathbf{W}(t)$ and $S(t)$ for any
$\mathcal{M}_{\text{arb}}$, even though $E_{k}[\mathbf{W}(t)]$'s are functions
of $\mathbf{W}(t)$ only.
\end{theorem}

\begin{proof}
We consider the differential
\begin{equation}
dE(t)\equiv%
{\textstyle\sum\nolimits_{k}}
p_{k}(t)dE_{k}(t)+%
{\textstyle\sum\nolimits_{k}}
E_{k}(t)dp_{k}(t). \label{Energy_Division-m-s}%
\end{equation}
As $p_{k}(t)$'s are unchanged in the first sum, this sum is evaluated at
\emph{constant entropy} so this is purely mechanical macroquantity
$dE_{\text{m}}$; see Eq. (\ref{Macro-heat-work-alpha}). This sum is a function
of $\mathbf{W}(t)$ as is seen clearly in Eq.
(\ref{HamiltonianChange-StateSpace}). The second contribution is at fixed
microstate energies $E_{k}$ so $\mathbf{W}(t)$ is held fixed, but require
changes in the probabilities so it is the stochastic contribution
$dE_{\text{s}}$, see Eq. (\ref{HamiltonianChange-StateSpace}). The changes
$\left\{  dp_{k}(t)\right\}  $ result in is $dS$. As $dE_{\text{s}}$ and
\begin{equation}
dS=-%
{\textstyle\sum\nolimits_{k}}
(\eta_{k}(t)+1)dp_{k}(t)=-%
{\textstyle\sum\nolimits_{k}}
\eta_{k}(t)dp_{k}(t) \label{dS-Gibbs}%
\end{equation}
are both extensive, they must be linearly related with an intensive constant
of proportionality. This proves that $E(t)$ is a function of $S(t)$ and
$\mathbf{W}(t)$ in general for any $\mathcal{M}_{\text{arb}}$.
\end{proof}

Note that we have used the identity $%
{\textstyle\sum\nolimits_{k}}
dp_{k}=0$ above; see also Eq. (\ref{p_alpha-sum}).

We introduce a special process, to be called a generalized \emph{isometric}
process, which is a process at fixed $\mathbf{W}(t)$ and is a generalization
of an \emph{isochoric} process. In this process, the work done by each
mechanical variables in $\mathbf{W}(t)$ remains zero so $dE_{\text{m}}\equiv
0$. We now prove the following theorem that establishes the physical
significance of the two contributions.

\begin{theorem}
\label{Theorem_Heat_Work}The isentropic contribution represents the
generalized macrowork $dW(t)$ and the stochastic contribution represents the
generalized macroheat $dQ(t)$ for any $\mathcal{M}_{\text{arb}}$.
\end{theorem}

\begin{proof}
We follow Landau and Lifshitz \cite{Landau} and rewrite the first term in Eq.
(\ref{Energy_Division-m-s}) as%
\begin{align*}
dE_{\text{m}}(t)  &  \equiv%
{\textstyle\sum\nolimits_{k}}
p_{k}(t)\frac{\partial E_{k}}{\partial\mathbf{W}}\cdot d\mathbf{W}(t)\\
&  =-%
{\textstyle\sum\nolimits_{k}}
p_{k}(t)\mathbf{F}_{\text{w}k}(t)\cdot d\mathbf{W}(t)
\end{align*}
where we have used Eq. (\ref{Generalized force-work}). The use of Eqs.
(\ref{GeneralizedMacroforce}) and (\ref{GeneralizedMacrowork}) proves that
\begin{equation}
dW(t)\equiv-dE_{\text{m}}(t) \label{dE_dW}%
\end{equation}
is the isentropic contribution, making macrowork a mechanical concept as we
have already pointed out. This identification then also proves that the
macroheat in the first law, see Eq. (\ref{FirstLaw-SI}), must be properly
identified with $dQ(t)$. Accordingly,
\begin{equation}
dQ(t)\equiv dE_{\text{s}}(t)\equiv%
{\textstyle\sum\nolimits_{k}}
E_{k}(t)dp_{k}(t), \label{dE_dQ}%
\end{equation}
is purely stochastic.
\end{proof}

The linear proportionality between $dQ=dE_{\text{s}}$ and $dS$ mentioned above
in the proof of Theorem \ref{Theorem_1} results in
\begin{equation}
dQ(t)/dS(t)=T_{\text{arb}}(t), \label{System_dQ_dS}%
\end{equation}
which is a statistical proof of the identity in Eq. (\ref{dQ-dS}) relating
$dQ(t)$ and $dS(t)$ for any $\mathcal{M}_{\text{arb}}$. We also note that the
ratio $T_{\text{arb}}(t)$ is related to the ratio of two SI-macroquantities.
Thus, it can be used to characterize the instantaneous macrostate
$\mathcal{M}_{\text{arb}}$. This should be contrasted with the \r{M}NEQT, in
which the ratio
\begin{equation}
d_{\text{e}}Q(t)/d_{\text{e}}S(t)=T_{0} \label{Medium_dQ_dS}%
\end{equation}
does not characterize the instantaneous macrostate $\mathcal{M}_{\text{arb}}$.
In Eq. (\ref{beta_arb}), we provide a general procedure for a thermodynamic
identification of $T_{\text{arb}}$.

\begin{remark}
\label{Remark-Regardless-of-Speed}It is worth emphasizing that $dQ(t)$ and
$dS(t)$ in Eqs. (\ref{dE_dQ}-\ref{dS-Gibbs}) are defined as instantaneous
quantities in terms of the instantaneous changes $\left\{  dp_{k}(t)\right\}
$, regardless of the speed of the segmental process $d\mathcal{P}_{\text{arb}%
}\doteq d\mathcal{P}(t)$, and instantaneous values $\left\{  E_{k}(t)\right\}
$ and $\left\{  p_{k}(t)\right\}  $. Therefore, the generalized macroheat and
entropy change are defined regardless of the speed of the arbitrary process.
As $dE(t)$ in Eq. (\ref{Energy_Division-m-s}) is also defined instantaneously,
it is clear from Eq. (\ref{FirstLaw-SI}) that the generalized work $dW$ is
also defined instantaneously regardless of the speed of the arbitrary process.
This is consistent with our above derivation of $dW$ in terms of generalized
forces. \emph{The observation is very important as it shows that the existence
of all SI-quantities does not depend on the speed of the arbitrary process}
$d\mathcal{P}_{\text{arb}}$. However, see also Sect. \ref{Sec-Choice-S_Z}
further clarification on the importance of $\tau_{\text{obs}}$. From now
onward, we will not make a distinction between $T$ and $T_{\text{arb}}$.
\end{remark}

We should point out that, as $\mathbf{W}(t)$ is a parameter, $d\mathbf{W}(t)$
is the same for all microstates. The statistical nature of $dE_{\text{m}}$ is
reflected in the statistical nature of $\mathbf{F}_{\text{w}}(t)$,$\ $such as
$P(t)$ and $A(t)$, of the system. Thus, the SI-fields $\mathbf{F}_{\text{w}%
k}(t)$ are \emph{fluctuating} quantities from microstate to microstate as
expected in any averaging process.

We can now identify $\mathbf{W}$ as the \emph{macrowork parameter}, and the
variation $d\mathbf{Z}(t)\doteq(dE(t),d\mathbf{W}(t))$ in $\mathfrak{S}%
_{\mathbf{Z}}$ defines not only the microwork $\left\{  dW_{k}\right\}  $, but
also\ a thermodynamic process $\mathcal{P}$. The trajectory $\gamma_{k}$ in
$\mathfrak{S}_{\mathbf{Z}}$ followed by $\mathfrak{m}_{k}$ as a function of
time will be called the \emph{Hamiltonian trajectory} during which
$\mathbf{W}$ varies from its initial (in) value $\mathbf{W}_{\text{in}}$\ to
its final (fin) value $\mathbf{W}_{\text{fin}}$ during $\mathcal{P}$, the the
path $\gamma_{\mathcal{P}}$ denotes the path the macrostate follows during
this process; see Definition \ref{Def-Path-trajectories}. The variation
produces the generalized microwork $dW_{k}$. As $p_{k}$ plays no role in
$dW_{k}$, its determination is simplifies in the MNEQT. The microwork $dW_{k}$
also does not change the index $k$ of $\mathfrak{m}_{k}$ as said above. The
ensemble average of $\mathbf{F}_{\text{w}k}$ is $\mathbf{F}_{\text{w}}$, see
Eq. (\ref{GeneralizedForce-W}),%

\begin{subequations}
\begin{equation}
\mathbf{F}_{\text{w}}=\left\langle \mathbf{F}_{\text{w}}\right\rangle \doteq%
{\textstyle\sum\nolimits_{k}}
p_{k}\mathbf{F}_{\text{w}k}; \label{GeneralizedMacroforce}%
\end{equation}
that of $dW_{k}$ is $dW$ given by
\begin{equation}
dW=\left\langle dW\right\rangle \doteq%
{\textstyle\sum\nolimits_{k}}
p_{k}dW_{k}\doteq\mathbf{F}_{\text{w}}\cdot d\mathbf{W,}
\label{GeneralizedMacrowork}%
\end{equation}
as given earlier in Eq. (\ref{GeneralizedWork-W}). It is based on using the
mechanical definition (force X displacement) of work. The macroforce
corresponding to $\mathbf{W}=(V,\xi)$ is $\mathbf{F}_{\text{w}}=(P,A)$, where
$P=\left\langle P\right\rangle $, and $A=\left\langle A\right\rangle $. The
corresponding SI-macrowork is given earlier in Eq. (\ref{SI-Work}).

The above discussion proves that the definition of macroheat and macrowork is
valid for any $\mathcal{M}_{\text{arb}}$.\ It is useful to compare the above
approach with the traditional formulation of the first law in terms of
$d_{\text{e}}Q(t)$ and $d_{\text{e}}W(t)$: \emph{both formulations are valid
in all cases}. It should be mentioned that the above identification is well
known in equilibrium statistical mechanics, but its extension to irreversible
processes and our interpretation is, to the best of our knowledge, novel.
While the instantaneous average $\mathbf{F}_{\text{w}}(t)$ such as the
pressure $P(t)$ is mechanically defined under all circumstances, it will only
be identified with the thermodynamic definition of the instantaneous pressure%
\end{subequations}
\begin{equation}
P(t)=-\left(  \partial E/\partial V\right)  _{S,\xi} \label{Pressure}%
\end{equation}
for a uniquely identified macrostate in $\mathfrak{S}_{\mathbf{Z}}$.

Being purely mechanical in nature, a trajectory is completely
\emph{deterministic} and cannot describe the evolution of a macrostate
$\mathcal{M}$ during $\mathcal{P}$ unless supplemented by thermodynamic
stochasticity, which requires $p_{k}(\mathcal{M})$ as discussed above
\cite{Landau}. Thermodynamics emerges when quantities pertaining to the
trajectories are averaged over the trajectory ensemble $\left\{  \gamma
_{k}\right\}  $ with appropriate probabilities that will usually change during
the process.

\begin{conclusion}
\label{Conc-Isentropic-Isometric-Changes} The change $dE$ consists of two
\emph{independent} contributions- an isentropic change $dE_{\text{m}}=-dW$,
and an stochastic change $dE_{\text{\textbf{s}}}=T_{\text{arb}}dS$. On the
other hand, the MI-macroheat and the MI-macrowork suffer from ambiguity; see,
for example, Kestin \cite{Kestin}.
\end{conclusion}

\begin{remark}
It is clear from the above discussion that it is the macroheat and not the
macrowork that causes $p_{k}(t)$, and therefore the entropy to change. This is
the essence of the common wisdom that heat is \emph{random motion}. But we now
have a mathematical definition: macroheat is the isometric part of $dE(t)$
that is directly related to the change in the entropy through changes in
$p_{k}(t)$. Macrowork is that part of the energy change caused by isentropic
variations in the "mechanical" state variables $\mathbf{W}(t)$. This is true
no matter how far the system is from equilibrium. Thus, our formulation of the
first law and the identification of the two terms is the most general one, and
applicable to any $\mathcal{M}_{\text{arb}}$.
\end{remark}

\begin{remark}
The relationship between the macroheat and the entropy becomes simple only
when $M$ happens to be in internal equilibrium, see Sec.
\ref{Sec-InternalEquilibrium}, in which case $T_{\text{arb}}(t)$ is replaced
by $T(t)$, which has a thermodynamic significance; see Eq. (\ref{T-beta}) and
we have the thermodynamic identity, called the Clausius Equality in Eq.
(\ref{dQ-dS}) $dQ(t)=T(t)dS(t)$ for $\mathcal{M}_{\text{ieq}}$,\ which is very
interesting in that it turns the well-known Clausius inequality $d_{\text{e}%
}Q=T_{0}d_{\text{e}}S\leq$ $T_{0}dS$ into an equality.
\end{remark}

For the sake of completeness, we briefly discuss the various attempts to the
study of the microanalogs $dW_{k}$ and $dQ_{k}$ of the $dW$ and $\ dQ$,
respectively, that has flourished into an active field in diverse branches of
NEQT at diverse length scales from mesoscopic to macroscopic lengths
\cite{Bochkov,Jarzynski,Crooks,Pitaevskii,Sekimoto,Sekimoto-Book,Seifert,Lebowitz,Alicki}%
; see also some recent reviews \cite{Seifert-Rev,Maruyama,Broeck}.
Unfortunately, this endeavor is apparently far from complete
\cite{Gislason,Bertrand,Bauman,Kestin,Kievelson,Bizarro,Honig,Jarzynski,Rubi,Peliti,Cohen,Sekimoto,Seifert,Nieuwenhuizen,Crooks,Jarzynski-Cohen,Sung,Gross,Jarzynski-Gross,Jarzynski-Rubi,Rubi-Jarzynski,Peliti-Rubi,Rubi-Peliti,Pitaevskii,Bochkov,Maruyama,Seifert-Rev,Broeck,Sekimoto-Book,Lebowitz,Alicki}%
. This is because of the confusion about the meaning of macrowork and
macroheat even in classical NEQT \cite{Fermi,Kestin} involving SI-\emph{ }or
MI-\ description, which has only recently been clarified
\cite{Gujrati-Heat-Work0,Gujrati-Heat-Work,Gujrati-I,Gujrati-II,Gujrati-III,Gujrati-Entropy1,Gujrati-Entropy2,Gujrati-LangevinEq}
in the MNEQT, where a clear distinction is made between the generalized
macrowork (macroheat) $dW$ ($dQ$) and the exchange macrowork (macroheat)
$d_{\text{e}}W$ ($d_{\text{e}}Q$). In an EQ process, both macroworks
(macroheats) have the same magnitude, but not in a NEQ\ process, where the
difference determines $d_{\text{i}}W\geq0$ ($d_{\text{i}}Q\geq0$).

It is important to draw attention to the following important fact. We first
recognize that the first law in Eq. (\ref{FirstLaw-MI}) refers to the change
in the energy caused by exchange quantities. Therefore, $dE$ on the left truly
represents $d_{\text{e}}E$. Accordingly, we write Eq. (\ref{FirstLaw-MI}) as%
\begin{equation}
d_{\text{e}}E=d_{\text{e}}Q-d_{\text{e}}W, \label{First-Law-deE}%
\end{equation}
which justifies Eq. (\ref{FirstLaw-alpha}) for $d_{\text{e}}$. Subtracting
this equation from Eq. (\ref{FirstLaw-SI}), we obtain the identity%
\begin{equation}
d_{\text{i}}E=d_{\text{i}}Q-d_{\text{i}}W\equiv0, \label{First-Law-diE}%
\end{equation}
which not only justifies Eq. (\ref{FirstLaw-alpha}) for $d_{\text{i}}$ but
also Eq. (\ref{diQ-diW-EQ}) for which we have used Eq.
(\ref{E-and-V-Partition}).

\begin{remark}
The above analysis demonstrates the important fact that the first law can be
applied either to the exchange process ($d_{\text{e}}$) or to the interior
process ($d_{\text{i}}$). The last formulation is also applicable to an
isolated system.
\end{remark}

\subsection{Medium $\widetilde{\Sigma}$}

The above discussion can be easily extended to the medium (the suffix
$\widetilde{k}$\ denotes its microstates) with the following results%
\begin{align}
d\widetilde{W}(t)  &  =-d\widetilde{E}_{\text{m}}\equiv-%
{\textstyle\sum\nolimits_{\widetilde{k}}}
\widetilde{p}_{\widetilde{k}}\frac{\partial\widetilde{E}_{\widetilde{k}}%
}{\partial\widetilde{\mathbf{w}}}\cdot d\widetilde{\mathbf{w}}\nonumber\\
&  =\mathbf{f}_{\text{w}0}\cdot d\widetilde{\mathbf{w}}=-d_{\text{e}%
}W,\label{MediumHeat-Work}\\
d\widetilde{Q}(t)  &  =d\widetilde{E}_{\text{s}}\equiv%
{\textstyle\sum\nolimits_{\widetilde{k}}}
\widetilde{E}_{\widetilde{k}}d\widetilde{p}_{\widetilde{k}}=-d_{\text{e}%
}Q,\nonumber
\end{align}
where all the quantities including $\widetilde{k}$ refer to the medium, except
$d_{\text{e}}W$ and $d_{\text{e}}Q$, and have their standard meaning. The
analog of Eq. (\ref{System_dQ_dS}) is $d\widetilde{Q/}d\widetilde{S}=T_{0}$ as
expected; see Eq. (\ref{Medium_dQ_dS}). We clearly see that
\begin{subequations}
\begin{equation}
dW_{0}\doteq dW+d\widetilde{W}=d_{\text{i}}W\geq0 \label{Isolated-Work}%
\end{equation}
such as when mechanical equilibrium is not present. In this case, we also have%
\begin{equation}
dQ_{0}\doteq dQ+d\widetilde{Q}=d_{\text{i}}Q\geq0, \label{Isolated-Heat}%
\end{equation}
with $dW_{0}=dQ_{0}$ in view of Eq. (\ref{diQ-diW-EQ}). In a finite process
$\mathcal{P}$, all infinitesimal quantities are replaced by their net changes%
\begin{equation}
\Delta W_{0}\doteq\Delta W+\Delta\widetilde{W}=\Delta Q_{0}=\Delta_{\text{i}%
}W\geq0, \label{Isolated-Heat-Work}%
\end{equation}
where $\Delta_{\text{i}}W$ is obtained by integrating $d_{\text{i}}W$ in Eq.
(\ref{WorkInequality}) over $\mathcal{P}$; the result is given in Eq.
(\ref{Delta_i-W-Full}), where it is discussed.

\subsection{Irreversible Macrowork and Macroheat \ \ }

We can now identify $d_{\text{i}}W(t)$ and $d_{\text{i}}Q(t):$%
\end{subequations}
\begin{equation}
d_{\text{i}}W(t)\equiv-(dE_{\text{m}}+d\widetilde{E}_{\text{m}}),d_{\text{i}%
}Q(t)\equiv(dE_{\text{s}}+d\widetilde{E}_{\text{s}}),
\label{IrreversibleWork-Heat}%
\end{equation}
satisfying Eq. (\ref{diQ-diW-EQ}), which follows from $d_{\text{i}}E=0$; see
Eq. (\ref{E-and-V-Partition}). It is easy to see that $d_{\text{i}}W(t)$
reproduces Eq. (\ref{Irrev-Work}), where we must use $d\widetilde{\mathbf{w}%
}=-d_{\text{e}}\mathbf{w}$, and $d\mathbf{w=}d_{\text{e}}\mathbf{w+}%
d_{\text{i}}\mathbf{w}$.

\section{Unique Macrostates \label{Sec-UniqueMicrostates}}

\subsection{Internal Equilibrium \label{Sec-InternalEquilibrium}}

We now revert back to the original notation $\mathbf{X}$\ and $\mathbf{Z}%
$.\ We will refer to $S(\mathbf{Z}(t))$ in terms of microstate number
$W(\mathbf{Z}(t))$ in Eq. (\ref{S_BoltzmannZ}) as the \emph{time-dependent
Boltzmann formulation }of the entropy or simply the Boltzmann entropy
\cite{Lebowitz}, whereas $S(\mathbf{X})$ in Eq. (\ref{S_Boltzmann}) represents
the equilibrium (Boltzmann) entropy. It is evident that the Gibbs formulation
in Eqs. (\ref{Gibbs_Formulation}) and (\ref{Conventional_Entropy0}) supersedes
the Boltzmann formulation in Eqs.(\ref{S_Boltzmann0}) and (\ref{S_Boltzmann}),
respectively, as the former contains the latter as a special limit. However,
it should be also noted that there are competing views on which entropy is
more general \cite{Lebowitz,Ruelle}. We believe that the above derivation,
being general, makes the Gibbs formulation more fundamental. The continuity of
$S(\mathbf{Z},t)$ follows directly from the continuity of $p_{k}(t)$. Its
existence follows from the observation that it is bounded above by $\ln
W(\mathbf{Z})$ and bounded below by $0$, see Eq. (\ref{S_BoltzmannZ}).

We now introduce the central concept of the MNEQT, which is based on the
existence of $S(\mathbf{Z})$ above; see Definition \ref{Def-NEQ-States}, which
we now expand.

\begin{definition}
\label{Def-InternalEQ} A NEQ macrostate $\mathcal{M}$ whose entropy is a state
function $S(\mathbf{Z})$ in $\mathfrak{S}_{\mathbf{Z}}$ is said to be an
\emph{internal equilibrium} (IEQ) macrostate $\mathcal{M}_{\text{ieq}}$
\cite{Gujrati-I,Gujrati-II}; if not, its entropy $S(\mathbf{Z},t)$ is an
explicit function of time $t$ in $\mathfrak{S}_{\mathbf{Z}}$. An
IEQ-macrostate in $\mathfrak{S}_{\mathbf{Z}}$ is a unique macrostate in
$\mathfrak{S}_{\mathbf{Z}}$.
\end{definition}

We clarify this point. If we do not use $\boldsymbol{\xi}$ for $\mathcal{M}$,
which is not unique in $\mathfrak{S}_{\mathbf{X}}$, then its entropy cannot be
a state function in $\mathfrak{S}_{\mathbf{X}}$, and must be expressed as
$S(\mathbf{X},t)$. Thus, the importance of $\boldsymbol{\xi}$ is to be able to
deal with a state function entropy $S(\mathbf{Z})$ by choosing an
\emph{appropriate} number of internal variables. Throughout this work, we will
only deal with IEQ macrostates. However, as we will see, our discussion of NEQ
macrowork will cover all states.

Being a state function, $S(\mathbf{Z})$ shares many of the properties of EQ
entropy $S(\mathbf{X})$, see Definition \ref{Def-NEQ-States}:

\begin{enumerate}
\item \textit{Maximum}: $S(\mathbf{Z})$ is the maximum possible value of the
NEQ entropy in $\mathfrak{S}_{\mathbf{Z}}$ for a given $\mathbf{Z}$
\cite{Gujrati-II}.

\item \textit{No memory} -Its value also does not depend on how the system
arrives in $\mathcal{M}_{\text{ieq}}\equiv\mathcal{M}(\mathbf{Z})$,
\textit{i.e., }whether\textit{ }it arrives there from another IEQ macrostate
or a non-IEQ macrostate \cite{Gujrati-II}. Thus, it has no memory of the
earlier macrostate.
\end{enumerate}

There are some macrostates that emerge in fast changing processes such as the
free expansion that possess memory of the initial states so that their entropy
will no longer be a state function in $\mathfrak{S}_{\mathbf{X}}$. In this
case, we need to enlarge the state space to $\mathfrak{S}_{\mathbf{Z}}$ by
including internal variables as done in Sec. \ref{Sec-Free Expansion}.

\begin{remark}
It may appear to a reader that the concept of entropy being a state function
is very restrictive. This is not the case as this concept, although not
recognized by several workers, is implicit in the literature where the
relationship of the thermodynamic entropy with state variables is
investigated. To appreciate this, we observe that the entropy of a body in
internal equilibrium \cite{Gujrati-I,Gujrati-II} is given by the Boltzmann
formula in Eq. (\ref{S_BoltzmannZ}) in terms of the number of microstates
corresponding to $\mathbf{Z}(t)$. In classical nonequilibrium thermodynamics
\cite{deGroot}, the entropy is always taken to be a state function. In the
Edwards approach \cite{Edwards} for granular materials, all microstates are
equally probable as is required for the above Boltzmann formula. Bouchbinder
and Langer \cite{Langer} assume that the nonequilibrium entropy is given by
Eq. (\ref{S_BoltzmannZ}). Lebowitz \cite{Lebowitz} also takes the above
formulation for his definition of the nonequilibrium entropy. As a matter of
fact, we are not aware of any work dealing with entropy computation that does
not assume the nonequilibrium entropy to be \ a state function. This does not,
of course, mean that all states of a system are internal equilibrium states.
For states that are not in internal equilibrium, the entropy is not a state
function so that it will have an explicit time dependence. But, as shown
elsewhere,\cite{Gujrati-II} this can be avoided by enlarging the space of
internal variables. The choice of how many internal variables are needed will
depend on experimental time scales and cannot be answered in generality just
as is the case in EQ thermodynamics for the number of observables; the latter
depends on the experimental setup. A detailed discussion is offered elsewhere
\cite{Gujrati-Hierarchy}.
\end{remark}

\subsection{Gibbs Fundamental Relation\label{Sec-GibbsRelation}}

Being a state function, $S(\mathbf{Z})$ in $\mathfrak{S}_{\mathbf{Z}}$ for
$\mathcal{M}_{\text{ieq}}$\ results in the following Gibbs fundamental
relation for the entropy%
\begin{subequations}
\begin{equation}
dS=\frac{\partial S}{\partial\mathbf{Z}}\cdot d\mathbf{Z=}\frac{\partial
S}{\partial E}dE\mathbf{+}\frac{\partial S}{\partial\mathbf{W}}\cdot
d\mathbf{W}, \label{Gibbs-FR-S}%
\end{equation}

which can be inverted to express the Gibbs fundamental relation for the energy
as%
\begin{equation}
dE=TdS-\mathbf{F}_{\text{w}}\cdot d\mathbf{W,} \label{Gibbs-FR-E}%
\end{equation}

where we have introduced%
\end{subequations}
\begin{align}
\beta &  =1/T\doteq\partial S/\partial E=1/(\partial E/\partial
S),\label{T-beta}\\
\mathbf{F}_{\text{w}}  &  \doteq T\partial S/\partial\mathbf{W}=-\partial
E/\partial\mathbf{W} \label{Fw}%
\end{align}

as the inverse temperature of the system (we set the Boltzmann constant
$k_{\text{B}}=1$ throughout the review), and have used Eq.
(\ref{GeneralizedForce-W}) for the generalized macroforce $\mathbf{F}%
_{\text{w}}$. Recalling Eq. (\ref{GeneralizedMacrowork}), we see that the
second term in Eq. (\ref{Gibbs-FR-E}) is nothing but the SI-macrowork $dW$.
Comparing Eq. (\ref{Gibbs-FR-E}) with Eq. (\ref{FirstLaw-SI}), we can identify
the generalized macroheat $dQ$ with $TdS$, which then proves Eq. (\ref{dQ-dS}).

It should be stated here that the choice and the number of state variables
included in $\mathbf{X}$ or $\mathbf{Z}$\ is not so trivial and must be
determined by the nature of the experiments \cite{Maugin}. \emph{We will
simply assume here that they have been specified}.\emph{ }Just as
$S=S(\mathbf{X})$ is a state function of $\mathbf{X}$ for $\mathcal{M}%
_{\text{eq}}$ in $\mathfrak{S}_{\mathbf{X}}$, there are $\mathcal{M}%
_{\text{ieq}}$ in $\mathfrak{S}_{\mathbf{Z}}$ for which $S(\mathbf{Z})$ is a
state function of $\mathbf{Z}$.

The possibility of a Gibbs fundamental relation for $\mathcal{M}_{\text{nieq}%
}$ is deferred to Sect. \ref{Sec-M_nieq}.

\subsection{A Digression on the NEQ-Temperature\label{Sec-Digression-T}}

While the concept of the macrowork is quite familiar from mechanics, the
concept of the macroheat is peculiar to thermodynamics in view of Eq.
(\ref{dQ-dS}). In EQ thermodynamics, the macroheat $dQ$ is directly
proportional to the change $dS$, and the constant of proportionality
determines the EQ temperature $T$. Indeed, the concepts of entropy and of
temperature are unique to thermodynamics and are well established in EQ
thermodynamics. A $\Sigma$ in thermal equilibrium with a $\widetilde{\Sigma}$
at $T_{0}$ obviously has the same temperature $T_{0}$. The temperature for an
isolated system in equilibrium is also well defined; its inverse is identified
with the energy derivative of the \emph{equilibrium} entropy \cite{Landau}.
The definition is valid for \emph{all} EQ systems, even those containing
gravitational interaction. This is confirmed by the fact that Bekenstein used
it to identify the temperature of an isolated black hole
\cite{Bekenstein,Hawking}. The formulation is valid both classically and
quantum mechanically \cite{Landau}.

The EQ definition of the temperature is formally identical to that in Eq.
(\ref{T-beta}), which is valid in NEQT
\cite{Gujrati-Symmetry,Gujrati-Heat-Work0,Gujrati-Heat-Work,Gujrati-Entropy1,Gujrati-Entropy2,Gujrati-I,Gujrati-II,Gujrati-III}%
. In this, we have a general thermodynamic definition of a temperature for any
$\mathcal{M}$. It is important to realize that the notion of a NEQ temperature
is an absolute necessity for the \emph{Clausius statement} of the second law
that the exchanged macroheat flows spontaneously from hot to cold to be meaningful.

It is clear from the above discussion that macrowork is the isentropic change
in the energy, while macroheat is the energy change due to the entropy change.
This is not as surprising a statement as it appears, since a mechanical system
is usually thought of as a system for which the entropy concept is not
meaningful. A different way to state this is that the entropy remains constant
(isentropic) in any mechanical process as we have done above. Planck
\cite{Planck} had already suggested that the temperature should be defined for
NEQ macrostates just as the entropy should be defined for them if we need to
carry out a thermodynamic investigation of a NEQ system. Such a temperature
was apparently first introduced by Landau \cite{Landau0} for partial set of
the degrees of freedom (dof). This then allows the possibility that the notion
of temperature can be separately applied, for example, to vibrational and
configurational dof in glasses that are known to be out of equilibrium with
each other \cite{Goldstein} in that they are ascribed different temperatures.
This means that macroheat would be exchanged between them until they come to
equilibrium, but this is internally exchanged. But there seems to be a lot of
confusion about the meaning of the entropy and temperature in NEQT \cite[for
example]{Muschik,Keizer,Morriss,Jou,Hoover,Ruelle,Evans,Oono,Maugin}, where
different definitions lead to different results. In contrast, the meaning of
entropy and fields in equilibrium thermodynamics has no such problem.

We agree with Planck and believe that there must exist a unifying approach to
identify the temperature for $\mathcal{M}_{\text{arb}}$; see Definition
\ref{Def-ArbitrayState}, with or without memory effects in $\mathfrak{S}%
_{\mathbf{Z}}$. The inverse temperature defined above in Eq. (\ref{T-beta}) is
not directly applicable to nonIEQ-macrostates in $\mathfrak{S}_{\mathbf{Z}}$
for which $S$ is not a state function, but can be extended to them so as to
accommodate memory effects as we do in Sect. \ref{Sec-M_nieq}. However, we
will not consider them in detail in this review.

\begin{criterion}
\label{Criterion-Temperature}The identification of temperature in
$\mathcal{M}_{\text{arb}}$ must satisfy some stringent but obvious criteria:
\end{criterion}

\begin{enumerate}
\item[C1] It must be intensive and must reduce to the temperature determined
by Eq. (\ref{T-beta}) for $\mathcal{M}_{\text{eq}}$ and $\mathcal{M}%
_{\text{ieq}}$ even for an isolated system.

\item[C2] It must cover negative temperatures \cite{Ramsey} that are commonly
observed for some dof such as nuclear spins in a system. As these dof are not
involved in any macroscopic motion \cite[Sec. 73]{Landau}, there is no kinetic
energy involved. Most common occurrence of a negative temperature is when the
above spin dof are out of equilibrium with the other dof such as lattice
vibrations in the system.

\item[C3] It must satisfy the Clausius statement that\ macroheat between two
objects always flows spontaneously from hot to cold for positive temperatures.
When negative temperatures are considered, macroheat must flow from a system
at a negative temperature to a system at a positive temperature.

\item[C4] It must be a global rather than a local property of the system so
that we can differentiate hot and cold between two different systems.
\end{enumerate}

The first criterion ensures that the new temperature is an extension of the
conventional notion of the temperature that is valid when the entropy is a
state function. This means that the new notion of temperature is valid for any
arbitrary macrostate. In addition, it must exist even for an isolated system.
The second criterion ensures that our formalism includes negative temperatures
that may occur in a lattice system. The third criterion ensures compliance
with the second law for interacting systems. This is a very important
criterion, which \emph{every notion of temperature must satisfy}. We will come
back to this issue again in Sec. \ref{Sec-IrreversibleInequalities} where we
prove it in the MNEQT. The last criterion ensures that the temperature is
associated with the entire system, whether the system is homogeneous or not.
This will be explained by direct calculations of inhomogeneous systems in
Sect. \ref{Sec-Applications}. By extension, the concept of a NEQ temperature
can be also applied to different dof of a system such as a glass under the
assumption that they are weakly interacting in accordance with the approach
taken by Landau \cite{Landau0}. This results in the Tool-Narayanaswamy
relation derived in Sect. \ref{Sec-Tool-Narayan}.

Before we close this discussion, we wish to point out major differences
between the NEQ temperature $T$ in the MNEQT and its other definitions. We
first consider the \r{M}NEQT. The most important theories belonging to this
class are the classical local irreversible thermodynamics (LNEQT)
\cite{deGroot}, the rational thermodynamics (RNEQT) \cite{ColemanRational},
and the extended irreversible thermodynamics (ENEQT) \cite{Jou} as we had
mentioned earlier. We refer the reader to \cite{Maugin,Jou} for excellent
reviews on these theories that use local densities of energy $e$ and entropy
$s$. They are continuum theories, and can all be classified as continuum \r
{M}NEQT to be denoted by the CNEQT here. We consider them critically later in
Sect. \ref{Sec-Conclusions}. They differ in the choice of their state spaces.
Considering the local entropy and energy densities $s$ and $e$, the inverse
local temperature is defined as $\partial s/\partial e$, and differs from the
global temperature in the MNEQT.

\begin{enumerate}
\item In the LNEQT, each local volume element is in EQ so the local
temperature is the EQ temperature of the volume element, and differs from $T$,
which is a global temperature.

\item In the RNEQT, the temperature is taken as a primitive quantity along
with the entropy. Because of the memory effect, the temperature at any time
depends on the entire history. Thus, it is a local analog of the global
temperature of $\mathcal{M}_{\text{nieq}}$ in the MNEQT, but the latter is
defined thermodynamically.

\item In the ENEQT, the fluxes are part of the state variables so the local
temperature also depends on them. Assuming the total entropy to also depend on
the fluxes \cite[see Eq. (5.66), for example]{Jou0}, one can identify the
global analog of the temperature in the ENEQT. However, as fluxes are
MI-quantities, this temperature cannot be compared with the SI-temperature in
the MNEQT.
\end{enumerate}

There is a recent attempt \cite{Lucia} to introduce another NEQ\ temperature
by using fluctuation theorems to determine the entropy generation, which is
then related to the Gouy-Stodola theorem derived later (see Eq.
(\ref{Irrev-Work-Applied-DelF-Del_i_S})). It is limited to a interacting
$\Sigma$ in a medium $\widetilde{\Sigma}$ so does not apply to an isolated
system. In addition, its validity is limited to the situation when the
Gouy-Stodola theorem is valid as seen from the derivation of Eq.
(\ref{Irrev-Work-Applied-DelF-Del_i_S}). \ \ \ \ \ \ \ \ \ \ \ \ \ \ \ \ \ 

\subsection{Uniqueness of $S_{\text{ieq}}(\mathbf{Z})$ and $T$%
\label{Sec-Unique-S-T}}

We now give an alternative demonstration of the uniqueness of the entropy of
$\mathcal{M}_{\text{ieq}}$ in $\mathfrak{S}_{\mathbf{Z}}$, which is based on
the discussion of the internal variables in Sect. \ref{Sec-InternalVariables}.
Let us assume that we divide $\Sigma$ into a finite number of nonoverlapping
EQ subsystems $\left\{  \Sigma_{i}\right\}  $ such that $\cup_{i}\Sigma
_{i}=\Sigma$. Without loss of generality, we assume that the subsystems are
not in EQ with each other (their fields are not identical) so that $\Sigma$ is
in a NEQ macrostate. Let $\lambda_{\text{corr}}^{(i)}$ denote the correlation
length of $\Sigma_{i}$, and we define $\lambda_{\text{corr}}=\max\left\{
\lambda_{\text{corr}}^{(i)}\right\}  $ to denote the maximum correlation
length determining quasi-independence required for entropy additivity as
discussed in Sect. \ref{Sec-Notation}; see Eq. (\ref{Entropy-Additivity}). For
this, we need to take the linear size $\Delta l_{i}\gtrsim\lambda
_{\text{corr}}$ of $\Sigma_{i}$. The EQ microstate $\mathcal{M}_{\text{eq}%
}^{(i)}$\ of $\Sigma_{i}$ is uniquely described in $\mathfrak{S}_{\mathbf{X}}%
$. The additivity of entropy gives $S_{\text{ieq}}$ that must be a function of
$\left\{  \mathbf{X}_{i}\right\}  $. Moreover, since each $\mathcal{M}%
_{\text{eq}}^{(i)}$ has a unique entropy $S_{i}(\mathbf{X}_{i})$,
$S_{\text{ieq}}$ also has a unique value%
\begin{subequations}
\begin{equation}
S_{\text{ieq}}(\left\{  \mathbf{X}_{i}(t)\right\}  )=%
{\textstyle\sum\nolimits_{i}}
S_{i}(\mathbf{X}_{i}(t)). \label{S-Additivity-1}%
\end{equation}
As we need to express $S_{\text{ieq}}$ in terms of $\mathbf{X}(t)=%
{\textstyle\sum\nolimits_{i}}
\mathbf{X}_{i}(t)$, we need additional independent linear combinations
$\boldsymbol{\xi}(t)=\cup_{i}\boldsymbol{\xi}_{i}(t)$ made from the set
$\left\{  \mathbf{X}_{i}\right\}  $ as already discussed in Sect.
\ref{Sec-InternalVariables} to ensure that $S(\mathbf{Z}(t))$ depends on the
\emph{same} number $n^{\ast}$ of state variables as there are in
$S_{\text{ieq}}(\left\{  \mathbf{X}_{i}(t)\right\}  )$. This \emph{uniquely}
defines%
\begin{equation}
S(\mathbf{Z}(t))=%
{\textstyle\sum\nolimits_{i}}
S_{i}(\mathbf{X}_{i}(t)) \label{S-Additivity-2}%
\end{equation}
in $\mathfrak{S}_{\mathbf{Z}}$ in terms of the unique values$\ S_{i}%
(\mathbf{X}_{i}(t))$. It is a mathematical identity between the left side for
$\Sigma_{\text{B}}$ and the right side for $\Sigma_{\text{C}}$. We can also
take $\Sigma_{i}$'s to be in $\mathcal{M}_{i\text{ieq}}$'s in $\mathfrak{S}%
_{\mathbf{Z}_{i}}$'s so that $S_{i}(\mathbf{Z}_{i}(t))$ are also uniquely
defined. Then, the same reasoning as above also proves that
\begin{equation}
S(\mathbf{Z}(t))=%
{\textstyle\sum\nolimits_{i}}
S_{i}(\mathbf{Z}_{i}(t)) \label{S-Additivity-3}%
\end{equation}
is unique by ensuring that the number of arguments $n^{\ast}$ are the same on
both sides between the entropies for $\Sigma_{\text{B}}$ and $\Sigma
_{\text{C}}$.

We now prove the following central theorem on the existence of NEQ entropy for
any $\mathcal{M}_{\text{ieq}}$ with $n^{\ast}$ independent state variables.
\end{subequations}
\begin{theorem}
\label{Theorem-Exixtence-S} Existence: The SI-entropy $S(\mathbf{Z}(t))$ for
any $\mathcal{M}_{\text{ieq}}$ exists, has a unique thermodynamic temperature
$T$, and is additive in $\mathfrak{S}_{\mathbf{Z}}$.
\end{theorem}

\begin{proof}
According to the postulates of classical thermodynamics, EQ entropies exist in
$\mathfrak{S}_{\mathbf{X}}$ and are continuous. Therefore, Eq.
(\ref{S-Additivity-1}) proves the existence of the entropy $S(\mathbf{Z}(t))$
for any $\mathcal{M}_{\text{ieq}}~$in the state space $\mathfrak{S}%
_{\mathbf{Z}}$, and is continuous. It follows from the existence of
$S(\mathbf{Z}(t))$ that $\mathcal{M}_{\text{ieq}}~$has a unique thermodynamic
temperature $T$. In addition, $S(\mathbf{Z}(t))$ is also additive as follows
from Eq. (\ref{S-Additivity-3}). This proves the theorem.
\end{proof}

\begin{corollary}
\label{Corollary-ProperSubspaces} The state space $\mathfrak{S}_{\mathbf{Z}}$
contains $\mathfrak{S}_{\mathbf{X}}$ as a proper subspace because of the
presence of the internal variables, except when $\Sigma$ is in EQ, when they
become the same.
\end{corollary}

\begin{proof}
For any $\mathcal{M}_{\text{ieq}}$, $\mathfrak{S}_{\mathbf{Z}}$ contains all
possible linear combinations of $\boldsymbol{\xi}(t)$ made from the set
$\left\{  \mathbf{X}_{i}\right\}  $. Hence, it contains $\mathfrak{S}%
_{\mathbf{X}}$ as a proper subspace. In EQ, internal variables become
superfluous as they are no longer independent so their affinity vanishes.
Thus, $\mathfrak{S}_{\mathbf{Z}}$ reduces to $\mathfrak{S}_{\mathbf{X}}$ in EQ.
\end{proof}

\subsection{Irreversibility Inequalities in $\mathcal{M}_{\text{ieq}}%
$\label{Sec-IrreversibleInequalities}}

We consider the Hamiltonian $\mathcal{H}(\left.  \mathbf{x}\right\vert
\mathbf{w},\boldsymbol{\xi})$ in $\mathfrak{S}_{\mathbf{Z}}$. We only consider
the case of extensive macrowork parameters. As $\mathfrak{m}_{k}$ evolves
under the variation in $\mathbf{W}$, its energy $E_{k}$ changes by
$dE_{k}=-dW_{k}\ $without changing $p_{k}$; see Eq.
(\ref{HamiltonianChange-StateSpace}). The change determines the isentropic
generalized macrowork $dW=\mathbf{F}_{\text{w}}\cdot d\mathbf{W}%
=-dE_{\text{m}}$. The stochasticity appears from the generalized macroheat
$dQ=dE_{\text{s}}=TdS$. Recalling that for $\widetilde{\Sigma}$,
$T=T_{0},\mathbf{f}_{\text{w}0}=(P_{0},\cdots),\mathbf{A}_{0}=0$, we have in
general,
\begin{subequations}
\label{d-exch}%
\begin{align}
d_{\text{e}}W  &  =-d\widetilde{W}=\mathbf{f}_{\text{w}0}\cdot d_{\text{e}%
}\mathbf{w=}P_{0}dV+\cdots,\label{deW}\\
d_{\text{e}}Q  &  =-d\widetilde{Q}=T_{0}d_{\text{e}}S, \label{deQ}%
\end{align}
where the missing terms in the top equation refires to other elements in
$\mathbf{w}$. The irreversible macrowork $d_{\text{i}}W\doteq dW-d_{\text{e}%
}W$ due to the thermodynamic macroforce\emph{ }$\Delta\mathbf{F}^{\text{w}}$
has been given in Eq. (\ref{Irrev-Work}).

Using $d_{\text{e}}Q\ $in $dQ$, we find
\end{subequations}
\begin{subequations}
\begin{equation}
d_{\text{i}}Q=TdS-T_{0}d_{\text{i}}S=\left\{
\begin{array}
[c]{c}%
(T-T_{0})d_{\text{e}}S+Td_{\text{i}}S\\
(T-T_{0})dS+T_{0}d_{\text{i}}S
\end{array}
\right.  \geq0. \label{diQ}%
\end{equation}
Equating this with $d_{\text{i}}W$ from Eq. (\ref{Irrev-Work}), we obtain for
the irreversible entropy generation
\begin{equation}
d_{\text{i}}S=\left\{
\begin{array}
[c]{c}%
\left\{  (T_{0}-T)d_{\text{e}}S+\Delta\mathbf{F}^{\text{w}}\cdot
d\mathbf{W}\right\}  /T\\
\left\{  (T_{0}-T)dS+\Delta\mathbf{F}^{\text{w}}\cdot d\mathbf{W}\right\}
/T_{0}%
\end{array}
\right.  \geq0; \label{diS}%
\end{equation}
see Eq. (\ref{Irrev-Work}) for $\Delta\mathbf{F}^{\text{w}}\cdot d\mathbf{W}$.
Each term on the right side must be \emph{nonnegative} for the second law to
be valid. Thus, in terms of $\Delta F^{\text{h}}=T_{0}-T$, we see that the
first term%
\end{subequations}
\begin{subequations}
\begin{equation}
\Delta F^{\text{h}}d_{\text{e}}S\geq0, \label{ClausiusStatement}%
\end{equation}
in the first equation, which proves the Clausius statement of the macroheat
flow from "hot" to "cold," thus making sure that $T$ indeed can be thought of
as a "thermodynamic temperature" of the entire system, even if the latter is
inhomogeneous. This is the requirement C4 for a thermodynamic temperature.
Another important consequence of the second law comes from the first term in
the second equation \cite{Gujrati-I}:%
\begin{equation}
\Delta F^{\text{h}}dS\geq0. \label{EntropyInequality}%
\end{equation}
Similarly, the second term results in the inequality%
\begin{equation}
d_{\text{i}}W\equiv\Delta\mathbf{F}^{\text{w}}\cdot d\mathbf{W}\geq0
\label{WorkInequality}%
\end{equation}
due to macroforce imbalance, and consists of three separate inequalities%
\begin{equation}
(\mathbf{f}_{\text{w}}-\mathbf{f}_{\text{w}0})\cdot d_{\text{e}}\mathbf{w}%
\geq0,\mathbf{f}_{\text{w}}\cdot d_{\text{i}}\mathbf{w}\geq0,\mathbf{A\cdot
}d\boldsymbol{\xi}\geq0. \label{WorkInequality123}%
\end{equation}
This thus proves the inequality for $d_{\text{i}}W$ in Eq. (\ref{Irrev-Work}).
Using the inequality for $d_{\text{i}}W$ in $d_{\text{i}}Q=d_{\text{i}}W$ also
proves the inequality for $d_{\text{i}}Q$ in Eq. (\ref{diQ}). All these
inequalities help drive the system towards EQ in accordance with the second
law. We summarize the result in the following corollary.
\end{subequations}
\begin{corollary}
\label{Cor-PositiveIrreversibleWork}The irreversible macrowork $d_{\text{i}%
}W(t)$ or macroheat $d_{\text{i}}Q(t)$ is \emph{nonnegative}.
\end{corollary}

\begin{proof}
From Eq. (\ref{WorkInequality}), we find that
\begin{equation}
d_{\text{i}}W(t)=d_{\text{i}}Q(t)>0\text{ \ for }T>0 \label{diW_positive}%
\end{equation}
in accordance with the second law. \ 
\end{proof}

For example, if $d_{\text{i}}W$ corresponds to the irreversible macrowork done
by pressure imbalance only (so that we omit the last term in $d_{\text{i}}W$
in Eq. (\ref{MI-Work})), then%
\begin{equation}
d_{\text{i}}W(t)=(P(t)-P_{0})dV(t)>0\text{.}
\label{Irreversible_Pressure-Work}%
\end{equation}
If the system's pressure $P(t)>P_{0}$, the pressure of the medium, the volume
of the system increases so that $dV(t)>0$. In the opposite case, $dV(t)<0$. In
both cases, $d_{\text{i}}W(t)>0$ out of equilibrium. When $d_{\text{i}}W(t)$
consists of several independent contributions, each contribution must be
nonnegative in accordance with the second law and Corollary
\ref{Cor-PositiveIrreversibleWork}. The significance of the irreversible
macrowork in Eq. (\ref{Irreversible_Pressure-Work}) has been discussed in
Refs. \cite{Gujrati-Heat-Work0,Gujrati-Heat-Work}, where it is shown that this
macrowork results in raising the kinetic energy of the center-of-mass of the
surface separating $\Sigma$ and $\widetilde{\Sigma}$ by $dK_{\text{S}}$ and
overcoming macrowork $dW_{\text{fr}}$ done by all sorts of viscous or
frictional drag. Because of the stochasticity associated with any statistical
system, both energies dissipate among the particles in the system and appear
in the form of macroheat $d_{\text{i}}Q(t)$.

In the absence of any heat exchange ($d_{\text{e}}S=0$) or for an isothermal
system ($T=T_{0}$),$\ $we have
\begin{equation}
d_{\text{i}}Q=Td_{\text{i}}S=d_{\text{i}}W,
\label{ClosedIsothermalSystem-InternalHeat}%
\end{equation}
where $d_{\text{i}}W$ is given by Eq. (\ref{Irrev-Work}).

\subsection{Internal Variables and the Isolated System}

The above formulation of MNEQT is perfectly suited for considering an isolated
system $\Sigma$ ($d_{\text{e}}W=d_{\text{e}}Q\equiv0$) so that Eq.
(\ref{diQ-diW-EQ}) or $d_{\text{i}}E=0$ in Eq. (\ref{E-and-V-Partition})
becomes the most important thermodynamic equality. For an isolated system,
$d\mathbf{X}=0$ so that $d_{\text{i}}W=\mathbf{A}\cdot d\boldsymbol{\xi}$.

\begin{theorem}
\label{Theorem-diW-diS-Isolated}The irreversible entropy generated within an
isolated system is still related to the dissipated macrowork performed by the
internal variables.
\end{theorem}

\begin{proof}
As $E$ remains fixed for an isolated system ($dQ=Td_{\text{i}}S$), we have
from Eq. (\ref{FirstLaw-SI})%
\begin{equation}
d_{\text{i}}Q=Td_{\text{i}}S=d_{\text{i}}W=\mathbf{A}\cdot d\boldsymbol{\xi
}\geq0 \label{EntropyDiff-Isolated}%
\end{equation}
in accordance with the second law.
\end{proof}

Note that the above equation, though it is identical to Eq.
(\ref{ClosedIsothermalSystem-InternalHeat}) in form, is very different in that
$d_{\text{i}}W$ here is simply $\mathbf{A}\cdot d\boldsymbol{\xi}$. Same
conclusion is also obtained when we apply Eq. (\ref{diS}) to an isolated system.

\begin{corollary}
\label{Corollary-IsolatedSystem}Neither the entropy can increase nor will
there be any dissipated work unless some internal variables are present in an
isolated system. If no internal variables are used to describe an isolated
system, then thermodynamics requires it to be in EQ.
\end{corollary}

\begin{proof}
The proof follows trivially from Eq. (\ref{EntropyDiff-Isolated}).
\end{proof}

\subsection{Dissipation and Thermodynamic Forces}

As the inequality $d_{\text{i}}W\geq0$, see Eqs. (\ref{Irrev-Work}) and
(\ref{WorkInequality}), or $\Delta_{\text{i}}W\geq0$, see
(\ref{Isolated-Heat-Work}), for the irreversible macrowork for $\mathcal{M}%
_{\text{ieq}}$ in $\mathfrak{S}_{\mathbf{Z}}$ follows from the second law, it
is natural to identify it as the \emph{dissipation} or the \emph{dissipated
work}; recall that $\Delta_{\text{i}}W$ is obtained by integrating
$d_{\text{i}}W$ in Eq. (\ref{Irrev-Work}) over $\mathcal{P}$%
\begin{equation}
\Delta_{\text{i}}W=%
{\textstyle\int\nolimits_{\mathcal{P}}}
\left[  (\mathbf{f}_{\text{w}}-\mathbf{f}_{\text{w}0})\cdot d_{\text{e}%
}\mathbf{w}+\mathbf{f}_{\text{w}}\cdot d_{\text{i}}\mathbf{w}+\mathbf{A\cdot
}d\boldsymbol{\xi}\right]  . \label{Delta_i-W-Full}%
\end{equation}

\begin{definition}
\label{Def-DissipatedWork}The irreversible macrowork $d_{\text{i}}W\geq0$ or
$\Delta_{\text{i}}W\geq0$ for $\mathcal{M}_{\text{ieq}}$ belonging to
$\mathfrak{S}_{\mathbf{Z}}$ along $\mathcal{P}$, is identified as the
\emph{dissipation} or the \emph{dissipated work }in the MNEQT\emph{. }
\end{definition}

The definition is applicable regardless of $\mathbf{Z}$, and has contributions
from macroforce imbalance in $\mathfrak{S}_{\mathbf{Z}}$ as given in Eq.
(\ref{work macroforce}) at each point in $\mathcal{P}$. All microstates along
the path $\gamma_{\mathcal{P}}$ of $\mathcal{P}$ denote IEQ-macrostates
$\mathcal{M}_{\text{ieq}}$ belonging to $\mathfrak{S}_{\mathbf{Z}}$. In this
sense, the definition is a generalization of the definition of the \emph{lost
work} $\Delta W_{\text{lost}}$ in the \r{M}NEQT
\cite{Kestin,Woods,Prigogine,Bejan,Landau} in an irreversible process
$\overline{\mathcal{P}}$\ between EQ macrostates $\mathcal{A}_{\text{eq}}$ and
$\mathcal{B}_{\text{eq}}$ to any process $\mathcal{P}$ containing
$\mathcal{M}_{\text{ieq}}$ between $\mathcal{A}_{\text{ieq}}$ and
$\mathcal{B}_{\text{ieq}}$. The overbar in $\overline{\mathcal{P}}$\ is for EQ
macrostates $\mathcal{A}_{\text{eq}}$ and $\mathcal{B}_{\text{eq}}$. The lost
work is well known in the \r{M}NEQT; see for example, p. 12 in Woods
\cite{Woods} or Sect. 20 in Landau and Lifshitz \cite{Landau}. As
$\Delta_{\text{i}}S$ does not directly appear in the \r{M}NEQT, $\Delta
W_{\text{lost}}$, which is given by%
\begin{equation}
\Delta W_{\text{lost}}=\left(  \Delta_{\text{e}}W\right)  _{\text{rev}}%
-\Delta_{\text{e}}W, \label{Irrev-Work-Applied}%
\end{equation}
where $\left(  \Delta_{\text{e}}W\right)  _{\text{rev}}$ is the exchange work
during the reversible process $\overline{\mathcal{P}}_{\text{rev}}$ associated
with $\overline{\mathcal{P}}$, is used to determine $\Delta_{\text{i}}S$
indirectly as we now explain. We take $\widetilde{\Sigma}$ $=\widetilde
{\Sigma}^{\prime}\cup\widetilde{\Sigma}_{\text{w}}^{\prime\prime}$, where
$\widetilde{\Sigma}^{\prime}$ at constant $T_{0},P_{0}$ is thermally insulated
from another working medium $\widetilde{\Sigma}_{\text{w}}^{\prime\prime}$,
with $\Sigma_{0}=\Sigma\cup\widetilde{\Sigma}$. Let $\Delta_{\text{e}%
}Q^{\prime}$ and $\Delta_{\text{e}}W^{\prime}$ be the exchange macroquantities
from $\widetilde{\Sigma}^{\prime}$, which are well defined, and $\Delta
_{\text{e}}W^{\prime\prime}=-\Delta_{\text{e}}\widetilde{W}^{\prime\prime}$
the exchange macrowork from $\widetilde{\Sigma}_{\text{w}}^{\prime\prime}$. We
will closely follow Landau and Lifshitz \cite[where $\Delta_{\text{e}%
}\widetilde{W}^{\prime\prime}$ is denoted by $R$ and $d_{\text{e}}%
\widetilde{W}^{\prime\prime}$ by $dR$]{Landau}. We first consider an
infinitesimal process $\delta\overline{\mathcal{P}}$. In the \r{M}NEQT,%
\[
dE=d_{\text{e}}Q^{\prime}-d_{\text{e}}W^{\prime}-d_{\text{e}}W^{\prime\prime
},
\]
so that $d_{\text{e}}\widetilde{W}^{\prime\prime}=dE-T_{0}d_{\text{e}%
}S^{\prime}+P_{0}dV=dE-T_{0}dS+P_{0}dV+T_{0}d_{\text{i}}S$, where we have used
$dS=d_{\text{e}}S^{\prime}+d_{\text{i}}S$. We thus have%
\[
d_{\text{e}}\widetilde{W}^{\prime\prime}=dG-SdT_{0}+VdP_{0}+T_{0}d_{\text{i}%
}S,
\]
from which we obtain for $\overline{\mathcal{P}}$%
\[
\Delta_{\text{e}}\widetilde{W}^{\prime\prime}=\Delta G-%
{\textstyle\int\nolimits_{\overline{\mathcal{P}}}}
SdT_{0}+%
{\textstyle\int\nolimits_{\overline{\mathcal{P}}}}
VdP_{0}+%
{\textstyle\int\nolimits_{\overline{\mathcal{P}}}}
T_{0}d_{\text{i}}S,
\]
which is the generalization of the known result in \cite{Landau} to an
arbitrary process in $\mathfrak{S}_{\mathbf{X}}$. For fixed and constant
$T_{0}$ and $P_{0}$, the first two integrals vanish and the third integral
reduces to $T_{0}\Delta_{\text{i}}S$ over $\overline{\mathcal{P}}$. Thus, as
the minimum of $\Delta_{\text{e}}\widetilde{W}^{\prime\prime}$\ is given by
$\Delta G$, we derive the result in \cite{Landau}:%
\[
\Delta_{\text{e}}\widetilde{W}^{\prime\prime}-\Delta G=T_{0}\Delta_{\text{i}%
}S~\text{for fixed }T_{0},P_{0}.
\]
Using $\Delta_{\text{e}}W=-(\Delta_{\text{e}}\widetilde{W}^{\prime}%
+\Delta_{\text{e}}\widetilde{W}^{\prime\prime})$, and recognizing that
$\left(  \Delta_{\text{e}}W\right)  _{\text{min}}=\left(  \Delta_{\text{e}%
}W\right)  _{\text{rev}}=-\Delta F$, we have proved not only Eq.
(\ref{Irrev-Work-Applied}) but also%
\begin{subequations}
\begin{equation}
\Delta W_{\text{lost}}=-\Delta F-\Delta_{\text{e}}W.
\label{Irrev-Work-Applied-DelF}%
\end{equation}
We now show that we obtain the same result in the MNEQT, where we assume that
the temperature of $\Sigma$\ remains equal to $T_{0}$\ as assumed by Landau
and Lifshitz \cite{Landau}, use $dE=T_{0}dS-dW$, and recognize that
$d_{\text{e}}W=d_{\text{e}}W^{\prime}+d_{\text{e}}W^{\prime\prime}$. Comparing
it with the $dE$ in the \r{M}NEQT above, we immediately obtain%
\begin{equation}
\Delta_{\text{i}}W=\Delta W-\Delta_{\text{e}}W=T_{0}\Delta_{\text{i}}S\text{.}
\label{Irrev-Work-Applied-DelF-Del_i_S}%
\end{equation}
Thus, both theories give the same result in this simple example. But our
general expression for $d_{\text{i}}W$ or $\Delta_{\text{i}}W$\ is not
restricted to EQ terminal macrostates $\mathcal{A}_{\text{eq}}$ and
$\mathcal{B}_{\text{eq}}$ of $\overline{\mathcal{P}}$; they refer to any two
end macrostates $\mathcal{A}_{\text{ieq}}$ and $\mathcal{B}_{\text{ieq}}$ of
any arbitrary process $\mathcal{P}$. The procedure described by Landau and
Lifshitz \cite{Landau} or by Woods \cite{Woods}\ is not general enough to make
$\Delta W_{\text{lost}}$ useful in all cases.

We now turn to our approach and relate dissipation with the entropy generation
$d_{\text{i}}S$ for $\Sigma$ in Eq. (\ref{diS}). The strategy is simple. We
use Eq. (\ref{diQ-diW-EQ}) and express $d_{\text{i}}Q$ using Eq. (\ref{diQ}).
Let us use the top equation, which gives%
\end{subequations}
\begin{subequations}
\begin{equation}
Td_{\text{i}}S=\frac{(T_{0}-T)}{T_{0}}d_{\text{e}}Q+d_{\text{i}}W\geq0.
\label{Irreversible EntropyGeneration-Complete0}%
\end{equation}
For $\mathbf{W}=(V,\xi)$, it reduces to%
\begin{equation}
Td_{\text{i}}S=\frac{(T_{0}-T)}{T_{0}}d_{\text{e}}Q+(P-P_{0})dV+Ad\xi\geq0;
\label{Irreversible EntropyGeneration-Complete}%
\end{equation}
see, for example, Ref. \cite{Prigogine}. The first term in both equations is
due to macroheat exchange $d_{\text{e}}Q$ with $\widetilde{\Sigma}$ at
different temperatures, which is not considered part of dissipation as we have
defined above.

It is clear that the root cause of dissipation is a "\emph{force imbalance}"
$P(t)-P_{0},A(t)-A_{0}\equiv A(t)$, etc.
\cite{Kestin,Woods,Gujrati-Heat-Work0,Gujrati-Heat-Work,Gujrati-I,Gujrati-II,Gujrati-III,Gujrati-Entropy1,Gujrati-Entropy2,Gujrati-GeneralizedWork}
between the external and the internal forces performing macrowork, giving rise
to an internal macrowork $d_{\text{i}}W$ due to all kinds of force imbalances
in $\Delta\mathbf{F}^{\text{w}}$, which is not properly captured by
$d\widetilde{W}-dF$ in the \r{M}NEQT in all cases as discussed above. The
force imbalance are commonly known as \emph{thermodynamic forces} driving the
system towards equilibrium.

The irreversible macrowork is present even if there is no temperature
difference such as in an isothermal process as long as there exists some
nonzero thermodynamic force. The resulting irreversible entropy generation is
then given by $Td_{\text{i}}S=d_{\text{i}}W\geq0$; see Eq. (\ref{diS}). We
summarize this as a conclusion \cite{Prigogine}:
\end{subequations}
\begin{conclusion}
\label{Conclusion-ThermodynamicForce-Irreversibility}To have dissipation, it
is necessary and sufficient to have a nonzero thermodynamic force. In its
absence, there can be no dissipation.
\end{conclusion}

We now prove one of the central results in the MNEQT in the following theorem.

\begin{theorem}
\label{Th-TemperatureCriteria}The proportionality parameter $T$ in Eq.
(\ref{dQ-dS}) or (\ref{System_dQ_dS}) satisfies all the criteria (C1-C4) of a
sensible temperature. Therefore, we identify $T$ as the temperature of the
system in any arbitrary macrostate $\mathcal{M}_{\text{arb}}$.
\end{theorem}

\begin{proof}
As $dQ$ and $dS$ scale the same way with the size of $\Sigma$, $T$ is an
intensive quantity. When the entropy is a state function in $\mathfrak{S}%
_{\mathbf{Z}}$ or $\mathfrak{S}_{\mathbf{Z}^{\prime}}\supset\mathfrak{S}%
_{\mathbf{Z}}$, we have a Gibbs fundamental relation given in Eq.
(\ref{Gibbs-FR-S}). So the temperature is defined by a derivative in
$\mathfrak{S}_{\mathbf{Z}}$ or $\mathfrak{S}_{\mathbf{Z}^{\prime}}$, the
latter giving $T_{\text{arb}}$. This shows that C1 is satisfied for any
$\mathcal{M}_{\text{arb}}$. As we have not imposed any restrictions on the
signs of $dQ(t)$ and $dS(t)$, the parameter $T(t)$ can be of any sign, which
shows that C2 is satisfied. To demonstrate consistency with the second law, we
rewrite the top equation in Eq. (\ref{diS}) to express $d_{\text{i}}S$ as a
sum of two independent contributions
\begin{align}
d_{\text{i}}S^{\text{Q}}(t)  &  \doteq(1/T-1/T_{0})d_{\text{e}}%
Q(t),\label{diS-Qe}\\
d_{\text{i}}S^{\text{W}}(t)  &  \doteq d_{\text{i}}W(t)/T\equiv d_{\text{i}%
}Q(t)/T,\nonumber
\end{align}
so that
\begin{equation}
d_{\text{i}}S(t)=d_{\text{i}}S^{\text{Q}}(t)+d_{\text{i}}S^{\text{W}}(t);
\label{IrreversibleEntropy}%
\end{equation}
here, $d_{\text{i}}S^{\text{Q}}(t)$ is generated solely by exchange macroheat
$d_{\text{e}}Q(t)$ at different temperatures, and $d_{\text{i}}S^{\text{W}%
}(t)$ by the irreversible macrowork or macroheat $d_{\text{i}}W(t)\equiv
d_{\text{i}}Q(t)$. The two contributions are \emph{independent} of each other.
Accordingly, both contributions individually must be nonnegative in accordance
with the second law. In particular, the inequality%
\begin{equation}
d_{\text{i}}S^{\text{Q}}(t)\geq0. \label{HeatFlowDirection}%
\end{equation}
For $T(t)>T_{0}$, $d_{\text{e}}Q(t)<0$ so that the macroheat flows from the
system to the medium. For $T(t)<T_{0}$, $d_{\text{e}}Q(t)>0$ so that the
macroheat flows from the medium to the system. This establishes that C3 is
satisfied. As $dQ(t)$ and $dS(t)$ are global quantities, the parameter $T(t)$
is also a global parameter, which means that C4 is also satisfied. This proves
the theorem.
\end{proof}

Because of the importance of C4, we give many example in Sect.
\ref{Sec-Applications} to justify that $T$ acts as a global temperature of the
system even if it is composite with different temperatures. These examples
leave no doubt that C4 is satisfied.

We conclude this subsection by considering a special case, also studied by
Landau and Lifshitz \cite[Sect. 13 and specifically Eq. (13.4)]{Landau}: it
deals with the irreversibility generated \emph{only by macroheat exchange}
\emph{at different temperatures but no internal (macrowork)\ dissipation}. It
follows from Eq. (\ref{IrreversibleEntropy}) that $d_{\text{i}}W=0$ in this
case, even though there is irreversibility ($d_{\text{i}}S(t)>0$) due to the
macroheat exchange. If there are internal variables also, then $\mathbf{A=A}%
_{0}=0$ to ensure $d_{\text{i}}W=0$. This example is important in that it
shows that just because there is irreversibility in the system, we do not have
$d_{\text{i}}W(t)\equiv d_{\text{i}}Q(t)\neq0$. We see from Eq.
(\ref{IrreversibleEntropy}) that $d_{\text{i}}S(t)=d_{\text{i}}S^{\text{Q}%
}(t)$, which can be rewritten, using $d_{\text{e}}Q(t)=T_{0}d_{\text{e}}S(t)$,
as
\begin{equation}
d_{\text{e}}Q(t)=T(t)dS(t)=T_{0}d_{\text{e}}S(t), \label{deQ-dS}%
\end{equation}
a result also derived by Landau and Lifshitz; note that they use $dQ$ for
$d_{\text{e}}Q$. For $d_{\text{i}}Q=0$, we have $dQ(t)=d_{\text{e}}Q(t)$ so
that Eq. (\ref{deQ-dS}) is consistent with Eq. (\ref{dQ-dS}), as it must. In
the presence of nonzero $d_{\text{i}}Q$, Eq. (\ref{deQ-dS}) gets
modified:\ one must subtract $d_{\text{i}}Q$ from the right side.

\subsection{Cyclic Process \label{Sec-CyclicProcess}}

For a general body that is not isolated, the concept of its internal
equilibrium state plays a very important role in that the body can come back
to this macrostate $\mathcal{M}_{\text{ieq}}$ several times in a
nonequilibrium process. In a cyclic nonequilibrium process, such a macrostate
can repeat itself in time after some cycle time $\tau_{\text{c}}$ so that all
state variables and functions including the entropy repeat themselves:%
\[
\mathbf{Z}(t+\tau_{\text{c}})=\mathbf{Z}(t),~\mathcal{M}(t+\tau_{\text{c}%
})=\mathcal{M}(t),~S(t+\tau_{\text{c}})=S(t).
\]
This ensures
\[
\Delta_{\text{c}}S\equiv S(t+\tau_{\text{c}})-S(t)=0
\]
in a cyclic process. All that is required for the cyclic process to occur is
that the body must start and end in the same internal equilibrium state;
however, during the remainder of the cycle, the body need not be in internal equilibrium.\ 

The same argument also applies to a cyclic process that starts and returns to
$\mathcal{M}_{\text{eq}}$ after some cycle time $\tau_{\text{c}}$. However,
the body need not be in EQ macrostates during the rest of the cycle. We will
consider such a case when we consider a NEQ\ Carnot cycle in Sec.
\ref{Sec-CarnotCycle}.

\subsection{Steady State\label{Sec-SteadyState}}

Consider a system between two different media as shown in Fig.
\ref{Fig-Sys-TwoSources}. For example, we can consider $\Sigma$ between two
heat baths $\widetilde{\Sigma}_{\text{h}}$ and $\widetilde{\Sigma}_{\text{h}}$
replacing the two media in the figure. We will study this example in the MNEQT
in Sec. \ref{Sec-Composite System} using a composite system $\Sigma$ between
the two heat sources. In the presence of two media, it is possible for
$\Sigma$ to reach a steady state, in which it satisfies conditions similar to
that for a cyclic process above in terms of MI-macroquantities:%
\begin{subequations}
\begin{equation}
d\mathbf{Z}=0,~dS=0, \label{SteadyState-Conditions}%
\end{equation}
where the changes are over the system $\Sigma$. The above conditions in the
MNEQT lead to important relations between exchange and irreversible
macroquantities:%
\begin{equation}
d_{\text{i}}\mathbf{Z}=-d_{\text{e}}\mathbf{Z},~d_{\text{i}}S=-d_{\text{e}}S;
\label{SteadyState-Relations}%
\end{equation}
as usual, the irreversible contributions satisfy the second law inequalities.
For $E$, we have from Eq. (\ref{FirstLaw-SI})%
\end{subequations}
\[
dQ=dW=0,
\]
which follows from $dQ=TdS=0$ for $\mathcal{M}_{\text{ieq}}$ in $\mathfrak{S}%
_{\mathbf{Z}}$ or for $\mathcal{M}_{\text{nieq}}$ in $\mathfrak{S}%
_{\mathbf{Z}^{\prime}}$. Therefore, in the MNEQT,%
\[
d_{\text{i}}Q=-d_{\text{e}}Q\geq0,d_{\text{i}}W=-d_{\text{e}}W\geq0.
\]
As a consequence,
\[
d_{\text{e}}Q=d_{\text{e}}W\leq0,
\]
a result that cannot be derived in the \r{M}NEQT by using the first law in Eq.
(\ref{FirstLaw-MI}).

It should be noted, as said earlier in Sect. \ref{Sec-Introduction}, that the
steady state occurs only over a short period $\tau\sim\tau_{\text{st}}$
compared to the time $\tau_{\text{EQ}}$ required for the two media to
equilibrate with each other. The latter time period is extremely large
compared to $\tau_{\text{st}}$ because of their extreme sizes. For a time
period longer than $\tau_{\text{st}}$, the steady state cannot be treated as
steady as $\Sigma$ will begin the equilibrium process between them so that
eventually at $\tau_{\text{EQ}}$, $d_{\text{i}}S$\ will vanish as
$d_{\text{e}}Q\rightarrow0$. We will not consider this possibility here, but
can be studied in the MNEQT.

\subsection{Intrinsic Adiabaticity Theorem}

We now have a clear statement of the generalization of the adiabatic theorem
\cite{Landau} for nonequilibrium processes going on in a body in an arbitrary
macrostate in terms of the intrinsic quantity $dS$. We will call it the
\emph{intrinsic adiabatic theorem}.

\begin{definition}
\label{Intrinsic Adiabatic Process}Intrinsic Adiabatic Process: \emph{An
intrinsic adiabatic process is an isentropic process (}$dS(t)=0$\emph{ and not
necessarily }$d_{\text{e}}S(t)=0$\emph{). }
\end{definition}

Such a process also includes the stationary limit, i.e. the steady macrostate
of a non-equilibrium process discussed in the previous section. However, the
extension goes beyond the conventional notion of an adiabatic process commonly
dealt with in the \r{M}NEQT, according to which an adiabatic process
\cite{Landau} is one for which $d_{\text{e}}Q(t)=0$, which is equivalent to
$d_{\text{e}}S(t)=0$. If $d_{\text{i}}S(t)=0$, it also represents a reversible
process in a thermally isolated system so that $d_{\text{e}}Q(t)=0$. One can
also have $dS(t)=0$ in an irreversible process during which%
\begin{equation}
d_{\text{i}}S(t)=-d_{\text{e}}S(t)>0; \label{Stationary_S}%
\end{equation}
as usual, Eq. (\ref{diQ-diW-EQ}) always remains satisfied. If the system is in
a $\mathcal{M}_{\text{arb}}$, then we must also have%
\begin{equation}
d_{\text{i}}Q(t)=-d_{\text{e}}Q(t)=T_{0}d_{\text{i}}S(t)>0;
\label{Stationary _Q}%
\end{equation}
recall Eq. (\ref{System_dQ_dS}) for a $\mathcal{M}_{\text{arb}}$.

\begin{theorem}
\label{Adiabatic Theorem} In an intrinsic adiabatic process, the sets of
microstates and of their probabilities $p_{k}$ do not change, but
$d_{\text{e}}p_{k}=-d_{\text{i}}p_{k}\neq0$ for all $k$.
\end{theorem}

\begin{proof}
In terms $d_{\text{e}}p_{k}$ and $d_{\text{i}}p_{k}$, Eqs. (\ref{Stationary_S}%
) and (\ref{Stationary _Q}) become
\begin{subequations}
\label{Adiabatic_entropy_heat}%
\begin{align}
\sum_{k}\eta_{k}d_{\text{i}}p_{k}  &  =-\sum_{k}\eta_{k}d_{\text{e}}%
p_{k},\label{adiabatic_entropy}\\
\sum_{k}E_{k}d_{\text{i}}p_{k}  &  =-\sum_{k}E_{k}d_{\text{e}}p_{k}.
\label{adiabatic_heat}%
\end{align}
Recognizing that there is only macrowork in $dE$, which requires $p_{k}$ not
to change, we conclude that
\end{subequations}
\[
dp_{k}=0\text{ \ for }\forall k
\]
in an adiabatic process. As $d_{\text{i}}S(t)$ does not vanish in an
irreversible process, $d_{\text{i}}p_{k}(t)$ cannot vanish. Accordingly,
$d_{\text{e}}p_{k}=-d_{\text{i}}p_{k}\neq0,\forall k$ for an irreversible
adiabatic process. The conditions in Eq. (\ref{p_alpha-sum}) remain valid as
expected. As $p_{k}$'s do not change, no microstate can appear or disappear.
This proves the theorem.
\end{proof}

\section{Clausius Equality \label{Sec_Clausius_Equality}}

We recall that Eq. (\ref{dQ-dS}), which we call the \emph{Clausius equality},
follows from the Gibbs fundamental equation for $\mathcal{M}_{\text{ieq}}$ in
$\mathfrak{S}_{\mathbf{Z}}$ or $\mathcal{M}_{\text{nieq}}$ in $\mathfrak{S}%
_{\mathbf{Z}^{\prime}}$, see Eq. (\ref{Gibbs-FR-E}). It is merely is a
consequence of the state function $S$ for a $\mathcal{M}_{\text{ieq}}$ or
$\mathcal{M}_{\text{nieq}}$ in respective state spaces so the equality is also
valid for any\ $\mathcal{M}_{\text{arb}}$, see Eq. (\ref{System_dQ_dS}). Here,
we are only concerned with some $\mathcal{M}_{\text{ieq}}$. The equality is
very interesting, and should be contrasted with the \emph{Clausius inequality}%
\begin{equation}
d_{\text{e}}Q\leq T_{0}dS. \label{ClausiusInequality}%
\end{equation}
First, it follows from Eq. (\ref{dQ-dS}) that $dQ/T$ is nothing but the
\emph{exact differential }$dS$ for $\mathcal{M}_{\text{ieq}}$\ so that
\begin{equation}
\oint dQ(t)/T(t)\equiv0 \label{Clausius_Equality_0}%
\end{equation}
for$\ $any cyclic process; here we have added the time argument for clarity.
It is only because of the use of $dQ(t)$ in place of $d_{\text{e}}Q(t)$\ that
the Clausius inequality has become an equality. The equality should not be
interpreted as the absence of irreversibility ($\Delta_{\text{i}}S>0$) as is
clear from Eq. (\ref{Irreversible_entropy_Cycle}) obtained by using
$d_{\text{i}}S(t)\equiv dS(t)-d_{\text{e}}S(t)$for a cyclic process taking
time $\tau$:
\begin{equation}
N(t,\tau)\equiv%
{\textstyle\oint}
d_{\text{i}}S(t)=-%
{\textstyle\oint}
d_{\text{e}}Q(t)/T_{0}\geq0, \label{Irreversible_entropy_Cycle}%
\end{equation}
which is the second law for a cyclic process, and represents the irreversible
entropy generated in a cycle. This is the original Clausius inequality. The
quantity $N(t)$ is the \emph{uncompensated transformation} of Clausius
\cite{Prigogine} that is directly related to $d_{\text{i}}S(t)$ \cite{Eu}; in
contrast, $N_{0}(t,\tau)$
\begin{equation}
N_{0}(t,\tau)\equiv%
{\textstyle\oint}
d_{\text{i}}Q(t)/T(t)\equiv%
{\textstyle\oint}
d_{\text{i}}W(t)/T(t)\geq0, \label{Irreversible_entropy_Cycle-0}%
\end{equation}
where we have used the fundamental identity in Eq. (\ref{diQ-diW-EQ}), is
determined by the irreversible\emph{ }macroheat $d_{\text{i}}Q(t)$ or the
irreversible\emph{ }macrowork $d_{\text{i}}W(t)$, and represents a different
quantity as is evident. In terms of the two macroheats, we have
\begin{equation}%
{\textstyle\oint}
d_{\text{i}}Q(t)/T(t)=-%
{\textstyle\oint}
d_{\text{e}}Q(t)/T(t)\geq0, \label{Clausius_Inequality}%
\end{equation}
which results in a new Clausius inequality\emph{ }%
\begin{equation}%
{\textstyle\oint}
d_{\text{e}}Q(t)/T(t)\leq0; \label{Clausius_Inequality_0}%
\end{equation}
compare with the original Clausius inequality in Eq.
(\ref{Irreversible_entropy_Cycle}).

\section{Extended State Space and $\mathcal{M}_{\text{nieq}}$}

\subsection{Choice of $\mathbf{Z}$\label{Sec-Choice-S_Z}}

We first discuss how to choose a particular state space\ for a unique
description of a macrostate $\mathcal{M}$ depending on the experimental setup.
To understand the procedure for this, we begin by considering a set
$\boldsymbol{\xi}_{n}$ of internal variables\emph{ }$\left(  \xi_{1},\xi
_{2},\cdots,\xi_{n}\right)  $ and $\mathbf{Z}_{n}\doteq\mathbf{X}%
\cup\boldsymbol{\xi}_{n}$ to form a sequence of state spaces $\mathfrak{S}%
_{\mathbf{Z}}^{(n)}$. In general, one may need many internal variables, with
the value of $n$ increasing as $\mathcal{M}$ is more and more out of EQ
\cite{Gujrati-Hierarchy} relative to $\mathcal{M}_{\text{eq}}$. We will take
$n^{\ast}$ to be the maximum $n$\ in this study, even though $n$ $<<n^{\ast}$
needed for $\mathfrak{S}_{\mathbf{Z}}^{(n)}$ will usually be a small number in
most cases. We refer to Sect. \ref{Sec-Unique-S-T}, where the choice of
$n^{\ast}$ is determined by the mathematical identity in Eq.
(\ref{S-Additivity-2}) in $\mathfrak{S}_{\mathbf{Z}}^{(n^{\ast})}$. The two
most important but distinct time scales are $\tau_{\text{obs}}$, the time to
make observations, and $\tau_{\text{eq}}$, the equilibration time for a
macrostate $\mathcal{M}$ to turn into $\mathcal{M}_{\text{eq}}$. For
$\tau_{\text{obs}}<\tau_{\text{eq}}$, the system will be in a NEQ macrostate.
Let $\tau_{i}\ $denote the relaxation time of $\xi_{i}$ needed to come to its
equilibrium value so that its affinity $A_{i}\rightarrow0$
\cite{deGroot,Prigogine,Gujrati-Hierarchy,Prigogine0,Gutzow,Nemilov}. For
convenience, we order $\xi_{i}$ so that
\[
\tau_{1}>\tau_{2}>\cdots;
\]
we assume distinct $\tau_{i}$'s for simplicity without affecting our
conclusions. For $\tau_{1}<\tau_{\text{obs}}$, all internal variables have
equilibrated so they play no role in equilibration except thermodynamic forces
$T-T_{0},P-P_{0}$, etc. associated with $\mathbf{X}$ that still drive the
system towards EQ. We choose $n$ satisfying $\tau_{n}>\tau_{\text{obs}}%
>\tau_{n+1}$ so that all of $\xi_{1},\xi_{2},\cdots,\xi_{n}$ have not
equilibrated (their affinities are nonzero). They play an important role in
the NEQT, while $\xi_{n+1},\xi_{n+2},\cdots$ need not be considered as they
have all equilibrated. This specify $\mathcal{M}$ \emph{uniquely} in
$\mathfrak{S}_{\mathbf{Z}}^{(n)}$, which was earlier identified as in IEQ.

Note that NEQ macrostates with $\tau_{n+1}>\tau_{\text{obs}}>\tau_{n+2}$ are
not uniquely identifiable in $\mathfrak{S}_{\mathbf{Z}}^{(n)}$, even though
they are uniquely identifiable in $\mathfrak{S}_{\mathbf{Z}}^{(n+1)}$. Thus,
there are many NEQ macrostates that are not unique in $\mathfrak{S}%
_{\mathbf{Z}}^{(n)}$. The unique macrostates $\mathcal{M}_{\text{ieq}}$\ are
special in that its Gibbs entropy $S(\mathbf{Z}_{n})$\ is a state function of
$\mathbf{Z}_{n}$ in $\mathfrak{S}_{\mathbf{Z}}^{(n)}$. Thus, given
$\tau_{\text{obs}}$, we look for the window $\tau_{n}>\tau_{\text{obs}}%
>\tau_{n+1}$ to choose the particular value of $n$. This then determines
$\mathfrak{S}_{\mathbf{Z}}^{(n)}$ in which the macrostates are in IEQ. From
now onward, we assume that $n$ has been found and $\mathfrak{S}_{\mathbf{Z}%
}^{(n)}$ has been identified. We now suppress $n$ and simply use
$\mathfrak{S}_{\mathbf{Z}}$ below.

\begin{remark}
\label{Remark-Small number-Subsystems} The linear sizes of various subsystems
introduced in Sect. \ref{Sec-Unique-S-T} must be larger than the correlation
length $\lambda_{\text{corr}}$ as discussed elsewhere \cite{Gujrati-II} for
the first time,\ and briefly revisited in Sect. \ref{Sec-Notation} to ensure
entropy additivity; see also Sect. \ref{Sec-Conclusions}. Therefore, it is
usually sufficient to take the linear size of $\Sigma$ to be a small multiple
(for example, $10$ to $20$) of the correlation length to obtain a proper
thermodynamics, which is extensive. This means that we will usually need a
theoretically manageable but small number of internal variables $n$.
\end{remark}

\subsection{Microstate probabilities for $\mathcal{M}_{\text{ieq}}%
$\label{Sec-MicrostateProbabilities}}

As $\mathcal{M}_{\text{ieq}}$ is unique in $\mathfrak{S}_{\mathbf{Z}}$, we
need to identify the unique set $\left\{  p_{k}\right\}  $. If we keep
$\mathbf{W}$ fixed in $\mathcal{M}_{\text{ieq}}$ as the parameter, then
$\mathbf{F}_{\text{w}k}$ are fluctuating microforces in $\mathfrak{S}%
_{\mathbf{Z}}$ as we have seen in Sect. (\ref{Sec_Stat_Concepts}). In
additions, we have microstate energies $E_{k}$ also fluctuating. We need to
maximize the entropy $S(\mathbf{Z})$ at fixed
\[
E=\sum_{k}E_{k}p_{k},\mathbf{F}_{\text{w}}=\sum_{k}\mathbf{F}_{\text{w}k}p_{k}%
\]
by varying $p_{k}$ without changing $\left\{  \mathfrak{m}_{k}\right\}  $ ,
i.e. $E_{k}$ and $\mathbf{F}_{\text{w}k}$. This variation has nothing to do
with $dp_{k}$\ in a physical process. Using the Lagrange multiplier technique,
it is easy to show that the condition for this in terms of three Lagrange
multipliers with obvious definitions is
\begin{equation}
\eta_{k}=\lambda_{1}+\lambda_{2}E_{k}+\boldsymbol{\lambda}_{3}\cdot
\mathbf{F}_{\text{w}k}, \label{index_i}%
\end{equation}
from which follows the statistical entropy $\mathcal{S}=-(\lambda_{1}%
+\lambda_{2}E+\boldsymbol{\lambda}_{3}\cdot\mathbf{F}_{\text{w}})$; we have
reverted back to the original symbol for the statistical entropy here. It is
now easy to identify $\lambda_{2}=-\beta,\boldsymbol{\lambda}_{3}%
=-\beta\mathbf{W}$ by comparing $d\mathcal{S}$ with $dS$ in Eq.
(\ref{Gibbs-FR-S}) by varying $E$ and $\mathbf{W}$ so we finally have
\begin{equation}
p_{k}=\exp[\beta(\widehat{G}-E_{k}-\mathbf{W}\cdot\mathbf{F}_{\text{w}k})],
\label{microstate probability}%
\end{equation}
where $\lambda_{1}=\beta\widehat{G}$\ with $\widehat{G}(t)$ is a normalization
constant and defines a NEQ partition function
\begin{subequations}
\begin{equation}
\exp(-\beta\widehat{G})\equiv\sum_{k}\exp[-\beta(E_{k}+\mathbf{W}%
\cdot\mathbf{F}_{\text{w}k})]. \label{NEQ-PF}%
\end{equation}
It is easy to verify that
\begin{equation}
\widehat{G}(T,\mathbf{W})=E+\mathbf{W}\cdot\mathbf{F}_{\text{w}}-TS,
\label{NEQ-Potential}%
\end{equation}
so that if we neglect the fluctuations $E_{k}-E$ and $\mathbf{F}_{\text{w}%
k}-\mathbf{F}_{\text{w}}$ and replace $E_{k}$ by $E$ and $\mathbf{F}%
_{\text{w}k}$ by $\mathbf{F}_{\text{w}}$, then $p_{k}$ reduces to the flat
distribution $p_{k}=1/W(E,\mathbf{W})=\exp[\beta(\widehat{G}-E-\mathbf{W}%
\cdot\mathbf{F}_{\text{w}})]=\exp(-S)$ in Remark \ref{Remark-FlatDistribution}%
, which can be identified as the microstate probability in the NEQ
microcanonical ensemble. It should be stressed that this is consistent with
the well-known fact that thermodynamics does not describe fluctuations; the
latter require using statistical mechanics \cite{Landau}.

It should be remarked that the Lagrange multipliers in $p_{k}$ are determined
by comparing the resulting entropy to match exactly the Gibbs fundamental
relation, a thermodynamic relation. This then proves that $\mathcal{S}$ is the
same as the thermodynamic entropy $S$ up to a constant \cite{Gujrati-Entropy2}%
, which can be fixed by appeals to the third law, according to which $S$
vanishes at absolute zero. We avoid considering here the issue of a residual
entropy, which is discussed elsewhere
\cite{Gujrati-ResidualEntropy,Gujrati-Hierarchy}. The $p_{k}$ above clearly
shows the effect of irreversibility and is very different from its equilibrium
analog $p_{k}^{\text{eq}}$
\end{subequations}
\[
p_{k}^{\text{eq}}=\exp[\beta_{0}(\widehat{G}(T_{0},\mathbf{w})-E_{k}%
-\mathbf{w}\cdot\mathbf{f}_{\text{w}k})],
\]
see Eq. (\ref{GeneralizedForce-W}), obtained by replacing $\mathbf{W}$ by
$\mathbf{w}$, $\mathbf{F}_{\text{w}k}$ by $\mathbf{f}_{\text{w}0k}$,\ and
$\beta$ by $\beta_{0}$. The fluctuating $E_{k},\mathbf{f}_{\text{w}k}$ satisfy%
\[
E=\sum_{k}E_{k}p_{k}^{\text{eq}},\mathbf{f}_{\text{w}0}=\sum_{k}%
\mathbf{f}_{\text{w}0k}p_{k}^{\text{eq}}.
\]

The observation time $\tau_{\text{obs}}$ is determined by the way $T$ and
$\mathbf{W}$ are changed during a process. Thus, during each change,
$\tau_{\text{obs}}$ must be compared with the time needed for $\Sigma$ to come
to the next IEQ macrostate, and for the microstate probabilities to be given
by Eq. (\ref{microstate probability}) with the new values of $T$ and
$\mathbf{W}$.

\subsection{$\mathcal{M}_{\text{nieq}}$ in $\mathfrak{S}_{\mathbf{Z}}$
\label{Sec-M_nieq}}

We now focus on a non-unique macrostate $\mathcal{M}_{\text{nieq}}$ in
$\mathfrak{S}_{\mathbf{Z}}$. This will be needed if $\tau_{\text{obs}}$ is
reduced to make the process faster so that instead of falling in the window
($\tau_{n},\tau_{n+1}$), it now falls in a higher window such as ($\tau
_{n+1},\tau_{n+2}$)$.$ As said above, $\mathcal{M}$ can now be treated as a
unique macrostate in a larger state space $\mathfrak{S}_{\mathbf{Z}^{\prime}%
}\supset\mathfrak{S}_{\mathbf{Z}}$. Let $\boldsymbol{\xi}^{\prime}(t)$ denote
the set of additional internal variables needed over $\mathfrak{S}%
_{\mathbf{Z}}$ so that
\[
\mathbf{Z}^{\prime}(t)=(\mathbf{Z}(t),\boldsymbol{\xi}^{\prime}(t)).
\]
The entropy $S(\mathbf{Z}^{\prime}(t))=S(\mathbf{Z}(t),t)$ for $\mathcal{M}%
_{\text{ieq}}(t)$\ in $\mathfrak{S}_{\mathbf{Z}^{\prime}}$ satisfies the Gibbs
fundamental relation%
\[
dS(\mathbf{Z}^{\prime}(t))\mathbf{=}\frac{\partial S}{\partial E}%
dE\mathbf{+}\frac{\partial S}{\partial\mathbf{W}}\cdot d\mathbf{W+}%
\frac{\partial S}{\partial\boldsymbol{\xi}^{\prime}}\cdot d\boldsymbol{\xi
}^{\prime},
\]
where $\mathbf{W}$ is the work variable in $\mathfrak{S}_{\mathbf{Z}}$.
Expressing the last term as%
\[
\frac{\partial S}{\partial\boldsymbol{\xi}^{\prime}}\cdot\frac
{d\boldsymbol{\xi}^{\prime}}{dt}dt,
\]
we obtain the following generalization of the Gibbs fundamental relation for
$\mathcal{M}_{\text{nieq}}(t)$\ in $\mathfrak{S}_{\mathbf{Z}}$:%
\begin{subequations}
\begin{equation}
dS(\mathbf{Z}(t),t)=\frac{\partial S}{\partial E}dE\mathbf{+}\frac{\partial
S}{\partial\mathbf{W}}\cdot d\mathbf{W+}\frac{\partial S}{\partial t}dt,
\label{GibbsFR-S-NIEQ}%
\end{equation}
where%
\begin{equation}
\frac{\partial S}{\partial t}\doteq\frac{\partial S}{\partial\boldsymbol{\xi
}^{\prime}}\cdot\frac{d\boldsymbol{\xi}^{\prime}}{dt}\geq0. \label{dS/dt-NIEQ}%
\end{equation}
In $\mathfrak{S}_{\mathbf{Z}^{\prime}}$, we can identify the temperature $T$
as the thermodynamic temperature in $\mathfrak{S}_{\mathbf{Z}^{\prime}}$ by
the standard definition. But, it is clear from the above discussion that
$\partial S(\mathbf{Z}^{\prime}(t))/\partial E$ in $\mathfrak{S}%
_{\mathbf{Z}^{\prime}}$ has the same value as $\partial S(\mathbf{Z}%
(t),t)/\partial E$ in $\mathfrak{S}_{\mathbf{Z}}$. Therefore, we are now set
to identify $T_{\text{arb}}$ in Eq. (\ref{System_dQ_dS}) as a thermodynamic temperature.
\end{subequations}
\begin{remark}
$T_{\text{arb}}$ in Eq. (\ref{System_dQ_dS}) in $\mathfrak{S}_{\mathbf{Z}}%
$\ is identified by the same derivative in the Gibbs fundamental relation in
$\mathfrak{S}_{\mathbf{Z}^{\prime}}$ as follows
\begin{equation}
\beta_{\text{arb}}=1/T_{\text{arb}}=\partial S(\mathbf{Z}^{\prime
}(t))/\partial E=\partial S(\mathbf{Z}(t),t)/\partial E. \label{beta_arb}%
\end{equation}

\end{remark}

\begin{definition}
\label{Def-HiddenEntropy}As the presence of $\partial S/\partial t$ above in
$\mathfrak{S}_{\mathbf{Z}}$ is due to "hidden" internal variables in
$\boldsymbol{\xi}^{\prime}$, we will call it the \emph{hidden entropy
generation rate}, and
\begin{subequations}
\begin{equation}
d_{\text{i}}S^{\text{hid}}(t)=\frac{\partial S}{\partial t}dt=\frac{\partial
S}{\partial\boldsymbol{\xi}^{\prime}}\cdot d\boldsymbol{\xi}^{\prime}\geq0,
\label{hiddenn-entropy-generation}%
\end{equation}
the \emph{hidden entropy generation}. It results in a \emph{hidden
irreversible macrowork}%
\begin{equation}
d_{\text{i}}W^{\text{hid}}\doteq Td_{\text{i}}S^{\text{hid}}=\mathbf{A}%
^{\prime}\cdot d\boldsymbol{\xi}^{\prime}, \label{hiddenn-irreversible-work}%
\end{equation}
in $\mathfrak{S}_{\mathbf{Z}}$ due to the hidden internal variable with
affinity $\mathbf{A}^{\prime}$.
\end{subequations}
\end{definition}

\begin{remark}
\label{Remark-IEQ-ARB-Macrostate}A macrostate $\mathcal{M}_{\text{nieq}}(t)$
with $\mathcal{S}(\mathbf{Z}(t),t)$ can be converted to $\mathcal{M}%
_{\text{ieq}}(t)$ with a state function $\mathcal{S}(\mathbf{Z}^{\prime}(t))$
in an appropriately chosen state space $\mathfrak{S}_{\mathbf{Z}^{\prime}%
}\supset\mathfrak{S}_{\mathbf{Z}}$ by finding the appropriate window in which
$\tau_{\text{obs}}$ lies. The needed additional internal variable
$\boldsymbol{\xi}^{\prime}$\ determines the hidden entropy generation rate
$\partial S/\partial t$ in Eq. (\ref{dS/dt-NIEQ}) due to the non-IEQ nature of
$\mathcal{M}_{\text{nieq}}(t)$\ in $\mathfrak{S}_{\mathbf{Z}}$, and ensures
validity of the Gibbs relation in Eq. (\ref{GibbsFR-S-NIEQ}) for it, thereby
providing not only a new interpretation of the temporal variation of the
entropy due to hidden variables but also extends the MNEQT to $\mathcal{M}%
_{\text{nieq}}(t)$\ in $\mathfrak{S}_{\mathbf{Z}}$.
\end{remark}

The above discussion strongly points towards the possible

\begin{proposition}
\label{Proposition-General-MNEQT}The MNEQT provides a very general framework
to study any $\mathcal{M}_{\text{nieq}}(t)$ in $\mathfrak{S}_{\mathbf{Z}}$,
since it can be converted into a $\mathcal{M}_{\text{ieq}}(t)$ in an
appropriately chosen state space $\mathfrak{S}_{\mathbf{Z}^{\prime}}$, with
$d_{\text{i}}S^{\text{hid}}(t)$ originating from hidden internal variable
$\boldsymbol{\xi}^{\prime}$.
\end{proposition}

\begin{remark}
In a process $\mathcal{P}$ resulting in $\mathcal{M}_{\text{nieq}}(t)$ in
$\mathfrak{S}_{\mathbf{Z}}$, it is natural to assume that the terminal
macrostates in $\mathcal{P}$ are $\mathcal{M}_{\text{ieq}}$ so the affinity
corresponding to $\boldsymbol{\xi}^{\prime}$ must vanish in them.
\end{remark}

\begin{remark}
\label{Remark-NonIEQ-S_X} By replacing $\mathbf{Z}$ by $\mathbf{X}$, and
$\mathbf{Z}^{\prime}$ by $\mathbf{Z}$, we can also express the Gibbs
fundamental relation for any NEQ macrostate in $\mathfrak{S}_{\mathbf{X}}$ as%
\begin{equation}
dS(\mathbf{X}(t),t)=\frac{\partial S}{\partial E}dE\mathbf{+}\frac{\partial
S}{\partial\mathbf{w}}\cdot d\mathbf{w+}\frac{\partial S}{\partial t}dt,
\label{Gibbs-FR-NonIEQ-S_X}%
\end{equation}
by treating $\mathcal{M}$ as $\mathcal{M}_{\text{ieq}}$ in $\mathfrak{S}%
_{\mathbf{Z}}$. In a NEQ process $\overline{\mathcal{P}}$ between two
EQ\ macrostates but resulting in $\mathcal{M}_{\text{ieq}}(t)$ between them in
$\mathfrak{S}_{\mathbf{Z}}$, the affinity corresponding to $\boldsymbol{\xi}$
must vanish in the terminal EQ macrostates of $\overline{\mathcal{P}}$.
\end{remark}

Eq. (\ref{Gibbs-FR-NonIEQ-S_X}) proves extremely useful to describe
$\mathcal{M}$ in $\mathfrak{S}_{\mathbf{X}}$ as it may not be easy to identify
$\boldsymbol{\xi}$ in all cases.

\begin{remark}
\label{Remark-dQ-dS-GeneralRelation} The explicit time dependence in the
entropy for $\mathcal{M}_{\text{neq}}$ in $\mathfrak{S}_{\mathbf{X}}$ or
$\mathcal{M}_{\text{nieq}}(t)$\ in $\mathfrak{S}_{\mathbf{Z}}$ is solely due
to the internal variables, which do not affect $dQ=dE_{\text{s}}$, Eq.
(\ref{dQ-dS}) remains valid, with $T$ defined as the inverse of $\partial
S/\partial E$ at fixed $\mathbf{w},t$\ or $\mathbf{W},t$ in the two state
spaces, respectively; see Eq. (\ref{T-beta}).
\end{remark}

\subsection{External and Internal Variations of $dp_{k}(t)$%
\label{Sec_External_Internal_Variations}}

We focus on $\mathcal{N}_{k}$ in Sect. \ref{Marker_NonEq-S}, and partition its
change $d\mathcal{N}_{k}$ in accordance with the micropartition rule; see
Definition \ref{Def-Macropartition}. We take $\mathcal{N}$ fixed. Then, the
macropartition results in the partition for $dp_{k}$ given in Eq.
(\ref{dpk-Partition}), where $d_{\text{e}}p_{k}$ is the change due to
exchanges with the medium and $d_{\text{i}}p_{k}$ the change due to internal
processes. It follows from the partition that%
\begin{equation}
d_{\text{e}}Q(t)\equiv%
{\textstyle\sum\nolimits_{k}}
E_{k}d_{\text{e}}p_{k}(t),\ d_{\text{i}}Q(t)\equiv%
{\textstyle\sum\nolimits_{k}}
E_{k}d_{\text{i}}p_{k}(t), \label{Statistical_Heat_Components}%
\end{equation}
where we have replaced $d$ by $d_{\alpha}$ in Eq. (\ref{dE_dQ}). As
$d_{\alpha}Q(t)$ are thermodynamic quantities, they must not change their
values if we change $E_{k}$ by adding a constant to $\mathcal{H}$. This
requires
\begin{equation}%
{\textstyle\sum\nolimits_{k}}
d_{\alpha}p_{k}(t)=0,\forall\alpha, \label{p_alpha-sum}%
\end{equation}
and put a limitation on the possible variations $d_{\alpha}p_{k}$. We will
assume this to be true. Using this fact, we similarly have%
\begin{equation}
d_{\text{e}}S(t)\equiv-%
{\textstyle\sum\nolimits_{k}}
\eta_{k}d_{\text{e}}p_{k}(t),\ d_{\text{i}}S(t)\equiv-%
{\textstyle\sum\nolimits_{k}}
\eta_{k}d_{\text{i}}p_{k}(t).\ \label{Statistical_Entropy_Components}%
\end{equation}
The relation $d_{\text{e}}Q(t)=T_{0}d_{\text{e}}S(t)$ can be expressed in
terms of $d_{\text{e}}p_{k}(t)$%
\[%
{\textstyle\sum\nolimits_{k}}
(\eta_{k}-\beta_{0}E_{k})d_{\text{e}}p_{k}=0.
\]
Similarly, the relation $d_{\text{i}}Q(t)=T(t)dS(t)-T_{0}d_{\text{e}}S(t)$ can
be written as
\[%
{\textstyle\sum\nolimits_{k}}
(E_{k}-T_{0}\eta_{k})d_{\text{i}}p_{k}=(T(t)-T_{0})%
{\textstyle\sum\nolimits_{k}}
\eta_{k}dp_{k},
\]
which acts as a constraint on possible variations $d_{\text{i}}p_{k}$, given
that $dp_{k}$ can be directly obtained from Eq. (\ref{microstate probability}).\ 

\section{A Model Entropy Calculation\label{Sec-EntropyCalculation}}

We consider a gas of non-interacting identical structureless particles with no
spin, each of mass $m$, in a fixed region confined by impenetrable walls
(infinite potential well). Initially, the gas is in a NEQ macrostate, and is
isolated in that region. In time, the gas will equilibrate and the microstate
probabilities change in a way that the entropy increases. We wish to
understand how the increase happens.

\subsection{1-dimensional ideal Gas:\ }

In order to be able to carry out an \emph{exact calculation}, we consider the
gas in a $1$-dimensional box of initial size $L_{\text{in}}$. As there are no
interactions between the particles, the wavefunction $\Psi$ for the gas is a
product of individual particle wavefunctions $\psi$. Thus, we can focus on a
single particle to study the nonequilibrium behavior of the gas
\cite{Gujrati-JensenInequality,Gujrati-QuantumHeat,GujTyler,GujBoyko,Wu}. The
simple model of a particle in a box has been extensively studied in the
literature but with a very different emphasis. \cite{Bender,Doescher,Stutz}
The particle only has non-degenerate eigenstates whose energies are determined
by $L$, and \ a quantum number $k$. We use the energy scale $\varepsilon
_{1}=\pi^{2}\hbar^{2}/2mL_{\text{in}}^{2}$ to measure the eigenstate energies,
and $\alpha=L/L_{\text{in}}$ so that%
\begin{equation}
\varepsilon_{k}(L)=k^{2}/\alpha^{2}; \label{particle-microstate energy}%
\end{equation}
the corresponding eigenfunctions are given by%
\begin{equation}
\psi_{k}(x)=\sqrt{2/L}\sin(k\pi x/L),\ \ k=1,2,3,\cdots.
\label{eigenfunctions}%
\end{equation}
The pressure generated by the eigenstate on the walls is given by
\cite{Landau-QM}
\begin{equation}
P_{k}(L)\equiv-\partial\varepsilon_{k}/\partial L=2\varepsilon_{k}(L)/L.
\label{particle-microstate presure}%
\end{equation}
In terms of the eigenstate probability $p_{k}(t)$, the average energy and
pressure are given by
\begin{subequations}
\label{particle energy pressure}%
\begin{align}
\varepsilon(t,L)  &  \equiv%
{\textstyle\sum\nolimits_{k}}
p_{k}(t)\varepsilon_{k}(L),\label{particle energy}\\
P(t,L)  &  \equiv%
{\textstyle\sum\nolimits_{k}}
p_{k}(t)P_{k}(L)=2\varepsilon(t,L)/L. \label{particle pressure0}%
\end{align}
The entropy follows from Eq. (\ref{Gibbs_Formulation}) and is given for the
single particle case by%
\end{subequations}
\[
s(t,L)\equiv-%
{\textstyle\sum\nolimits_{k}}
p_{k}(t)\ln p_{k}(t).
\]
The time dependence in $\varepsilon(t)$\ or $P(t)$ is due to the time
dependence in $p_{k}$ and $\varepsilon_{k}(L)$. Even for an isolated system,
for which $\varepsilon$ remains constant, $p_{k}$ cannot remain constant as
follows directly from the second law \cite{Gujrati-Symmetry} and creates a
conceptual problem because the eigenstates are mutually orthogonal and there
can be no transitions among them to allow for a change in $p_{k}$.\ 

As the gas is isolated, its energy, volume and the number of particles remain
constant. As it is originally not in equilibrium, it will eventually reach
equilibrium in which its entropy must increase. This requires the introduction
of some internal variables even in this system whose variation will give rise
to entropy generation by causing internal variations $d_{\text{i}}p_{k}(t)$ in
$p_{k}(t)$. Here, we will assume a single internal variable $\xi(t)$. What is
relevant is that the variation in $\xi(t)$ is accompanied by changes
$dp_{k}(t)$ occurring within the isolated system. According to our
identification of heat with changes in $p_{k}(t)$, these variations must be
associated with heat, which in this case will be associated with irreversible
heat $d_{\text{i}}Q(t)$.

\subsection{Chemical Reaction Approach\label{Sec-ChemReaction}}

A way to change $p_{k}$ in an isolated system is to require the presence of
some stochastic interactions, whose presence allows for transitions among
eigenstates \cite{Gujrati-Entropy1}. As these transitions are happening within
the system, we can treat them as "chemical reactions" between different
eigenstates \cite{DeDonder,deGroot,Prigogine} by treating each eigenstate $k$
as a chemical species. During the transition, these species undergo chemical
reactions to allow for the changes in their probabilities.

We follow this analogy further and extend the traditional approach
\cite{DeDonder,deGroot,Prigogine} to the present case. For the sake of
simplicity, our discussion will be limited to the ideal gas in a box; the
extension to any general system is trivial. Therefore, we will use microstates
$\left\{  \mathfrak{m}_{k}\right\}  $ instead of eigenstates in the following
to keep the discussion general. Let there be $N_{k}(t)$ particles in
$\mathfrak{m}_{k}$ at some instant $t$ so that
\[
N=%
{\textstyle\sum\nolimits_{k}}
N_{k}(t)
\]
at all times, and $p_{k}(t)=N_{k}(t)/N$. We will consider the general case
that also includes the case in which final microstates refer to a box size
$L^{\prime}$ different from its initial value $L$. Let us use $A_{k}$ to
denote the reactants (initial microstates) and $A_{k}^{\prime}$ to denote the
products (final microstates). For the sake of simplicity of argument, we will
assume that transitions between microstates is described by a single chemical
reaction, which is expressed in stoichiometry form as
\begin{equation}%
{\textstyle\sum\nolimits_{k}}
a_{k}A_{k}\longrightarrow%
{\textstyle\sum\nolimits_{k}}
a_{k}^{\prime}A_{k}^{\prime}. \label{General Reaction}%
\end{equation}
Let $N_{k}$ and $N_{k}^{\prime}$ denote the population of $A_{k}$ and
$A_{k}^{\prime}$, respectively, so that $N=%
{\textstyle\sum\nolimits_{k}}
N_{k}=%
{\textstyle\sum\nolimits_{k}}
N_{k}^{\prime}$. Accordingly, $p_{k}(t)=N_{k}(t)/N$ for the reactant and
$p_{k}(t+dt)=N_{k}^{\prime}(t)/N$ for the product. The single reaction is
described by a single extent of reaction $\xi$ and we have
\[
d\xi(t)\equiv-dN_{k}(t)/a_{k}(t)\equiv dN_{k^{\prime}}^{\prime}%
(t)/a_{k^{\prime}}^{\prime}(t)\text{ \ \ for all }k,k^{\prime}.
\]
It is easy to see that the coefficients satisfy an important relation%
\[%
{\textstyle\sum\nolimits_{k}}
a_{k}(t)=%
{\textstyle\sum\nolimits_{k}}
a_{k}^{\prime}(t),
\]
which reflects the fact that the change $\left\vert dN\right\vert $ in the
reactant microstates is the same as in the product microstates. The
\emph{affinity} in terms of the chemical potentials $\mu$ is given by%
\[
A(t)=%
{\textstyle\sum}
a_{k}(t)\mu_{A_{k}}(t)-%
{\textstyle\sum}
a_{k}^{\prime}(t)\mu_{A_{k}^{\prime}}(t),
\]
and will vanish only in "equilibrium," i.e. only when $p_{k}$' s attain their
equilibrium values. Otherwise, $A(t)$ will remain non-zero. It acts as the
thermodynamic force in driving the chemical reaction
\cite{DeDonder,deGroot,Prigogine}. But we must wait long enough for the
reaction to come to completion, which happens when $A(t)$ and $d\xi/dt$ both
vanish. The extent of reaction $\xi$ is an example of an internal variable.
There may be other internal variables depending on the initial NEQ macrostate.
This will be discussed in the following section.

\section{Simple Applications\label{Sec-Applications}}

\subsection{Isothermal Expansion}

Let us first consider an isothermal expansion of an ideal gas in which the
temperature $T$ of the gas remains constant and equal to that of the medium
$T_{0}$. During an irreversible isothermal expansion, energy is pumped into
the gas isothermally from outside so $E(t)$ remains constant. The pumping of
energy will result in the change $d_{\text{e}}p_{k}(t)$. This will determine
$d_{\text{e}}S(t)=d_{\text{e}}Q(t)/T_{0}$. In addition, the gas may undergo
transitions among various energy levels, as discussed in Sect.
\ref{Sec-ChemReaction}, without any external energy input, which will
determine the change $d_{\text{i}}p_{k}(t)$. From Eq. (\ref{diQ}), we
determine $d_{\text{i}}Q(t)=T_{0}d_{\text{i}}S(t)$, and consequently
$d_{\text{i}}W(t)$. Thus,%
\[
\lbrack P(t)-P_{0}]dV(t)+A(t)d\xi(t)=T_{0}dS(t)-d_{\text{e}}Q(t).
\]
Such a calculation will not be possible using the first law in Eq.
(\ref{FirstLaw-MI}) in the \r{M}NEQT.

\subsection{Intrinsic Adiabatic Expansion}

In a nonequilibrium intrinsic adiabatic process, we have $d_{\text{i}}W(t)=-$
$d_{\text{e}}Q(t)$ so the heat exchange $\left\vert d_{\text{e}}%
Q(t)\right\vert =$ $T_{0}\left\vert d_{\text{e}}S(t)\right\vert $ is converted
into the irreversible work. We can use this to determine the work
$d_{\text{i}}W_{\xi}(t)$ due to the single internal variable%
\[
A(t)d\xi(t)=-d_{\text{e}}Q(t)-(P(t)-P_{0})dV>0.
\]
The identification $d_{\text{i}}W(t)=-$ $d_{\text{e}}Q(t)$ and the calculation
of $A(t)d\xi(t)$ and of $d_{\text{i}}S(t)$\ cannot be done in the traditional
formulation of the first law in the \r{M}NEQT, in which $d_{\text{e}}Q(t)=0$
for the traditional adiabatic process so that $dE=-d_{\text{e}}W(t)$.

\subsection{Composite $\Sigma$ with Temperature Inhomogeneity.
\label{Sec-Composite System}}

Here, we will show by examples that the thermodynamic temperature $T$ of
$\Sigma$ allows us to treat it as a "black box" $\Sigma_{\text{B}}$ without
knowing its detailed internal structure such as its composition in terms of
two subsystems $\Sigma_{1}$ and $\Sigma_{2}$. Alternatively, we can treat
$\Sigma$ as a combination $\Sigma_{\text{C}}$ of $\Sigma_{1}$ at $T_{1}$\ and
$\Sigma_{2}$ at at $T_{2}<T_{1}$, and obtain same thermodynamics. Thus, both
approaches are equivalent, which justifies the usefulness and uniqueness (see
Theorem \ref{Theorem-Exixtence-S}) of $T$ as a thermodynamically appropriate
global temperature.

In the following, we will consider various cases that can be obtained as
special cases of the following general situation (\ref{Fig-Sys-TwoSources}):
$\Sigma_{1}$ in thermal contact with the medium $\widetilde{\Sigma}%
_{\text{h1}}$ at temperature $T_{01}$, and $\Sigma_{2}$ in thermal contact
with the medium $\widetilde{\Sigma}_{\text{h2}}$ at temperature $T_{02}$, with
the two media having no mutual interaction.

We will consider the two realizations for $\Sigma$: $\Sigma_{\text{B}}$ and
$\Sigma_{\text{C}}$ to compare their predictions. As discussed for the case
(b) in Sec. \ref{Sec-InternalVariables}, $\Sigma_{1}$ and $\Sigma_{2}$ are
always taken to be in EQ, but $\Sigma$ in IEQ. The entropies in the two
realizations are%
\begin{equation}
S_{\text{B}}(t)=S(E(t),\xi(t));S_{\text{C}}=S_{1}(E_{1}(t))+S_{2}(E_{2}(t)),
\label{S_BC}%
\end{equation}
and have the same value; recall that $E(t)=E_{1}(t)+E_{2}(t)$, and
$\xi(t)=E_{1}(t)-E_{2}(t)$ for $\Sigma(t)$; see Eq. (\ref{Internal Variable-1}%
). For clarity, we will often use the argument $t$ to emphasize the variations
in time $t$ in this section. In general, the irreversible entropy generation
is given by%
\begin{equation}
d_{\text{i}}S(t)=d\widetilde{S}_{1}(t)+d\widetilde{S}_{2}(t)+dS(t),
\label{Irreversible EntropyGeneration-Inhomogeneity}%
\end{equation}
where $dS$ should be replaced by $dS_{\text{B}}$ or $dS_{\text{C}}$ as the
case may be:%
\begin{equation}%
\begin{tabular}
[c]{l}%
$dS_{\text{B}}(t)=\beta(t)dE(t)+\beta(t)A(t)d\xi(t),$\\
$dS_{\text{C}}(t)=\beta_{1}(t)dE_{1}(t)+\beta_{2}(t)dE_{2}(t),$%
\end{tabular}
\ \ \ \ \ \label{dSbc}%
\end{equation}
where we are using the inverse temperatures for various bodies. Let
$d_{\text{e}}Q_{l}(t),l=1,2$ be the energy or macroheat transferred to
$\Sigma_{l}(t)$ from $\widetilde{\Sigma}_{\text{h}}^{(l)}$, and $dE_{\text{in}%
}(t)=d_{\text{e}}Q_{\text{in}}(t)$ the energy or macroheat transferred from
$\Sigma_{1}(t)$ to $\Sigma_{2}(t)$. We have, using $\delta_{1}=+1$ and
$\delta_{2}=-1,$
\begin{subequations}
\begin{align}
dE_{l}(t)  &  =d_{\text{e}}Q_{l}(t)+\delta_{l}dE_{\text{in}}(t),\nonumber\\
dE(t)  &  =d_{\text{e}}Q_{1}(t)+d_{\text{e}}Q_{2}(t),\label{dEl-dE-dSl}\\
d\widetilde{S}_{l}(t)  &  =-d_{\text{e}}S_{l}(t)=-\beta_{0l}d_{\text{e}}%
Q_{l}(t).\nonumber
\end{align}
We see that $dE(t)$ is unaffected by the internal energy transfer
$dE_{\text{in}}(t)$, while
\begin{equation}
d\xi(t)=d_{\text{e}}Q_{1}(t)-d_{\text{e}}Q_{2}(t)+2dE_{\text{in}}(t),
\label{General-dXi}%
\end{equation}
is affected by the macroheat exchange disparity $d_{\text{e}}Q_{1}%
(t)-d_{\text{e}}Q_{2}(t)$ along with $dE_{\text{in}}(t)$.

We finally have%
\end{subequations}
\begin{equation}
d_{\text{i}}S(t)=-%
{\textstyle\sum\nolimits_{l}}
\beta_{0l}d_{\text{e}}Q_{l}(t)+dS. \label{diS-Media}%
\end{equation}
We now consider various cases to make our point.

\subsubsection{Isolated $\Sigma$}

We first consider the realization $\Sigma_{\text{B}}$. Using $dE(t)=dE_{1}%
(t)+dE_{2}(t),d\xi(t)=dE_{1}(t)-dE_{2}(t)$, see Eqs.
(\ref{Internal Variable-1}), and (\ref{dSbc}) for $dS_{\text{B}}(t)$ above, we
obtain
\begin{subequations}
\begin{equation}
\beta(t)=\frac{\beta_{1}(t)+\beta_{2}(t)}{2},\beta(t)A(t)=\frac{\beta
_{1}(t)-\beta_{2}(t)}{2}. \label{beta-A-Composite}%
\end{equation}
This identifies $T\left(  t\right)  $ in terms of $T_{1}(t)$ and $T_{2}(t)$.
As EQ is attained, $T(t)\rightarrow T_{0}$, the EQ temperature between
$\Sigma_{1}$ and $\Sigma_{2}$, and $A(t)\rightarrow A_{0}=0$ as expected. In
the following, we will use $A^{\prime}(t)$ for $\beta(t)A(t)$ for simplicity.
In terms of $\beta$ and $A^{\prime}$, we also have%
\begin{equation}
\beta_{1}=\beta+A^{\prime},\beta_{2}=\beta-A^{\prime}. \label{beta1-beta2}%
\end{equation}

We now justify that in this simple example, $A^{\prime}(t)d\xi(t)$ determines
$d_{\text{i}}S(t)$ due to irreversibilty in $\Sigma(t)$; see Eq.
(\ref{ClosedIsothermalSystem-InternalHeat}). Setting $dE(t)=0$ in
$dS_{\text{B}}(t)$, we have by direct evaluation,
\end{subequations}
\begin{equation}
d_{\text{i}}S(t)=A^{\prime}(t)d\xi(t)=\beta(t)d_{\text{i}}W(t).
\label{diS-Isolated-Composite}%
\end{equation}
The last equation follows from the general result in Eq.
(\ref{EntropyDiff-Isolated}). It should be emphasized that the existence of
$d_{\text{i}}S(t)\geq0$ due to $\xi$ in $\mathcal{M}_{\text{ieq}}$ is
consistent with $\mathcal{M}_{\text{ieq}}$ as a NEQ macrostate, even though
its entropy is a state function in the extended state space.

We now consider $\Sigma_{\text{C}}$, which is also very instructive to
understand the origin of $d_{\text{i}}S(t)$ in a different way. Considering
internal energy or macroheat transfer $dE_{\text{in}}(t)=d_{\text{e}%
}Q_{\text{in}}(t)$ between $\Sigma_{1}(t)$ and $\Sigma_{2}(t)$ at some instant
$t$, we have
\begin{subequations}
\begin{equation}
dS_{1}(t)=\frac{dE_{\text{in}}(t)}{T_{1}(t)},dS_{2}(t)=-\frac{dE_{\text{in}%
}(t)}{T_{2}(t)}, \label{dS1-dS2-dEin}%
\end{equation}
due to this transfer. This results in
\begin{equation}
d_{\text{i}}S(t)=\left[  \beta_{1}(t)-\beta_{2}(t)\right]  dE_{\text{in}%
}(t)=A^{\prime}d\xi(t), \label{diS-dEin}%
\end{equation}
since $d\xi(t)=dE_{1}(t)-dE_{2}(t)=2dE_{\text{in}}(t)$. Thus, the physical
origin of $d_{\text{i}}S(t)$ is the internal entropy change of the subsystems.

\subsubsection{$\Sigma$ Interacting with $\widetilde{\Sigma}_{\text{h}}$}

To further appreciate the physical significance of the NEQ $T(t)$ of the above
composite system $\Sigma(t)$, we allow it to interact with $\widetilde{\Sigma
}_{\text{h}}$, a heat bath, at the EQ temperature $T_{0}$. For this, we take
$\widetilde{\Sigma}_{\text{h1}}$ and $\widetilde{\Sigma}_{\text{h2}}$ at the
same common temperature $T_{0}=T_{01}=T_{02}$ so that we can treat them as a
single medium $\widetilde{\Sigma}_{\text{h}}$ with macroheat exchange
$d_{\text{e}}Q(t)$. We thus obtain from Eq. (\ref{diS-Media})
\end{subequations}
\[
d_{\text{i}}S(t)=-\beta_{0}d_{\text{e}}Q(t)+dS.
\]
We will consider two different kinds of interaction below:

\qquad(i) We first consider $\Sigma_{\text{B}}(t)$ in $\mathcal{M}%
_{\text{ieq}}$ at $T(t)$ so we use $dS_{\text{B}}(t)$ above. We thus have
(using the identity $d_{\text{e}}S(t)=\beta_{0}d_{\text{e}}Q(t)$)
\begin{equation}
d_{\text{i}}S(t)=[\beta(t)-\beta_{0}]d_{\text{e}}Q(t)+A^{\prime}(t)d\xi(t),
\label{Single-Interacting diS}%
\end{equation}
which is consistent with the general identity given by the top equation in Eq.
(\ref{diQ}), a result which was derived for a single system at temperature
$T\left(  t\right)  $. This confirms that the composite $\Sigma_{\text{C}}$
here can be treated as a noncomposite $\Sigma_{\text{B}}$\ at $T\left(
t\right)  $. To be convinced that the above $d_{\text{i}}S(t)$ includes the
internally generated irreversibility in Eq. (\ref{diS-Isolated-Composite}) due
to macroheat transfer between $\Sigma_{1}(t)$ and $\Sigma_{2}(t)$, we only
have to set $d_{\text{e}}S(t)=0$ to ensure the isolation of $\Sigma$. We
reproduce Eq. (\ref{diS-Isolated-Composite}) as $d_{\text{i}}Q(t)=d_{\text{i}%
}W(t)$. The remaining source of irreversibility $T(t)d_{\text{i}}S^{\text{Q}%
}(t)$ given by the first term above is due to external macroheat
exchange\ between $\Sigma$ and $\widetilde{\Sigma}_{\text{h}}$%
\begin{subequations}
\begin{equation}
d_{\text{i}}S^{\text{Q}}(t)=[T_{0}\beta(t)-1]d_{\text{e}}S(t),
\label{diS_Sigma-s}%
\end{equation}
as expected; see the first term on the right in Eq.
(\ref{Irreversible EntropyGeneration-Complete}).

\qquad(ii) We take treat $\Sigma(t)$ as $\Sigma_{\text{C}}(t)$ in contact with
$\widetilde{\Sigma}_{\text{h}}$. We deal directly with the two macroheat
exchanges $d_{\text{e}}Q_{l}(t),l=1,2$ to $\Sigma_{l}(t)$ from $\widetilde
{\Sigma}_{\text{h}}$, and the internal energy transfer $dE_{\text{in}}(t)$.
Using $dE_{l}(t)$ from Eq. (\ref{dEl-dE-dSl}) in $dS_{\text{C}}$ given in Eq.
(\ref{diS-Media}), we find that%
\end{subequations}
\[
d_{\text{i}}S(t)=%
{\textstyle\sum\nolimits_{l}}
[\beta_{l}(t)-\beta_{0}]d_{\text{e}}Q_{l}(t)+\left[  \beta_{1}(t)-\beta
_{2}(t)\right]  dE_{\text{in}}(t).
\]
Using Eq. (\ref{beta1-beta2}) to express $\beta_{l}$, we can rewrite the above
equation as%
\[
d_{\text{i}}S(t)=[\beta(t)-\beta_{0}]d_{\text{e}}Q(t)+A^{\prime}d\xi,
\]
where we have used the identity
\begin{equation}
d_{\text{e}}Q(t)=d_{\text{e}}Q_{1}+d_{\text{e}}Q_{2}, \label{total=Exch-heat}%
\end{equation}
and have found
\begin{equation}
d\xi=d_{\text{e}}Q_{1}-d_{\text{e}}Q_{2}+2dE_{\text{in}} \label{internal-dxi}%
\end{equation}
using its general definition $d\xi(t)=dE_{1}(t)-dE_{2}(t)$. We thus see that
$d_{\text{i}}S(t)$ obtained by both realizations are the same as they must.
However, the realization $\Sigma_{\text{C}}(t)$ allows us to also identify
$d\xi$.\ 

Each exchange generates irreversible entropy following Eq. (\ref{diS_Sigma-s}%
). Using $d_{\text{e}}Q(t)=d_{\text{e}}Q_{1}(t)+d_{\text{e}}Q_{2}(t)$ in
$dQ(t)=T(t)dS(t)$ to determine $d_{\text{i}}Q(t)$, we find the generalization
of Eq. (\ref{Single-Interacting diS}):
\begin{equation}
d_{\text{i}}S(t)=\left[  \beta_{1}(t)-\beta_{2}(t)\right]  d_{\text{e}%
}Q_{\text{in}}(t)+%
{\textstyle\sum\nolimits_{l}}
[T_{0}\beta_{l}(t)-1]d_{\text{e}}S_{l}(t). \label{Composite-Interacting diS}%
\end{equation}
It is easy to see that the last term above gives nothing but the sum of the
irreversible entropies due to external exchanges of macroheat by $\Sigma
_{1}(t)$ and $\Sigma_{2}(t)$ with $\widetilde{\Sigma}_{\text{h}}$:%
\begin{equation}
d_{\text{i}}S^{\text{Q}}(t)=d_{\text{i}}S_{1}^{\text{Q}}(t)+d_{\text{i}}%
S_{2}^{\text{Q}}(t), \label{diS_S1-S2}%
\end{equation}
where
\begin{equation}
d_{\text{i}}S_{l}^{\text{Q}}(t)=[T_{0}\beta_{l}(t)-1]d_{\text{e}}%
S_{l}(t),l=1,2 \label{diS_Sel}%
\end{equation}
\ is the external entropy exchange of $\Sigma_{l}(t)$ with $\widetilde{\Sigma
}_{\text{h}}$.

Thus, whether we treat $\Sigma$ as a system $\Sigma_{\text{B}}$ at temperature
$T(t)$ or a collection $\Sigma_{\text{C}}$ of $\Sigma_{1}(t)$ and $\Sigma
_{2}(t)$ at temperatures $T_{1}(t)$ and $T_{2}(t)$, respectively, we obtain
the same irreversibility. In other words, $T(t)$ is a sensible thermodynamic
temperature even in the presence of inhomogeneity.

\subsection{$\Sigma$ Interacting with $\widetilde{\Sigma}_{\text{h1}}$ and
$\widetilde{\Sigma}_{\text{h2}}$}

We now consider our composite $\Sigma$ in thermal contact with two distinct
and mutually noninteracting stochastic media $\widetilde{\Sigma}_{\text{h1}}$
and $\widetilde{\Sigma}_{\text{h2}}$ at temperatures $T_{01}$ and $T_{02}$. We
will again discuss the two different realizations as above.

(i) We first consider $\Sigma_{\text{B}}(t)$ at temperature $T(t)$, which
interacts with the two $\widetilde{\Sigma}_{\text{h}}$'s, and use the general
result in Eq. (\ref{diS-Media}). A simple calculation using $dS_{\text{B}}$
generalizes Eq. (\ref{Single-Interacting diS}) and yields%
\begin{subequations}
\begin{equation}
d_{\text{i}}S(t)=%
{\textstyle\sum\nolimits_{l}}
[\beta(t)-\beta_{0l}]d_{\text{e}}Q_{l}(t)+A^{\prime}(t)d\xi(t),
\label{diSb-two media}%
\end{equation}
since this reduces to that result when we set $\beta_{01}=\beta_{02}=\beta
_{0}$. As above, $d_{\text{i}}Q(t)=d_{\text{i}}W(t)=A(t)d\xi(t)$; see Eq.
(\ref{diS-Isolated-Composite}), which gives rise to the last term above. Thus,
setting $d_{\text{e}}Q_{l}(t)=0,l=1,2$ to make $\Sigma$ isolated,\ we retrieve
$d_{\text{i}}S(t)$ in Eq. (\ref{diS-Isolated-Composite}) as expected. The
first sum above gives the external entropy exchanges with the two stochastic
media as above.

(ii) We now consider $\Sigma_{\text{C}}$, and allow $\widetilde{\Sigma
}_{\text{h1}}$ to directly interact with $\Sigma_{1}(t)$ at temperature
$T_{1}(t)$ and $\widetilde{\Sigma}_{\text{h2}}$ to directly interact with
$\Sigma_{2}(t)$ at temperature $T_{2}(t)$. Using $dS_{\text{C}}$ generalizes
Eq. (\ref{Single-Interacting diS}) and yields%
\begin{equation}
d_{\text{i}}S(t)=%
{\textstyle\sum\nolimits_{l}}
[\beta_{l}(t)-\beta_{0l}]d_{\text{e}}Q_{l}(t)+\left[  \beta_{1}(t)-\beta
_{2}(t)\right]  dE_{\text{in}}(t). \label{diSc-two media}%
\end{equation}
Again using Eq. (\ref{beta1-beta2}) to express $\beta_{l}$, we can rewrite the
above $d_{\text{i}}S(t)$ as the $d_{\text{i}}S(t)$ in Eq.
(\ref{diSb-two media}) for $\Sigma_{\text{B}}$, and also find that $d\xi$ is
given by Eq. (\ref{internal-dxi}).

It should be emphasized that the determination of $d_{\text{i}}S(t)$ in Eqs.
(\ref{diSb-two media}-\ref{diSc-two media}) is valid for all cases of $\Sigma$
interacting with $\widetilde{\Sigma}_{\text{h1}}$ and $\widetilde{\Sigma
}_{\text{h2}}$ as we have not imposed any conditions on $T_{1}(t)$ and
$T_{2}(t)$ with respect to $T_{01}$ and $T_{02}$, respectively.\ Thus it is
very general. The derivation also applies to the NEQ stationary state, which
happens when $T_{1}(t)\rightarrow T_{01}$ and $T_{2}(t)\rightarrow T_{02}$.
For the stationary case, using Eq. (\ref{diSc-two media}), we have%
\end{subequations}
\begin{equation}
d_{\text{i}}S^{\text{st}}=\left[  \beta_{01}-\beta_{02}\right]  dE_{\text{in}%
}, \label{diS-stationary}%
\end{equation}
where all quantities on the right have their steady values. Thus,
$d_{\text{i}}S^{\text{st}}$ is only determined by the stationary value of the
internal energy exchange $dE_{\text{in}}$. The reader can easily verify that
$d_{\text{i}}S(t)$ in Eqs. (\ref{diSb-two media}) also reduces to the above
result in the stationary limit.

From the above examples, we see that we can consider $\Sigma$\ in any of the
two realization $\Sigma_{\text{B}}$ and $\Sigma_{\text{C}}$ as we obtain the
same thermodynamics in that $d_{\text{i}}S(t)$ is identical. We emphasize this
important observation by summarizing it in the following conclusion.

\begin{conclusion}
\label{Conclusion-SigmaB-C} If we consider $\Sigma(t)$ as a single system
$\Sigma_{\text{B}}$ with an uniform temperature $T(t)$ and with an internal
variable $\xi(t)$, we do not need to consider the energy transfer
$dE_{\text{in}}(t)$ explicitly to obtain $d_{\text{i}}S(t)$. If we consider
$\Sigma(t)$ as a composite system $\Sigma_{\text{C}}$ formed of $\Sigma
_{1}(t)$ and $\Sigma_{2}(t)$ at their specific temperatures, then we
specifically need to consider the energy transfer $dE_{\text{in}}(t)$ to
obtain $d_{\text{i}}S(t)$ but no internal variable.
\end{conclusion}

This conclusion emphasizes the most important fact of the MNEQT that the
homogeneous thermodynamic temperature $T(t)$ of $\Sigma_{\text{B}}$\ can also
describe an inhomogeneous system $\Sigma_{\text{C}}$. This observation
justifies using the thermodynamic temperature $T(t)$ for treating $\Sigma(t)$
as a single system $\Sigma_{\text{B}}$, a black box, without any need to
consider the internal energy transfers.

The above discussion can be easily extended to also include inhomogeneities
such as two different work media $\widetilde{\Sigma}_{\text{w}}^{(1)}$ and
$\widetilde{\Sigma}_{\text{w}}^{(2)}$ corresponding to different pressures
$P_{01\text{ }}$and $P_{02}$. We will not do that here.

\subsection{$\Sigma$ Interacting with $\widetilde{\Sigma}_{\text{w}}$ and
$\widetilde{\Sigma}_{\text{h}}$}

In this case, $\Sigma$ is specified by two observables $E$ and $V$ so to
describe any inhomogeneity will require considering at least two subsystems
$\Sigma_{1}$ and $\Sigma_{2}$ specified by $E_{1},V_{1}$ and $E_{2,}V_{2}$,
respectively. From these four observables, we construct the following four
combinations%
\begin{align*}
E_{1}+E_{2}  &  =E,\xi_{\text{E}}=E_{1}-E_{2},\\
V_{1}+V_{2}  &  =V,\xi_{\text{V}}=V_{1}-V_{2},
\end{align*}
to express the entropy of the system%
\[
S(E,V,\xi_{\text{E}},\xi_{\text{V}})=S_{1}(E_{1},V_{1})+S_{2}(E_{2},V_{2})
\]
in terms of
\[
E_{1,2}=\frac{E\pm\xi_{\text{E}}}{2},V_{1,2}=\frac{V\pm\xi_{\text{V}}}{2}.
\]
Note that we have assumed that $\Sigma_{1}$ and $\Sigma_{2}$ are in EQ (no
internal variables for them). We now follow the procedure carried out in Sec.
\ref{Sec-Composite System} to identify thermodynamic temperature $T$, pressure
$P$, and affinities:%
\begin{equation}%
\begin{array}
[c]{c}%
\beta=\frac{(\beta_{1}+\beta_{2})}{2},\beta P=\frac{(\beta_{1}P_{1}+\beta
_{2}P_{2})}{2},\\
\beta A_{\text{E}}=\frac{(\beta_{1}-\beta_{2})}{2},\beta A_{\text{V}}%
=\frac{(\beta_{1}P_{1}-\beta_{2}P_{2})}{2}.
\end{array}
\label{System-Both Interactions}%
\end{equation}
All these quantities are SI-quantities and have the same values regardless of
whether $\Sigma$ is isolated or interacting. A more complicated
inhomogeneities will require more internal variables.

\begin{remark}
\label{Marker-Inhomogeneities}We now make an important remark about Eq.
(\ref{Irreversible EntropyGeneration-Complete}) that contains only a single
internal variable. From what is said above, it must include at least two
internal variables if $\Sigma$ contains inhomogeneity. in both $E$ and~$V$. If
it contains inhomogeneity. in only one variable, then and only then we will
have at least one internal variable. Thus, either we will $\xi_{\text{E}}$ or
$\xi_{\text{V}}$ as the case may be.
\end{remark}

\section{Tool-Narayanaswamy Equation \label{Sec-Tool-Narayan}}

We consider a simple NEQ\ laboratory problem to model the situation in a glass
\cite{Gujrati-I}. It is a system consisting of two "interpenetrating" parts at
different temperatures $T_{1}$ and $T_{2}>T_{1}$, but insulated from each
other so that they cannot come to equilibrium. The two parts are like slow and
fast motions in a glass, and the insulation allows us to treat them as
independent, having different temperatures. This is a very simple model for a
glass. A more detailed discussion is given elsewhere \cite{Gujrati-I}, where
each part was assumed to be in EQ macrostates. Here, we go beyond the earlier
discussion, and assume that the two parts are in some IEQ macrostates
$\mathcal{M}_{1}$ and $\mathcal{M}_{2}$ with temperatures $T_{1}$ and $T_{2}$,
respectively; we have suppressed "ieq" in the subscripts for simplicity. Thus,
there are irreversible processes going on within each part so that there are
nonzero irreversible macroheat $d_{\text{i}}Q_{1}$ and $d_{\text{i}}Q_{2}$
generated within each part. We wish to identify the temperature of the system
that we treat as a black box $\Sigma_{\text{B}}$. This will require
introducing its global temperature $T$. However, we also need to relate it to
$T_{1}$ and $T_{2}$ so that we need to treat $\Sigma$ as $\Sigma_{\text{C}}$.
We now imagine that each part is added a certain \emph{infinitesimal} amount
of exchange macroheat from outside, which we denote by $d_{\text{e}}Q_{1}$ and
$d_{\text{e}}Q_{2}$ so that $dQ_{1}=d_{\text{e}}Q_{1}+d_{\text{i}}Q_{1}$ and
$dQ_{2}=d_{\text{e}}Q_{2}+d_{\text{i}}Q_{2}$. This does not affect their
temperatures. We assume the entropy changes to be $dS_{1}$ and $dS_{2}$. Then,
we have for the net macroheat and entropy change
\[
dQ=dQ_{1}+dQ_{2},dS=dS_{1}+dS_{2}.
\]
We introduce the temperature $T$ by $dQ=TdS$. This makes it a thermodynamic
temperature of the black box. Using $dQ_{1}=T_{1}dS_{1},dQ_{2}=T_{2}dS_{2}$,
we immediately find

\qquad\qquad\qquad%
\[
dQ(1/T-1/T_{2})=dQ_{1}(1/T_{1}-1/T_{2}).
\]
By introducing $x=dQ_{1}/dQ$, which is determined by the setup, we find that
$T$ is given by%
\begin{equation}
\frac{1}{T}=\frac{x}{T_{1}}+\frac{1-x}{T_{2}}. \label{T_eff}%
\end{equation}
As $x$ is between $0$ and $1$, it is clear that $T$ lies between $T_{1}$ and
$T_{2}$ depending on the value of $x$. For $x=1/2$, this heuristic model
calculation reduces to that in Eq. (\ref{beta-A-Composite}) as expected. The
derivation also shows that the thermodynamic temperature $T$ is not affected
by having two nonoverlapping parts or overlapping parts. A similar relation
also exists for the pressure $P$ of a composite system; see Eq. (\ref{P_eff}).

If the insulation between the parts is not perfect, there is going to be some
energy transfer between the two parts, which would result in maximizing the
entropy of the system. As a consequence, their temperatures will eventually
become the same. During this period, $T$ will also change until all the three
temperatures become equal. This will require additional internal variable or
variables as in Sect. \ref{Sec-Composite System}.
\begin{figure}
[ptb]
\begin{center}
\includegraphics[
trim=0.199961in 0.000000in 0.401889in 0.000000in,
height=1.7158in,
width=3.0052in
]%
{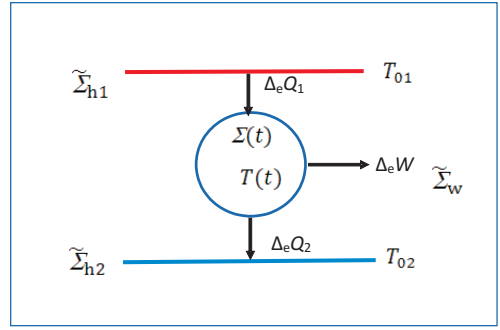}%
\caption{An irreversible Carnot cycle running between two heat reservoirs
$\widetilde{\Sigma}_{\text{h1}}$ and $\widetilde{\Sigma}_{\text{h2}}$.}%
\label{Fig-CarnotCycle}%
\end{center}
\end{figure}
\ 

\section{Irreversible Carnot Cycle\label{Sec-CarnotCycle}}

Let us consider an \textquotedblleft irreversible\textquotedblright\ Carnot
engine running between two heat sources $\widetilde{\Sigma}_{\text{h1}}$ and
$\widetilde{\Sigma}_{\text{h2}}$ as shown in Fig. \ref{Fig-CarnotCycle} that
are always maintained at fixed temperatures $T_{01}$ and $T_{02}$,
respectively, during each cyclic process $\mathcal{P}_{\text{cyc}}$. As
$\Sigma$ needs to perform work, we also need to consider it to be in constant
with a work source $\widetilde{\Sigma}_{\text{w}}$. We first observe the
following features of a reversible Carnot cycle. The system, which we take to
be formed by an ideal gas, starts in thermal contact with $\widetilde{\Sigma
}_{\text{h1}}$ in $\mathcal{A}_{\text{eq}}=\mathcal{M}_{\text{eq}}%
(T_{01},V_{1})$ as it expands to $V_{2}>V_{1}$, and ends in $\mathcal{M}%
_{\text{eq}}(T_{01},V_{2})$ through an isothermal process $\mathcal{P}%
_{1\text{eq}}$ resulting in $\Delta_{\text{e}}Q_{1\text{eq}}=\Delta_{\text{e}%
}W_{1\text{eq}}>0$. It is then detached from $\widetilde{\Sigma}_{\text{h1}}$
so no heat is exchanged ($\Delta_{\text{e}}Q_{2\text{eq}}=0$) but exchanges
work $\Delta_{\text{e}}W_{2\text{eq}}>0$\ during the process $\mathcal{P}%
_{2\text{eq}}$\ as it expands to $V_{3}$ and ends in $\mathcal{M}_{\text{eq}%
}(T_{02},V_{3})$ at temperature $T_{02}$. The system is brought in thermal
contact with $\widetilde{\Sigma}_{\text{h2}}$ now, and the volume is
compressed to $V_{4}$ isothermally during the process $\mathcal{P}%
_{3\text{eq}}$ and ends in $\mathcal{M}_{\text{eq}}(T_{02},V_{4})$. During
$\mathcal{P}_{3\text{eq}}$, $\Delta_{\text{e}}Q_{3\text{eq}}=\Delta_{\text{e}%
}W_{3\text{eq}}<0$. The choice of $V_{4}$ is chosen so that $\Sigma$ comes
back to $\mathcal{A}_{\text{eq}}=\mathcal{M}_{\text{eq}}(T_{01},V_{1})$ along
a process $\mathcal{P}_{4\text{eq}}$ after detaching it from $\widetilde
{\Sigma}_{\text{h2}}$ during which $\Delta_{\text{e}}Q_{4\text{eq}}=0$, but
$\Delta_{\text{e}}W_{4\text{eq}}<0$. The four segments bring back $\Sigma$ to
its starting state $\mathcal{A}_{\text{eq}}$, and form a cycle $\mathcal{P}%
_{\text{eq,cyc}}$. It is well known that the EQ efficiency $\epsilon
_{\text{eq}}$ of the Carnot cycle is%
\begin{subequations}
\begin{equation}
\epsilon_{\text{eq}}=1-T_{02}/T_{01}, \label{CarnotEfficiency-EQ}%
\end{equation}
so that
\begin{equation}
\Delta_{\text{e}}W_{\text{eq}}=\epsilon_{\text{eq}}\Delta_{\text{e}}Q_{1},
\label{Carnot-EQMacrowork}%
\end{equation}
the equilibrium macrowork obtained from the cycle for a given $\Delta
_{\text{e}}Q_{1}$.

We now consider an irreversible cyclic process $\mathcal{P}$, which consists
of the same four segments except that some or all may be irreversible. We have
discussed such a process in Sect. \ref{Sec-CyclicProcess}. However, to have a
cyclic process, the system must start and end in $\mathcal{M}_{\text{eq}%
}(T_{01},V_{1})$, which does not require any internal variable. Being an
irreversible process, there is no guarantee that $\Sigma$ would be in EQ
macrostates at the end of $\mathcal{P}_{1},\mathcal{P}_{2}$, and
$\mathcal{P}_{3}$; $\mathcal{P}_{4}$ must bring $\Sigma$ to the EQ initial
macrostate $\mathcal{M}_{\text{eq}}(T_{01},V_{1})$. However, we will simplify
the calculation here by assuming that the end states in $\mathcal{P}%
_{1},\mathcal{P}_{2}$, and $\mathcal{P}_{3}$ are $\mathcal{M}_{\text{eq}%
}(T_{01},V_{2}),\mathcal{M}_{\text{eq}}(T_{02},V_{3})$, and $\mathcal{M}%
_{\text{eq}}(T_{02},V_{4})$, respectively. However, relaxing this condition
does not change the results below.

Being a cyclic process, we have%
\end{subequations}
\begin{subequations}
\begin{equation}
\Delta_{\text{c}}E=\Delta_{\text{c}}S=0 \label{Carnot-E-S}%
\end{equation}
over $\mathcal{P}$. Thus, over $\mathcal{P}$,
\begin{equation}
\Delta_{\text{e}}Q=\Delta_{\text{e}}W. \label{Carnot-Qe-We}%
\end{equation}
In the MNEQT, we also have over $\mathcal{P}$,
\begin{equation}
\Delta Q=\Delta W. \label{Carnot-Q-W}%
\end{equation}

Let $\Delta_{\text{e}}Q_{1}=T_{01}\Delta_{\text{e}}S_{1}$ and $\Delta
_{\text{e}}Q_{3}=T_{02}\Delta_{\text{e}}S_{3}$ be the macroheat exchanges
during $\mathcal{P}_{1}$ and $\mathcal{P}_{3}$, and $\Delta_{\text{e}}%
Q=\Delta_{\text{e}}Q_{1}+\Delta_{\text{e}}Q_{3}$. \ Similarly, let
$\Delta_{\text{e}}W_{l}$ be the macrowork exchanges during $\mathcal{P}%
_{l},l=1,2,3,4$, and $\Delta_{\text{e}}W=%
{\textstyle\sum\nolimits_{l}}
\Delta_{\text{e}}W_{l}$, the net exchange macrowork. From $\Delta_{\text{c}%
}S=0$ follows
\end{subequations}
\[
\Delta_{\text{e}}Q_{1}/T_{01}+\Delta_{\text{e}}Q_{3}/T_{02}+\Delta_{\text{i}%
}S=0,
\]
which can be rewritten by simple manipulation as%
\begin{subequations}
\begin{equation}
\Delta_{\text{e}}W=\Delta_{\text{e}}Q_{1}(1-T_{02}/T_{01})-T_{02}%
\Delta_{\text{i}}S, \label{Carnot-Exchange Work}%
\end{equation}
where we have used the identity in Eq. (\ref{Carnot-Qe-We}). We can also
express it as%
\begin{equation}
\Delta_{\text{e}}W=\epsilon_{\text{eq}}\Delta_{\text{e}}Q_{1}-T_{02}%
\Delta_{\text{i}}S\leq\Delta_{\text{e}}W_{\text{eq}}.
\label{Carnot-Exchange Work-InEquality}%
\end{equation}
The efficiency of the irreversible Carnot cycle is given by%
\end{subequations}
\begin{equation}
\epsilon_{\text{irr}}=\epsilon_{\text{eq}}-T_{02}\Delta_{\text{i}}%
S/\Delta_{\text{e}}Q_{1}\leq\epsilon_{\text{eq}}. \label{CarnotEfficiency-NEQ}%
\end{equation}
We remark that $\epsilon_{\text{irr}}$ above is similar to the result obtained
by Eu \cite[see Eq. (7.139)]{Eu} in which the numerator of the last term is
identified as "the total dissipation" by Eu; however, the analysis is tedious
compared to the one given here.

The determination of $\Delta_{\text{i}}S$ requires the extended state space
$\mathfrak{S}_{\mathbf{Z}}$ needed for the four process segments. We will
focus on Eq. (\ref{Irreversible EntropyGeneration-Complete}). Let $T_{l}(t)$
and $P_{l}(t)$ be the temperature and pressure of $\Sigma$ along
$\mathcal{P}_{l}$, respectively, with $l=1,\cdots,4$ for $\mathcal{P}_{l}$. As
seen from Eq. (\ref{System-Both Interactions}), we need at least two internal
variables $\xi_{\text{E}l}$ and $\xi_{\text{V}l}$ along $\mathcal{P}_{l}$ that
are usually different along the four segments. The corresponding affinities
must vanish at the end points of each segment because they are EQ macrostates.
We will assume this to be the case, and introduce%
\[
\mathbf{A}_{l}(t)\cdot d\boldsymbol{\xi}_{l}(t)=A_{\text{E}l}(t)d\xi
_{\text{E}l}(t)+A_{\text{V}l}(t)d\xi_{\text{V}l}(t)
\]
for each segment. Then,%
\begin{align}
\Delta_{\text{i}}S_{l}  &  =\int_{\mathcal{P}_{l}}[(\beta_{l}(t)-\beta
_{0l})d_{\text{e}}Q_{l}(t)+(P_{l}(t)-P_{0l})dV(t)\label{DeltaS_i-Carnot}\\
&  +\beta_{l}(t)\mathbf{A}_{l}(t)\cdot d\boldsymbol{\xi}_{l}(t)],\nonumber
\end{align}
where $d_{\text{e}}Q_{l}(t)=0$ for $l=2,4$, and where $P_{0l}$ is the external
pressures of $\widetilde{\Sigma}_{\text{w}}$ along $\mathcal{P}_{l}$ and must
be the same as for the reversible Carnot engine. This then determines
$\Delta_{\text{i}}S$ so that $\epsilon_{\text{irr}}$ is determined.

Using $\Delta_{\text{i}}Q=\Delta Q-\Delta_{\text{e}}Q$, we also have%
\begin{equation}
\Delta_{\text{i}}W=%
{\textstyle\sum\limits_{l=1}^{4}}
{\textstyle\int\limits_{\mathcal{P}_{l}}}
\left[  T_{l}(t)d_{\text{i}}S_{l}(t)+(T_{l}(t)-T_{0l})d_{\text{e}}%
S_{l}(t)\right]  \label{Irre-Work-Carnot1}%
\end{equation}
Recognizing that $d_{\text{e}}S_{l}(t)$ is nonzero only for $l=1,3$, we can
also rewritten $\Delta_{\text{i}}W$ as%
\begin{equation}
\Delta_{\text{i}}W=%
{\textstyle\sum\limits_{l=1,3}}
{\textstyle\int\limits_{\mathcal{P}_{l}}}
(\beta_{0l}/\beta_{l}(t)-1)d_{\text{e}}Q_{l}(t)+%
{\textstyle\int\limits_{\mathcal{P}}}
T(t)d_{\text{i}}S(t). \label{Irre-Work-Carnot2}%
\end{equation}

It should be noted that nowhere did we use the vanishing of the affinities in
the EQ states at the end of $\mathcal{P}_{1},\mathcal{P}_{2}$, and
$\mathcal{P}_{3}$ so the calculation above is not limited by this requirement.
Thus, the results in this section are general.

\section{Origin of Friction and Brownian Motion\label{Sec-Friction}}

It is well known, see for example Kestin \cite[Secs. 4.7 and 5.12]{Kestin},
that there are several ways to incorporate friction in a system in
thermodynamics. This has to do with the difficulties in making an unambiguous
distinction between various possibilities of exchange macroheat in a process
$\mathcal{P}$. We overcome this problem by using the MNEQT in which this is
not an issue as both $dQ$ and $d_{\text{e}}Q$ are uniquely defined. We
identify the origin of friction in our approach
\cite{Gujrati-Entropy1,Gujrati-Entropy2,Gujrati-II,Gujrati-Heat-Work0,Gujrati-Heat-Work}
by considering relative motion between parts of a system or between the system
and the medium; see also \cite{Gujrati-LangevinEq,Landau}. Such a situation
arises during sudden mixing of fluids or in a Couette flow or when friction is
involved between two bodies. The origin of friction is also applicable to NEQ
terminal states of the process in the MNEQT.

\subsection{Piston-Gas System}

We consider the piston gas system in Fig. \ref{Fig_Piston-Spring}(a). As
discussed in Sec. \ref{Sec-Piston-Spring}, the system entropy in
$\mathcal{M}_{\text{ieq}}$ is $S(E,V,\mathbf{P}_{\text{gc}},\mathbf{P}%
_{\text{p}})$. Hence, the corresponding Gibbs fundamental relation becomes
\[
dS=\beta\lbrack dE+PdV-\mathbf{V}_{\text{gc}}\mathbf{\cdot}d\mathbf{P}%
_{\text{gc}}-\mathbf{V}_{\text{p}}\mathbf{\cdot}d\mathbf{P}_{\text{p}}],
\]
where we have used the conventional conjugate fields
\begin{equation}%
\begin{tabular}
[c]{c}%
$\beta\doteq\partial S/\partial E,\beta P\doteq\partial S/\partial V,,$\\
$\beta\mathbf{V}_{\text{gc}}\doteq-\partial S/\partial\mathbf{P}_{\text{gc}%
}\mathbf{,}\beta\mathbf{V}_{\text{p}}\doteq-\partial S/\partial\mathbf{P}%
_{\text{p}}$%
\end{tabular}
\ \ \ \ \ \ \ \ \ \ \ \ \ \ \label{Conjugate-Fields_Gas-Piston}%
\end{equation}
as shown by Landau and Lifshitz \cite{Landau-Fluid} and by us elsewhere
\cite[and references theirin]{Gujrati-II}. Using Eq.
(\ref{Stationary_Momentum_Condition}), we can rewrite this equation as%
\begin{equation}
dS=\beta\lbrack dE+PdV-\mathbf{V\cdot}d\mathbf{P}_{\text{p}}]
\label{Gibbs-Fundamental_Gas-Piston}%
\end{equation}
in terms of the $\emph{relative}$ \emph{velocity}, also known as the
\emph{drift velocity }$\mathbf{V\doteq V}_{\text{p}}-\mathbf{V}_{\text{gc}}$
of the piston with respect to $\Sigma_{\text{gc}}$. We can cast the drift
velocity term as $\mathbf{V\cdot}d\mathbf{P}_{\text{p}}\equiv\mathbf{F}%
_{\text{p}}\mathbf{\cdot}d\mathbf{R}$, where $\mathbf{F}_{\text{p}}\doteq
d\mathbf{P}_{\text{p}}\mathbf{/}dt$ is the \emph{force}\ (a macroforce in
nature) and $d\mathbf{R=V}dt$ is the \emph{relative displacement} of the piston.

The internal motions of $\Sigma_{\text{gc}}$\ and $\Sigma_{\text{p}}$\ is not
controlled by any external agent so the relative motion described by the
relative displacement $\mathbf{R}$ represents an \emph{internal variable
}\cite{Kestin} so that the corresponding affinity $\mathbf{F}_{\text{p}0}=0$
for $\widetilde{\Sigma}$. Because of this, the first law $dE=T_{0}d_{\text{e}%
}S-P_{0}dV$ as given in Eq. (\ref{FirstLaw-MI}) does not involve the relative
displacement $\mathbf{R}$. We now support this claim using our approach in the
following. This also shows how $\mathcal{H}(\left.  \mathbf{x}\right\vert
V,\mathbf{P}_{\text{gc}},\mathbf{P}_{\text{p}})$ develops a dependence on the
internal variable $\mathbf{R}$. We manipulate $dS$ in Eq.
(\ref{Gibbs-Fundamental_Gas-Piston}) by using the above first law for $dE$ so
that
\[
TdS=T_{0}d_{\text{e}}S+(P-P_{0})dV-\mathbf{F}_{\text{p}}\mathbf{\cdot
}d\mathbf{R,}%
\]
which reduces to%
\[
T_{0}d_{\text{i}}S=(T_{0}-T)dS+(P-P_{0})dV-\mathbf{F}_{\text{p}}\mathbf{\cdot
}d\mathbf{R.}%
\]
This equation expresses the irreversible entropy generation as sum of three
distinct and independent irreversible entropy generations. To comply with the
second law, we conclude that for $T_{0}>0$,%
\begin{equation}
(T_{0}-T)dS\geq0,(P-P_{0})dV\geq0,\mathbf{F}_{\text{p}}\mathbf{\cdot
}d\mathbf{R\leq}0, \label{SecondLaw-Consequences}%
\end{equation}
which shows that each of the components of $d_{\text{i}}S$\ is nonnegative. In
equilibrium, each irreversible component vanishes, which happens when%
\begin{equation}
T\rightarrow T_{0},P\rightarrow P_{0}\text{, and }\mathbf{V}\rightarrow0\text{
or }\mathbf{F}_{\text{p}}\rightarrow0. \label{Equilibrium-Piston}%
\end{equation}
The inequality $\mathbf{F}_{\text{p}}\mathbf{\cdot}d\mathbf{R\leq}0$ shows
that $\mathbf{F}_{\text{p}}$ and $d\mathbf{R}$ are antiparallel, which is what
is expected of a \emph{frictional} force $\mathbf{F}_{\text{fr}}$. Thus, we
can identify $\mathbf{F}_{\text{p}}$ with $\mathbf{F}_{\text{fr}}$, and the
irreversible frictional macrowork $d_{\text{i}}W_{\text{fr}}$ done during this
motion by%
\begin{equation}
d_{\text{i}}W_{\text{fr}}\doteq-\mathbf{V}(t)\cdot d\mathbf{P}(t)=-d\mathbf{R}%
(t)\cdot\mathbf{F}_{\text{fr}}(t)>0. \label{FrictionalWork}%
\end{equation}
This macrowork, which appears as part of $d_{\text{i}}W$ in our approach,
vanishes as the motion ceases so that the equilibrium value $\mathbf{V}_{0}$
of $\mathbf{V}(t)$ is $\mathbf{V}_{0}=0$ just as the equilibrium affinity
$\mathbf{A}_{0}=0$ for $\mathbf{\xi}$. This causes the piston to finally come
to rest. As $\mathbf{F}_{\text{fr}}$ and $\mathbf{V}$ vanish together, we can
express this force as
\begin{equation}
\mathbf{F}_{\text{fr}}=-\mu\mathbf{V}f(\mathbf{V}^{2}),
\label{Friction-GeneralForm}%
\end{equation}
where $\mu>0$ and $f$ is an \emph{even} function of $\mathbf{V}$. The medium
$\widetilde{\Sigma}$ is specified by $T=T_{0},P=P_{0}$ and $\mathbf{V}_{0}=0$
or $\mathbf{F}_{\text{p}}=0$. We will take $\mathbf{F}_{\text{fr}}$ and
$d\mathbf{R}$ to be colinear and replace $\mathbf{F}_{\text{fr}}\mathbf{\cdot
}d\mathbf{R}$ by $-F_{\text{fr}}dx$ ($F_{\text{fr}}dx\geq0$), where $dx$ is
the magnitude of the relative displacement $d\mathbf{R}$. The sign convention
is that $F_{\text{fr}}$ and increasing $x$ point in the same direction. We can
invert Eq. (\ref{Gibbs-Fundamental_Gas-Piston}) to obtain%
\begin{equation}
dE=TdS-PdV-F_{\text{fr}}dx \label{Gibbs-Fundamental-Energy_Gas-Piston}%
\end{equation}
in which $dQ=TdS$ from our general result in Eq. (\ref{dQ-dS}). Comparing the
above equation with the first law in Eq. (\ref{FirstLaw-SI}), we conclude
that
\begin{equation}
dW=PdV+F_{\text{fr}}dx. \label{Work-Friction}%
\end{equation}
The important point to note is that the friction term $F_{\text{fr}}dx$
properly belongs to $dW$.\ As $d_{\text{e}}W=P_{0}dV$, we have%
\begin{equation}
d_{\text{i}}W=(P-P_{0})dV+F_{\text{fr}}dx. \label{Irreversible_Work-Piston}%
\end{equation}
Both contributions in $d_{\text{i}}W$ are separately nonnegative.

We can determine the exchange macroheat $d_{\text{e}}Q=dQ-d_{\text{i}}W$%
\begin{equation}
d_{\text{e}}Q=TdS-(P-P_{0})dV-F_{\text{fr}}dx \label{ExchangeHeat-Friction}%
\end{equation}
\ 

It should be emphasized that in the above discussion, we have not considered
any other internal motion such as between different parts of the gas besides
the relative motion between $\Sigma_{\text{gc}}$ and $\Sigma_{\text{p}}$.
These internal motions within $\Sigma_{\text{g}}$ can be considered by
following the approach outlined elsewhere \cite{Gujrati-II}. We will not
consider such a complication here.

\subsection{Particle-Spring-Fluid System}

It should be evident that by treating the piston as a mesoscopic particle such
as a pollen or a colloid, we can treat its thermodynamics using the above
procedure. This allows us to finally make a connection with the system
depicted in Fig. \ref{Fig_Piston-Spring}(b) in which the particle (a pollen or
a colloid) is manipulated by an external force $F_{0}$. We need to also
consider two additional forces $F_{\text{s}}$ and $F_{\text{fr}}$, both
pointing in the same direction as increasing $x$; the latter is the frictional
force induced by the presence of the fluid in which the particle is moving
around. The analog of Eq. (\ref{Irreversible_Work-Piston}) for this case
becomes%
\begin{equation}
d_{\text{i}}W=(F_{\text{s}}+F_{0})dx+F_{\text{fr}}dx=F_{\text{t}}dx,
\label{diW-Particle-Spring}%
\end{equation}
where $F_{\text{t}}=F_{\text{s}}+F_{0}+F_{\text{fr}}$ is the net force. The
other two macroworks are $dW=(F_{\text{s}}+F_{\text{fr}})dx$ and
$d\widetilde{W}=F_{0}dx=-d_{\text{e}}W$. In EQ, $F_{\text{fr}}=0$ and
$F_{\text{s}}+F_{0}=0$ ($F_{0}\neq0$) to ensure $d_{\text{i}}W=0$.

\subsection{Particle-Fluid System}

In the absence of a spring in the previous subsection, we must set
$F_{\text{s}}=0$ so
\begin{equation}
dW=F_{\text{fr}}dx,d\widetilde{W}=F_{0}dx=-d_{\text{e}}W,d_{\text{i}}%
W=(F_{0}+F_{\text{fr}})dx. \label{diW-Particle-Fluid}%
\end{equation}
In EQ, $F_{0}+F_{\text{fr}}=0$ so that $F_{\text{fr}}=-F_{0}$. This means that
in EQ, the particle's nonzero terminal velocity is determined by $F_{0}$ as
expected.
\begin{figure}
[ptb]
\begin{center}
\includegraphics[
height=1.9147in,
width=2.5918in
]%
{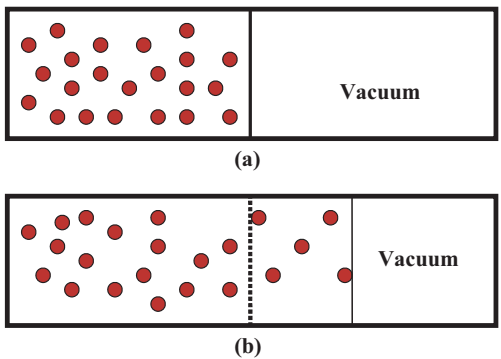}%
\caption{Free expansion of a gas. The gas is confined to the left chamber,
which is separated by a hard partition (shown by a solid black vertical line)
from the vacuum in the right chamber as shown in (a). At time $t=0$, the
partition is removed abruptly as shown by the broken line in its original
place in (b). The gas expands in the empty space, devoid of matter and
radiation, on the right but the expansion is gradual as shown by the solid
front, which separates it from the vacuum on its right. We can also think of
the hard partition in (a) as a piston, which maintains the volume of the gas
on its left. The piston can be moved slowly or rapidly to the right within the
right chamber with a pressure $P_{0}<P$ to change this volume. The free
expansion occurs when the piston moves extremely (infinitely) fast by letting
$P_{0}\rightarrow0$.}%
\label{Fig_Expansion}%
\end{center}
\end{figure}

\section{Free Expansion\label{Sec-Free Expansion}}

What makes NEQ thermodynamics complicated than EQ thermodynamics is the
evaluation of nonzero irreversible entropy generation $d_{\text{i}}S(t)\geq0$.
As $d_{\text{i}}S^{\text{Q}}(t)$ and $d_{\text{i}}S^{\text{W}}(t)$ are
independent contributions, it is simpler to consider an isolated system for
which $d_{\text{i}}S^{\text{Q}}(t)\equiv0$ so that we only deal with
$d_{\text{i}}S^{\text{W}}(t)$ in Eq. (\ref{IrreversibleEntropy}). Then the use
of Eq. (\ref{EntropyDiff-Isolated}) allows us to determine the temperature of
the system in any arbitrary macrostate. In free expansion, there is no
exchange of any kind so $d=d_{i}$. This simplifies our notation as we do not
need to use $d_{\text{i}}$ when referring to $\Sigma$, which we will do in
this section.

\subsection{Classical Free Expansion in $\mathfrak{S}_{\mathbf{Z}}%
$\label{Sec-FreeExpansion-MNEQT}}

The gas, which forms $\Sigma$, expands freely in a vacuum ($\widetilde{\Sigma
}$) from $V_{\text{in}}$, the volume of the left chamber, to $V_{\text{fin}%
}=2V_{\text{in}}$, the volume of $\Sigma_{0}$; the volume of the right chamber
is $V_{\text{fin}}-V_{\text{in}}=V_{\text{in}}$. The initial and final
macrostate are denoted by $\mathcal{A}_{\text{eq}}$\ and $\mathcal{B}%
_{\text{eq}}$. The vacuum exerts no pressure ($\widetilde{P}=P_{\text{vacuum}%
}=0$). The left (L) and right (R) chambers are initially separated by an
impenetrable partition, shown by the solid partition in Fig.
\ref{Fig_Expansion}(a), to ensure that they are thermodynamically independent
regions, with all the $N$ particles of $\Sigma$ in the left chamber, which are
initially in an EQ macrostate $\mathcal{A}_{\text{eq}}$ with entropy
$S_{\text{in}}$. For ideal gas, we have \cite{Landau}%
\[
S(E,V)=N\ln(eV/N)+f(E),
\]
where $N$ is kept as a suffix for a reason that will become evident below. The
initial pressure and temperature of the gas prior to expansion in
$\mathcal{A}_{\text{eq}}$ at time $t=0$ are $P_{\text{in}}$ and $T_{\text{in}%
}=T_{0}$, respectively, that are related to $E_{0}=E_{\text{in}}$ and
$V_{\text{in}}$ by its EQ equation of state. A similar set of quantities also
pertain to $\mathcal{B}_{\text{eq}}$. As $\Sigma_{0}$ is isolated, the
expansion occurs at \emph{constant} energy $E_{0}$, which is also the energy
of $\Sigma$.

It should be stated, which is also evident from Fig. \ref{Fig_Expansion}(b),
that while the removal of the partition is instantaneous, the actual process
of gas expanding in the right chamber is continuous and gradually fills it.
This is obviously a very complex internal process in a highly inhomogeneous
macrostate. As thus, it will require many internal variables to describe
different number of particles, different energies, different pressures,
different flow pattern which may be even chaotic, etc. in each of the
chambers. For example, we can divide the volume $V_{\text{fin}}$ into many
layers of volume parallel to the partition, each layer in equilibrium with
itself but need not be with others; see the example in Sec.
\ref{Sec-InternalVariables}. As our aim is to show the feasibility of the
MNEQT in this investigation, we will simplify the situation by limiting to two
internal variables. The first internal variable $\xi_{\text{N}}$
\ \ \ \ \ \ \ \ \ \ \ \ \ \ \ \ \ \ \ \ \ \ \ \ \ \ \ \ \ \ \ \ \ \ \ \ \ \ \ \ \ \ \ \ \ \ \ \ \ \ \ \ \ \ \ \ \ \ \ \ \ \ \ \ \ \ \ \ \ \ \ \ \ \ \ \ \ \ \ \ \ \ \ \ \ \ \ \ \ \ \ \ \ \ \ \ \ \ \ \ \ \ \ \ \ \ \ \ \ \ \ \ \ \ \ \ \ \ \ \ \ \ \ \ \ \ \ \ \ \ \ \ \ \ \ \ \ \ \ \ \ \ \ \ \ \ \ \ \ \ \ \ \ \ \ \ \ \ \ \ \ \ \ \ \ \ \ \ \ \ \ \ \ \ \ \ \ \ \ \ \ \ \ \ \ \ \ \ \ \ \ \ \ \ \ \ \ \ \ \ \ \ \ \ \ \ \ \ \ \ \ \ \ \ \ \ \ \ \ \ \ \ \ \ \ \ \ \ \ \ \ \ \ \ \ \ \ \ \ \ \ \ \ \ \ \ \ \ \ \ \ \ \ \ \ \ \ \ \ \ \ \ \ \ \ \ \ \ \ \ \ \ \ \ \ \ \ \ \ \ \ \ \ \ \ \ \ \ \ \ \ \ \ \ \ \ \ \ \ \ \ \ \ \ \ \ \ \ \ \ \ \ \ \ \ \ \ \ \ \ \ \ \ \ \ \ \ \ \ \ \ \ \ \ \ \ \ \ \ \ \ \ \ \ \ \ \ \ \ \ \ \ \ \ \ \ \ \ \ \ \ \ \ \ \ \ \ \ \ \ \ \ \ \ \ \ \ \ \ \ \ \ \ \ \ \ \ \ \ \ \ \ \ \ \ \ \ \ \ \ \ \ \ \ \ \ \ \ \ \ \ \ \ \ \ \ \ \ \ \ \ \ \ \ \ \ \ \ \ \ \ \ \ \ \ \ \ \ \ \ \ \ \ \ \ \ \ \ \ \ \ \ \ \ \ \ \ \ \ \ \ \ \ \ \ \ \ \ \ \ \ \ \ \ \ \ \ \ \ \ \ \ \ \ \ \ \ \ \ \ \ \ \ \ \ \ \ \ \ \ \ \ \ \ \ \ \ \ \ \ \ \ \ \ \ \ \ \ \ \ \ \ \ \ \ \ \ \ \ \ \ \ \ \ \ \ \ \ \ \ \ \ \ \ \ \ \ \ \ \ \ \ \ \ \ \ \ \ \ \ \ \ \ \ \ \ \ \ \ \ \ \ \ \ \ \ \ \ \ \ \ \ \ \ \ \ \ \ \ \ \ \ \ \ \ \ \ \ \ \ \ \ \ \ \ \ \ \ \ \ \ \ \ \ \ \ \ \ \ \ \ \ \ \ \ \ \ \ \ \ \ \ \ \ \ \ \ \ \ \ \ \ \ \ \ \ \ \ \ \ \ \ \ \ \ \ \ \ \ \ \ \ \ \ \ \ \ \ \ \ \ \ \ \ \ \ \ \ \ \ \ \ \ \ \ \ \ \ \ \ \ \ \ \ \ \ \ \ \ \ \ \ \ \ \ \ \ \ \ \ \ \ \ \ \ \ \ \ \ \ \ \
\begin{equation}
\xi_{\text{N}}\doteq N_{\text{R}}/N \label{xi-free-expansion}%
\end{equation}
is obtained by considering only two layers to describe different numbers
$N_{\text{L}}=(1-\xi_{\text{N}})N$ particles to left and $N_{\text{R}}%
=N\xi_{\text{N}}$ particles to the right of the chamber partition in Fig.
\ref{Fig_Expansion}(a) as a function of time. Initially, $\xi_{\text{N}}=0$
and finally at EQ, $\xi_{\text{N}}=1/2$. At each instant, we imagine a front
of the expanding gas shown by the solid vertical line in Fig.
\ref{Fig_Expansion}(b) containing all the particles to its left. We denote
this volume by a time-dependent $V=V(t)$ to the right of which exists a
vacuum. This means that at each instant when there is a vacuum to the right of
this front, the gas is expanding against zero pressure so that $d_{\text{e}%
}W=0$. Since we have a NEQ expansion, $dW>0$. As $V(t)$ cannot be controlled
externally, it can be used to determine another internal variable by using
$V^{\prime}=V-V_{\text{in}}$:
\[
\xi_{\text{V}}\doteq V^{\prime}/V=1-V_{\text{in}}/V,
\]
so that $V^{\prime}=\xi_{\text{V}}V$ and $V_{\text{in}}=(1-\xi_{\text{V}})V$.
Initially, $\xi_{\text{V}}=0$ and finally at EQ, $\xi_{\text{V}}=1/2$. The
choice of the two internal variables $\xi_{\text{N}}(t)$\ and $\xi_{\text{V}%
}(t)$\ follows the procedure in Sect. \ref{Sec-Tool-Narayan} for two
subsystems of different sizes, and allow us to distinguish between
$\mathcal{P}^{\prime}$ and $\mathcal{P}^{\prime\prime}$ as we will see below.
We assume that the expansion is isothermal (which it need not be) so there is
no additional internal variable associated with temperature variation. As
$dQ=dW>0$, the expansion is irreversible so the entropy continues to increase.

At $t=0$, the partition is suddenly removed, shown by the broken partition in
Fig. \ref{Fig_Expansion}(b) and the gas expands freely to the final volume
$V(t^{\prime})=V_{\text{fin}}$ at time $t^{\prime}<\tau_{\text{eq}}$ during
$\mathcal{P}^{\prime}$. At $t^{\prime}$, the free expansion stops but there is
no reason a priori for $\xi_{\text{N}}=0$ so the gas is still inhomogeneous
($\xi_{\text{N}}\neq0$). This is in a NEQ macrostate until $\xi_{\text{N}}$
achieves its EQ value $\xi_{\text{N}}=0$ during $\mathcal{P}^{\prime\prime}$,
at the end of which at $t=\tau_{\text{eq}}$ the gas eventually comes into
$\mathcal{B}_{\text{eq}}$ isoenergetically. The complete process is
$\overline{\mathcal{P}}=\mathcal{P}^{\prime}\cup\mathcal{P}^{\prime\prime}%
$\ between $\mathcal{A}_{\text{eq}}$ and $\mathcal{B}_{\text{eq}}$.\ We
briefly review this expansion in the MNEQT \cite{Gujrati-Heat-Work}.

We work in the extended state space with the two internal variables, which we
denote simply by $\mathfrak{S}$ here. Using Eq. (\ref{FirstLaw-SI}), we have%
\begin{subequations}
\begin{equation}
dS(t)=dW(t)/T(t). \label{dS-dW-Free Expansion}%
\end{equation}
Setting $P_{0}=0$ in Eq. (\ref{MI-Work}), we have
\begin{equation}
dW(t)=\left\{
\begin{array}
[c]{c}%
P(t)dV(t)+\mathbf{A}(t)\cdot d\boldsymbol{\xi}(t)\text{ for }t<t^{\prime}%
<\tau_{\text{eq}},\\
A(t)d\xi(t)\text{ for }t^{\prime}<t\leq\tau_{\text{eq}};
\end{array}
\right.  \label{dW-Free Expansion}%
\end{equation}
here, we have used the fact that $V(t)$ does not change for $\tau^{\prime
}<t\leq\tau_{\text{eq}}$. Thus,
\end{subequations}
\begin{align*}
\Delta S  &  =\int_{\overline{\mathcal{P}}}\frac{dW(t)+\mathbf{A}(t)\cdot
d\boldsymbol{\xi}(t)}{T(t)}>0,\\
\Delta Q  &  =\int_{\overline{\mathcal{P}}}dW(t)=\Delta W>0;
\end{align*}
the last equation is the fundamental identity in Eq. (\ref{diQ-diW-EQ}). The
irreversible entropy change $\Delta S$\ from EQ macrostate from $\mathcal{A}%
_{\text{eq}}$ to $\mathcal{B}_{\text{eq}}$ is the EQ entropy change
$\Delta_{\text{i}}S$ is%
\begin{equation}
\Delta S\equiv S_{\text{fin}}-S_{\text{in}}, \label{Entropy-Change-Eq0}%
\end{equation}
and can be directly obtained since the EQ entropy $S(E,V)$ is known. The above
analysis is also valid for any arbitrary free expansion process $\mathcal{P}$
and not just $\overline{\mathcal{P}}$\ as we have not used any information yet
about $\mathcal{A}_{\text{eq}}$ to $\mathcal{B}_{\text{eq}}$.

For $V_{\text{fin}}=2V_{\text{in}}$, $\Delta_{\text{i}}S=N\ln2$, a well-known
result \cite{Prigogine}. Here, we provide a more general result for the
entropy for $t\leq t^{\prime}$, which can be trivially determined:%
\[
S(\xi_{\text{V}},\xi_{\text{N}})=N_{\text{L}}\ln(eV_{\text{in}}/N_{\text{L}%
})+N_{\text{R}}\ln(eV^{\prime}/N_{\text{R}})+f(E_{0}).
\]
Thus, for arbitrary $\xi_{\text{V}}$ and $\xi_{\text{N}}$, we have
$\Delta_{\text{i}}S(\xi_{\text{V}},\xi_{\text{N}})=S(\xi_{\text{V}}%
,\xi_{\text{N}})-S_{\text{in}}$. We can determine the two affinities. A simple
calculation gives%
\begin{align*}
A_{\text{V}}/T  &  =\partial S/\partial\xi_{\text{V}}=1-\xi_{\text{N}}%
-\xi_{\text{V}},\\
A_{\text{N}}/T  &  =\partial S/\partial\xi_{\text{N}}=\ln\frac{(1-\xi
_{\text{N}})\xi_{\text{V}}}{(1-\xi_{\text{V}})\xi_{\text{N}}}.
\end{align*}
We see that $A_{\text{V}}$ does not vanish when $V^{\prime}=V_{\text{fin}}$ as
discussed above. It is easy to verify that $\mathbf{A}$ vanish in
$\mathcal{B}_{\text{eq}}$. The pressure of the expanding gas is obtained by
using the derivative $\partial S/\partial V$ as usual. A simple calculation
yields%
\begin{align}
\beta P  &  =(1-\xi_{\text{V}})\frac{N_{\text{L}}}{V_{\text{in}}}%
+\xi_{\text{V}}\frac{N_{\text{R}}}{V^{\prime}}\label{P_eff}\\
&  =(1-\xi_{\text{V}})\beta P_{\text{L}}+\xi_{\text{V}}\beta P_{\text{R}%
}.\nonumber
\end{align}
The last expression for pressure has a close similarity with the
Toll-Narayanaswamy equation (\ref{T_eff}), which should not be surprising.
Before expansion, we have $\beta P_{\text{in}}=N/V_{\text{in}}$ in
$\mathcal{A}_{\text{eq}}$ and $\beta P_{\text{fin}}=N/V_{\text{fin}}$ in
$\mathcal{B}_{\text{eq}}$\ as expected. At EQ in $\mathcal{B}_{\text{eq}}$,
the entropy is given by $S_{\text{fin}}=N\ln(2eV_{\text{in}}/N)$, which gives
$\Delta_{\text{i}}S=N\ln2$, as expected. We can also take the initial
macrostate to be not an EQ one in $\mathcal{P}$ by using one or more
additional internal variables. Thus, the approach is very general.

\subsection{Quantum Free Expansion \label{Sec-QuantumExpansion}}

The sudden expansion has been studied
\cite{Bender,Gujrati-QuantumHeat,Gujrati-JensenInequality} quantum
mechanically (without any $\xi_{\text{N}}$) as a particle in an isolated box
$\Sigma_{0}$\ of length $L_{\text{fin}}$, which we restrict to $2L_{\text{in}%
}$ here, with rigid, insulating walls. We briefly revisit this study and
expand on it by introducing a $\xi_{\text{N}}$ to parallel the study of the
classical expansion above but using the $\mu$NEQT. Thus, we will closely
follow the microstates and follow Ref. \cite{Gujrati-JensenInequality} closely.

We make the very simplifying assumptions in the previous section to introduce
$\xi_{\text{N}}$. At time $t=0$, all the $N$ particles (or their
wavefunctions) are confined in EQ in the left chamber of length $L_{\text{in}%
}$ so that $N_{\text{L}}=N$ initially. We can think of an intermediate length
$L_{\text{fin}}\geq L(t)>L_{\text{in}}$, in analogy with $V(t)$ in the
previous section, so that $N_{\text{R}}=N-N_{\text{L}}$ particles are
simultaneously confined in the intermediate chamber of size $L(t)$, while
$N_{\text{L}}$ particles are still confined in the left chamber for all $t>0$.
This is slightly different from what we did in the previous section.
Eventually, at $t=\tau_{\text{eq}}$, all the $N_{\text{R}}=N$ particles are
confined in the larger chamber of size $L_{\text{fin}}$ so that there are no
particles are confined in the initial chamber. We let $\xi_{\text{N}%
}=N_{\text{R}}/N$, which gradually increases from $\xi_{\text{N}}=0$ to
$\xi_{\text{N}}=1$. Note that this definition is different from the previous
section but we make this choice for the sake of simplicity.\ At some
intermediate time $\tau^{\prime}<\tau_{\text{eq}}$ that identifies
$\mathcal{P}^{\prime}$, $L(t)=L_{\text{fin}}$, but $N_{\text{R}}$ is still not
equal to $N$ ($\xi_{\text{N}}\neq0$). We then follow its equilibration during
$\mathcal{P}^{\prime\prime}$ as the gas come to EQ in the larger chamber at
the end of $\overline{\mathcal{P}}$ when $\xi_{\text{N}}=1$. Again, there are
two internal variables $L$ and $\xi_{\text{N}}$. The expansion is isoenergetic
at each instant. As we will see below, this means that it is also isothermal.
However, $dQ=dW\neq0$ ensuring a irreversible process so the microstate
probabilities continue to change.

Since we are dealing with an ideal gas, we can focus on a single particle
whose energy levels are in appropriate units $E_{k}=k^{2}/l^{2}$, where $l$ is
the length of the chamber confining it. The single-particle partition function
for arbitrary $l$ and inverse temperature $\beta=1/T$\ is given by
\[
Z(\beta,l)=%
{\textstyle\sum\nolimits_{k}}
e^{-\beta E_{k}(l)},
\]
from which we find that the single particle free energy is $\overline
{F}=-(T/2)\ln(\pi Tl^{2}/4)$ and the average single particle energy is
$E=1/2\beta$, which depends only on $\beta$ but not on $l$. Assuming that the
gas is in IEQ so that the particles in each of the two chambers are in EQ (see
the second example in Sec. \ref{Sec-InternalVariables}) at inverse
temperatures $\beta_{\text{L}}$ and $\beta$, we find that the $N$-particle
partition function is given by%
\[
Z_{N}(\beta_{\text{L}},\beta)=\left[  Z(\beta_{\text{L}},L_{\text{in}%
})\right]  ^{N(1-\xi_{\text{N}})}\left[  Z(\beta,L)\right]  ^{N\xi_{\text{N}}}%
\]
so that the average energy is $E_{N}(\beta_{\text{L}},\beta,L_{\text{in}%
},L,\xi_{\text{N}})=N(1-\xi_{\text{N}})/2\beta_{\text{L}}+N\xi_{\text{N}%
}/2\beta$. As this must equal $N/2\beta_{0}$ for all values of $L$ and
$\xi_{\text{N}}$, it is clear that $\beta_{\text{L}}=\beta=\beta_{0}$, which
proves the above assertion of an isothermal free expansion at $T_{0}$.

To determine $\Delta W_{k}$, we merely have to determine the microenergy
change $\Delta E_{k}=E_{k\text{,fin}}-E_{k\text{,in}}$
\cite{Gujrati-GeneralizedWork,Gujrati-JensenInequality}.

Below we will show that the quantum calculation here deals with an
irreversible $\overline{\mathcal{P}}$. The single-particle energy change
$\Delta E_{k}$ is
\[
\Delta E_{k}=k^{2}(1/L^{2}-1/L_{\text{in}}^{2})<0,L>L_{\text{in}}.
\]
The micropressure%
\begin{equation}
P_{k}=-\partial E_{k}/\partial L=2E_{k}/L\neq0 \label{P_k-FreeExpansion}%
\end{equation}
determines the microwork%
\begin{equation}
\Delta W_{k}=%
{\displaystyle\int\nolimits_{L_{\text{in}}}^{L_{\text{fin}}}}
P_{k}dL>0. \label{DW-FreeExpansion}%
\end{equation}
It is easy to see that this microwork is precisely equal to \ $(-\Delta
E_{k})$ as expected. It is also evident from Eq. (\ref{P_k-FreeExpansion})
that for each $L$ between $L_{\text{in}}$ and $L_{\text{fin}}$,
\[
P=%
{\textstyle\sum\nolimits_{k}}
p_{k}P_{k}=2E/L\neq0,
\]
We can use this average pressure to calculate the thermodynamic macrowork
\[
\Delta W=%
{\displaystyle\int\nolimits_{L_{\text{in}}}^{L_{\text{fin}}}}
PdL=2%
{\textstyle\sum\nolimits_{k}}
{\displaystyle\int\nolimits_{L_{\text{in}}}^{L_{\text{fin}}}}
p_{k}E_{k}dL/L\neq0.
\]
as expected. As $\Delta E=0$, this means that the irreversible macroheat and
macrowork are nonnegative and equal: $\Delta Q=\Delta W>0$. This establishes
that the expansion we are studying is \emph{irreversible}.

We now turn to the entire system in which the work is done by $N_{\text{R}}$
particles occupying the larger box. We need to think of the microstate index
$k$ as an $N$-component vector $\mathbf{k}=\left\{  k_{i}\right\}  $ denoting
the indices for the single-particle microstates. For a given $\xi_{\text{N}}$,
we have $\Delta W_{\mathbf{k}}(L,\xi_{\text{N}})=-%
{\textstyle\sum\nolimits_{i}}
\Delta E_{k_{i}}$, where $i$ runs over the $N_{\text{R}}$\ particles. We can
compute the macrowork, which turns out to be $\Delta W_{N}(\xi_{\text{N}%
})=N\xi_{\text{N}}\Delta W>0$. The corresponding change in the free energy is
\begin{align*}
\Delta\overline{F}_{N}(L,\xi_{\text{N}})  &  =N\xi_{\text{N}}[\overline
{F}(\beta_{0},L)-\overline{F}(\beta_{0},L_{\text{in}})]\\
&  =-\Delta W_{N}(\xi_{\text{N}}),
\end{align*}
which is consistent with Eq. (\ref{Irrev-Work-Applied-DelF}) for an isolated
system for any $\xi_{\text{N}}$.

At the end of $\mathcal{P}_{0}$, $\Delta W_{N}(0)=N\Delta W>0$, and
$\Delta\overline{F}_{N}(0)=N[\overline{F}(\beta_{0},L_{\text{fin}}%
)-\overline{F}(\beta_{0},L_{\text{in}})]$. We find that for the isothermal
expansion
\begin{equation}
\Delta W_{N}=-\Delta\overline{F}_{N}=T_{0}\Delta_{\text{i}}S_{N}>0.
\label{DW-DF-DiS}%
\end{equation}
after using Eq. (\ref{EntropyDiff-Isolated}). The same result is also obtained
from the classical isothermal expansion; see Eq.
(\ref{Irrev-Work-Applied-DelF-Del_i_S}). All this is in accordance with
Theorem \ref{Theorem-diW-diS-Isolated} in the MNEQT, as expected.

\section{Discussion and Conclusions\label{Sec-Conclusions}}

As we noted in Sect. \ref{Sec-Introduction}, thermodynamics is a science of
entropy and temperature. As these macroquantities should uniquely describe the
system, we have required them to SI-quantities in developing the new NEQT,
called the MNEQT, to go beyond the EQ\ thermodynamics. We will now briefly
summarize and discuss our conclusions form this thermodynamics. We will
consider them separately.

\subsection{Unique NEQ $S$ in $\mathfrak{S}_{\mathbf{Z}}$\ }

We first point out the important consequence of the restriction imposed by
quasi-independence discussed in Sects. \ref{Sec-Notation} and
(\ref{Sec-Unique-S-T}). By always dealing with the SI-entropy $S$, which we
have shown to be identical to the statistical quantity $\mathcal{S}$ in all
cases, we can appreciate the concept of quasi-independence by considering
$p_{k}$ that appears in Eq. (\ref{Gibbs_Formulation}). Considering $\Sigma$ to
consist of two subsystems $\Sigma_{1}$ and $\Sigma_{2}$, which are in
macrostates $\mathcal{M}_{1}\doteq\left\{  \mathfrak{m}_{k_{1}},p_{k_{1}%
}\right\}  $ and $\mathcal{M}_{2}\doteq\left\{  \mathfrak{m}_{k_{2}},p_{k_{2}%
}\right\}  $. If $\mathcal{M}_{1}$ and $\mathcal{M}_{2}$ are quasi-independent
and form $\mathcal{M}$ for $\Sigma$, then
\[
p_{k}\simeq p_{k_{1}}p_{k_{2}}.
\]
As a consequence, the entropy additivity
\[
S(\mathbf{X}(t)\mathbf{,}t)\simeq S_{1}(\mathbf{X}_{1}(t)\mathbf{,}%
t)+S_{2}(\mathbf{X}_{2}(t)\mathbf{,}t)
\]
is approximately satisfied. This has a generalization to many subsystems
$\left\{  \Sigma_{i}\right\}  $ given in Eqs. (\ref{S-Additivity-1}%
-\ref{S-Additivity-3}) so that they are \emph{all} quasi-independent. In terms
of the volume $\Delta V_{i}$ of $\Sigma_{i}$ so that $V=%
{\textstyle\sum\nolimits_{i}}
\Delta V_{i}$, the generalization can be simply written in the form of the
entropy additivity requirement over $\Delta V_{i}$
\begin{subequations}
\begin{equation}
S(\mathbf{X}(t)\mathbf{,}t)=%
{\textstyle\sum\nolimits_{i}}
S_{i}(\mathbf{X}_{i}(t)\mathbf{,}t), \label{S-Additivity-4}%
\end{equation}
in accordance with quasi-independence. The requirement of quasi-independence
forces the linear size $\Delta l_{i}$ of $\Sigma_{i}$ to be not less than the
correlation length $\lambda_{\text{corr}}$\ as discussed in Sects.
\ref{Sec-Notation} and \ref{Sec-Unique-S-T}; see also \cite{Gujrati-II}. Thus,
there will be no nonlocal effects \cite[for example.]{Jou1,Hutter} to consider
in the MNEQT as they are subsumed within each subsystem. Each subsystem has
its own Hamiltonian $\mathcal{H}$ containing all the information regarding
interactions between its constituent particles and internal variables (see how
$\mathcal{H}$ in Sect. ) so its microstates will contain the effects of all
the interactions in $\left\{  E_{k}\right\}  $.

By a proper choice of $\mathfrak{S}_{\mathbf{Z}}$, $S(\mathbf{X}%
(t)\mathbf{,}t)$ can be replaced by a unique state function $S(\mathbf{Z}%
(t))$. Similarly, by a proper choice of $\mathfrak{S}_{\mathbf{Z}_{i}}$, a
subspace of $\mathfrak{S}_{\mathbf{Z}}$, $S_{i}(\mathbf{X}_{i}(t)\mathbf{,}t)$
can be replaced by a state function $S_{i}(\mathbf{Z}_{i}(t))$ in
$\mathfrak{S}_{\mathbf{Z}_{i}}$. By matching the number of independent
variables $n^{\ast}$ on both sides in Eq. (\ref{S-Additivity-4}) as discussed
in Sect. \ref{Sec-Unique-S-T}, we ensures that $S(\mathbf{Z}(t))$ is uniquely
determined as a sum of $S_{i}(\mathbf{Z}_{i}(t))$ in accordance with Eq.
(\ref{S-Additivity-3}). By replacing $\mathfrak{S}_{\mathbf{Z}_{i}}$ by
$\mathfrak{S}_{\mathbf{X}},\forall i$, and using Theorem
\ref{Theorem-Exixtence-S} and Corollary \ref{Corollary-ProperSubspaces}, we
know that $S(\mathbf{Z}(t))$ uniquely exists in the MNEQT so there is no
freedom to choose any other variables on which $S(\mathbf{Z}(t))$ can depend
on. But the actual choice of $n<<n^{\ast}$ for a given $\mathcal{M}%
_{\text{ieq}}$\ is determined by the experimental setup. It is this
$n>n_{\text{obs}}$ that is physically relevant for $\mathcal{M}_{\text{ieq}}%
$\ unless we are dealing with a $\mathcal{M}_{\text{eq}}$, where
$n_{\text{obs}}$ is the number of independent variables in $\mathbf{X}$. The
remaining $n^{\ast}-n$ internal variables have equilibrated so their
affinities vanish.

We turn back to $S(\mathbf{X}(t)\mathbf{,}t)$ and $S_{i}(\mathbf{X}%
_{i}(t)\mathbf{,}t)$. If $\Delta l_{i}$ is less than the correlation length,
then Eq. (\ref{S-Additivity-4}) must be replaced by
\begin{equation}
S(\mathbf{X}(t)\mathbf{,}t)=%
{\textstyle\sum\nolimits_{i}}
S_{i}(\mathbf{X}_{i}(t)\mathbf{,}t)+S_{\text{corr}}(\mathbf{X}(t)\mathbf{,}t)
\label{S-Additivity-5}%
\end{equation}
due to the correlations present among various $\Sigma_{i}$'s; cf. Eq.
(\ref{Entropy-Additivity}). In a continuum NEQT introduced as CNEQT in Sect.
\ref{Sec-Digression-T}, we are in the limit of "infinitesimal volume element"
over which $S_{i}(\mathbf{X}_{i}(t)\mathbf{,}t)$ can be expressed as
$s(\mathbf{r\mid x}(\mathbf{r,}t)\mathbf{,}t)d\mathbf{r}$; here,
$\mathbf{x}(\mathbf{r,}t)$ is the local analog of $\mathbf{X}_{i}(t)$ over
this volume element. In this limit, Eq. (\ref{S-Additivity-5}) reduces to
\begin{equation}
S(\mathbf{X}(t)\mathbf{,}t)=%
{\textstyle\int}
s(\mathbf{r\mid x}(\mathbf{r,}t)\mathbf{,}t)d\mathbf{r}+S_{\text{corr}%
}(\mathbf{X}(t)\mathbf{,}t), \label{S-Integral-1}%
\end{equation}
where $S_{\text{corr}}$ cannot be converted to a volume integral as it is a
nonlocal quantity over at least the correlation length. Unfortunately,
$S_{\text{corr}}$ is almost invariably overlooked in the CNEQT
\cite{Maugin,Jou0}, which allows the function $s(\mathbf{r\mid x}%
(\mathbf{r,}t)\mathbf{,}t)$ to be commonly identified as the entropy density.
This is most probably a misleading nomenclature as $S_{\text{corr}}%
(\mathbf{X}(t)\mathbf{,}t)$ has been neglected. Even in EQ, the correct
entropy density should be precisely $S(\mathbf{X}(t)\mathbf{,}t)/V(t)$, which
is not $s(\mathbf{r\mid x}(\mathbf{r,}t)\mathbf{,}t)$. With this
approximation, the SI-entropy is replaced by $S_{\text{CT}}(\mathbf{X}%
(t)\mathbf{,}t)$ given by the integral on the right side (CT for the CNEQT),%
\begin{equation}
S(\mathbf{X}(t)\mathbf{,}t)\overset{?}{\simeq}S_{\text{CT}}(\mathbf{X}%
(t)\mathbf{,}t)=%
{\textstyle\int}
s(\mathbf{r\mid x}(\mathbf{r,}t)\mathbf{,}t)d\mathbf{r;}
\label{S-Additivity-6}%
\end{equation}
the questionmark is because it is hard to estimate the error due to the
neglect of $S_{\text{corr}}$. Thus, the additivity of $s$ in the integral is
not the postulated additivity of $S$ even in EQ\ thermodynamics. To ensure
that $s$ satisfies the second law, it is postulated that $s$ shares this
property \cite{Jou0}. Because of these issues, we do not focus on $s$ in this
review as the volume elements are not usually quasi-independent unless we are
at high enough temperatures so that the correlation lengths become small
enough to make them quasi-independent.\ 

The above limitations also distinguish the MNEQT with all CNEQT theories,
which fall under the category of the \r{M}NEQT. Here, we will briefly comment
on two successful theories.

The first one is the extended irreversible thermodynamics (ENEQT) \cite{Jou0},
a well-known CNEQT, which also neglects $S_{\text{corr}}$ but treats the
corresponding entropy density $s_{\text{ET }}$(ET for the ENEQT)$\ $as a state
function involving various dissipative fluxes such as the heat flux. As said
above, one needs to be careful to incorporate nonlocal effects \cite[for
example.]{Jou1,Hutter} in the CNEQT. In addition, the total entropy
$S_{\text{ET}}$ \cite[see Eq. (5.66) and the discussion thereafter]{Jou0} is
also a state function involving the same fluxes for $\Sigma$, which violates
Corollary \ref{Corollary-ProperSubspaces} about requiring a larger state space
relative to $s_{\text{ET}}$: Both $S_{\text{ET}}$ and $s_{\text{ET}}$\ cannot
have the same state space for the additivity of entropy to be an identity; see
also Remark \ref{Remark-Small number-Subsystems}. As the fluxes determine
MI-macroquantities $d_{\text{e}}Q$ and $d_{\text{e}}W$, $S_{\text{ET}}$ is not
a SI-entropy as $S$ is in the MNEQT.

The second one is the MNET \cite{Rubi-Vilar-Reguera} that is based on the idea
of \emph{internal} dof (dof$_{\text{in}}$) proposed by Prigogine and Mazur
\cite{Prigogine2} for a $\Sigma$ in contact with a $\widetilde{\Sigma}$. The
authors provide a very good comparison of the MNET with other important
theories to which we direct the reader. Here, we only compare it with the
MNEQT. The emphasis in the MNET is to study slow relaxation in $\Sigma$ (cf.
Sect. \ref{Sec-Tool-Narayan}) caused by the dof, that we denote here by
dof$_{\text{slow}}$ or by observables $\mathbf{X}_{\text{slow}}$,\ and the
corresponding part of $\Sigma$\ by $\Sigma_{\text{slow}}$; the remainder of
the system is denoted\ by $\Sigma_{\text{fast}}$ with observables
$\mathbf{X}_{\text{fast}}$. In addition to $\widetilde{\Sigma}$,
$\Sigma_{\text{slow}}$ is also allowed to interact with another work medium,
which we denote by $\widetilde{\Sigma}^{\prime\text{(w)}}$ with an extra
prime, with which it exchanges macrowork only; see \cite[Sect. 20]{Landau}.
This makes As is well-known from the Tool-Narayanaswamy equation, see Eq.
(\ref{T_eff}), and other works \cite{Gujrati-II,Langer,Nagel,Liu},
$\Sigma_{\text{slow}}$ and $\Sigma_{\text{fast}}$\ usually have different
temperatures, see Eq. (\ref{T_eff}), pressures, see for example Eq.
(\ref{System-Both Interactions}), etc. and they need not be equal to those of
$\widetilde{\Sigma}$. This is not considered in the MNET, where it is assumed
that $T=T_{0},P=P_{0}$, etc. so $\Sigma$ is assumed to be in EQ with
$\widetilde{\Sigma}$, so there cannot be any internal exchanges between
$\Sigma_{\text{slow}}$ and $\Sigma_{\text{fast}}$.\ The main focus in the MNET
is only on $\Sigma_{\text{slow}}$ and not the entire $\Sigma$. This makes
$\Sigma$ as the realization $\Sigma_{\text{C}}$; see Sect.
\ref{Sec-Applications}. The entropy in the MNET in our discrete notation is
given by
\end{subequations}
\begin{equation}
S_{\text{MNET}}(\mathbf{X},t)=S_{\text{eq}}(\mathbf{X})-H(\mathbf{X}%
_{\text{slow}},t), \label{RRV-1}%
\end{equation}
where%
\[
H(\mathbf{X}_{\text{slow}},t)\doteq%
{\textstyle\sum\limits_{k}}
P_{k}(t)\ln\frac{P_{k}(t)}{P_{k}^{\text{eq}}}%
\]
is the net contribution from $\mathbf{X}_{\text{slow}}$ \cite[Eq.
(7)]{Rubi-Vilar-Reguera}; here, $P_{k}(t)$ is the probability of
$\mathfrak{m}_{k}^{\text{c}}$, the microstate in the \emph{internal
configurational space} (c) formed by dof$_{\text{slow}}$ \cite{Prigogine2},
with $\mathbf{X}_{\text{slow}k}$ denoting its configuration.$\ $As
$\Sigma_{\text{slow}}$ equilibrates with $\widetilde{\Sigma}^{\prime}$,
$P_{k}(t)\rightarrow P_{k}^{\text{eq}}$ so that $H(t)\rightarrow0$.
Consequently, $S_{\text{MNET}}(\mathbf{X},t)\rightarrow S_{\text{eq}%
}(\mathbf{X})$ so that $S_{\text{MNET}}(\mathbf{X},t)-S_{\text{eq}}%
(\mathbf{X})=-H(\mathbf{X}_{\text{slow}},t)\ $is the contribution from the NEQ
dof$_{\text{slow}}$ \cite{vanKampen}. The presence of $P_{k}^{\text{eq}}$,
which surely depends on the conjugate fields of the medium $\widetilde{\Sigma
}_{\text{slow}}$ controlling $\mathbf{X}_{\text{slow}}$, makes $S_{\text{MNET}%
}(\mathbf{X},t)$ an MI-quantity. Thus, it is different from our SI-entropy
$S(\mathbf{X},t)$. Moreover, as $\Sigma_{\text{slow}}$ interacts with
$\widetilde{\Sigma}^{\prime\text{(w)}}$, there is an exchange work
$\Delta\widetilde{W}^{\prime}=-\Delta_{\text{e}}W_{\text{slow}}$ done by
$\widetilde{\Sigma}^{\prime\text{(w)}}$ \cite{Landau}. As $\mathbf{X}%
_{\text{slow}}$ in the MNET is controlled externally, it \ does not represent
an internal variable in the sense used in the MNEQT, which explains the use of
$\mathbf{X}_{\text{slow}}$ and not $\boldsymbol{\xi}$ to represent
dof$_{\text{slow}}$. This is also consistent with Conclusion
\ref{Conclusion-SigmaB-C} since there is no need to consider any internal
variable for the realization $\Sigma_{\text{C}}$. This is further clarified in
the next paragraph. It is easy to see that $\Delta_{\text{e}}W_{\text{slow}}$
satisfies Eq. (\ref{FirstLaw-MI}):
\begin{equation}
\Delta E=T_{0}\Delta_{\text{e}}S-P_{0}\Delta V-\Delta_{\text{e}}%
W_{\text{slow}}; \label{RRV-2}%
\end{equation}
see \cite[see the derivation leading to Eq. (20.1)]{Landau}. Thus, MNET
belongs to the \r{M}NEQT as pointed out above. A configuration temperature for
$\mathfrak{m}_{k}^{\text{c}}$ is also introduced in the MNET by using
$s_{\text{c}k}=-\ln P_{k}(t)$, which is not considered in the MNEQT, where
only a global thermodynamic temperature is defined.

As the examples in Sect. \ref{Sec-Composite System} have revealed, we can
treat $\Sigma$\ either as $\Sigma_{\text{B}}$\ or $\Sigma_{\text{C}}$. We need
internal variables to specify $\Sigma_{\text{B}}$ that help to describe
whatever is going on within $\Sigma$ without knowing these processes. While we
do not need internal variables to specify $\Sigma_{\text{C}}$, we need to know
internal processes such as\ the internal transfer $dE_{\text{in}%
}(t)=d_{\text{e}}Q_{\text{in}}(t)$. Both realizations are equivalent in the
MNEQT. As the entropy is a unique function in $\mathfrak{S}_{\mathbf{Z}}$,
there is no room for any extra dependence such as external fluxes in either
realization; see Theorem \ref{Theorem-Exixtence-S}. The internal fluxes such
as $dE_{\text{in}}(t)$ are needed for $\Sigma_{\text{C}}$, but they are not
controlled by the medium (they are present even if $\Sigma$ is isolated).
Thus, the MNEQT always deals with a SI-entropy.

Thus, the two entropies, $S_{\text{ET}}$ and $S$, are very distinct in many
ways, and cannot be compared as their predictions will be very different.

We now turn to the significance of the new NEQT (MNEQT) in the enlarged state
space, which is a SI-thermodynamics. As macroheat and macrowork are two
independent quantities in the MNEQT, it is clear that the notion of
temperature can be understood by merely focusing of the relationship between
$dQ$ and $dS$ for any arbitrary process; $dW$ plays no role in it. This is
what makes the MNEQT a very useful thermodynamic approach. It should be
stressed that the generalized macrowork $dW$ (the generalized macroheat
$dQ$)\ is not the same as the exchanged macrowork $d_{\text{e}}W$ (exchanged
macroheat $d_{\text{e}}Q$) with the medium unless $d_{\text{i}}W=0$
($d_{\text{i}}Q=0$), i.e., unless there is no \emph{internal irreversibility}
caused by internal processes.

Thus, any deviation of $dW$ from\emph{ }$d_{\text{e}}W$ or $dQ$ from
$d_{\text{e}}Q$ in a process is the result of irreversibility due to internal
processes alone. Indeed, $d_{\text{i}}W$ is the macrowork done internally by
the system against all \emph{dissipative forces within the system}, see Eq.
(\ref{Delta_i-W-Full}), which explains why $d_{\text{i}}W$ is a measure of
dissipative irreversibility (Definition \ref{Def-DissipatedWork}) within the
system. In a similar manner, $d_{\text{i}}Q$ is the macroheat generated
internally by the system, which from $d_{\text{i}}Q=d_{\text{i}}W$ is also due
to all dissipative forces within the system. It must be emphasized that
irreversible macroheat transfer due to temperature difference between a system
and a medium does not affect $d_{\text{i}}W=d_{\text{i}}Q$; see Eq.
(\ref{diS-Qe}). On the other hand, irreversible macroheat transfer between
different \emph{internal} parts of $\Sigma$ will be part of $d_{\text{i}%
}W=d_{\text{i}}Q$ as seen from the discussion of $\Sigma_{\text{C}}$ in Sect.
\ref{Sec-Applications}.

The use of SI-quantities to specify a macrostate and microstates of a system
allows us to determine a statistical definition of the\ entropy of any
$\mathcal{M}_{\text{arb}}$ in Sect. \ref{Sec-NEQ-S}, which is based on the
ideas of Boltzmann. Using the flatness hypothesis, see Remark
\ref{Remark-FlatDistribution}, known to be valid for macroscopic systems, we
provide a simple proof of the second law in Sect. \ref{Sec-SecondLawProof};
see also \cite[Theorem 4]{Gujrati-Symmetry}.

\subsection{Unique NEQ $T$ in $\mathfrak{S}_{\mathbf{Z}}$}

As $dQ$ and $dS$ are SI-macroquantities, their extensivity requires a linear
relation between them for any $\mathcal{M}_{\text{arb}}$ as discussed in Sect.
\ref{Sec_Stat_Concepts}; see Eq. (\ref{System_dQ_dS}). The proportionality
parameter is identified as the temperature $T$ in Eq. (\ref{beta_arb}). With
this extension to deal with $\mathcal{M}_{\text{nieq}}$\ in $\mathfrak{S}%
_{\mathbf{Z}^{\prime}}\supset\mathfrak{S}_{\mathbf{Z}}$, the definition is
applicable to any $\mathcal{M}_{\text{arb}}$. Thus, there is no need to
differentiate between $T_{\text{arb}}$ and $T$ in the MNEQT as said earlier.
The same definition also applies to an isolated system in an arbitrary
macrostate. Determining other fields related to $dW$ such as the pressure do
not pose any new\ complications in the new approach as they are mechanical in
nature as discussed in Sect. \ref{Sec-GibbsRelation}. We have thus proposed a
novel approach to define the unique temperature (see Theorem
\ref{Theorem-Exixtence-S}) that is applicable to any $\mathcal{M}_{\text{arb}%
}$ by selecting the particular $\mathfrak{S}_{\mathbf{Z}^{\prime}}$ where
$\mathcal{M}_{\text{nieq}}$ is convert to $\mathcal{M}_{\text{ieq}}$. Then the
changes in $\left\{  dp_{k}\right\}  $ and $\left\{  E_{k}\right\}  $ identify
$dS,dQ$ and $dW$ for given $T$ and $\mathbf{W}$ in $\mathfrak{S}%
_{\mathbf{Z}^{\prime}}$. All these SI-macroquantities have the same values
also in $\mathfrak{S}_{\mathbf{Z}}$ for $\mathcal{M}_{\text{nieq}}$. All of
this strongly supports Proposition \ref{Proposition-General-MNEQT} as the
fundamental axiom of the MNEQT that can explain the behavior of any
$\mathcal{M}_{\text{arb}}$.

It should be pointed out that the statistical definition of temperature in
Eqs. (\ref{dQ-dS}) or (\ref{EntropyDiff-Isolated}) is not limited to extensive
systems only. The discussion and the conclusions are also valid for systems
for which $dQ$ and $dS$ scale the same way with $N$. We have only considered a
linear scaling between the two SI-macroquantities in this work. It should be
pointed out that our concept of temperature has some similarity with the idea
of a contact temperature for a system in thermal contact with a medium. The
latter is introduced \cite{Muschik,Muschik-2016,Muschik-2020} by the
inequality in Eq. (\ref{HeatFlowDirection}) for a system in thermal contact
with a medium (but not when the system is isolated). We, instead, define the
temperature as an equality in Eq. (\ref{dQ-dS}) for any arbitrary macrostate,
which works even for an isolated system. Eq. (\ref{HeatFlowDirection}) is a
consequence of our definition.

We have seen that this definition of $T$ satisfies the four requirements, see
Criterion \ref{Criterion-Temperature}, listed in Sect.
\ref{Sec-UniqueMicrostates}. This thus solves the dreams of Planck and Landau
\cite{Landau0,Planck}. For example, we need to ensure that \emph{macroheat
}$d_{\text{e}}Q$\emph{, if it is transferred at different temperatures, always
flows from hot to cold}. Indeed, this is a fundamental requirement for a
consistent notion of temperature due to the second law; see Criterion C3. To
the best of our knowledge, this question has not been answered satisfactorily
\cite{Muschik,Keizer,Morriss,Jou,Hoover,Ruelle} for an arbitrary
nonequilibrium macrostate. The question is not purely academic as it arises in
various contexts of current interest in applying nonequilibrium thermodynamics
to various fields such as the Szilard engine \cite{Marathe,Zurek,Kim},
Jarzynski process \cite{Jarzynski}, stochastic thermodynamics \cite{Seifert},
Maxwell's demon \cite{Wiener,Brillouin}, thermogalvanic cells, corrosion,
chemical reactions, biological systems \cite{Hunt,Horn,Forland}, etc. to name
a few. Our approach thus finally \emph{solves} the long-standing unsolved
problem of defining the temperature for an arbitrary macrostate in a
consistent way that satisfy the stringent criteria C1-C4 as proven in Theorem
\ref{Th-TemperatureCriteria}. Thus, $T$ must be treated as a genuine unique
temperature of the system in any macrostate $\mathcal{M}$.

Our definition of the temperature in Eqs. (\ref{dQ-dS}) and
(\ref{EntropyDiff-Isolated}) introduces $T(t)$ as a global quantity, see C4 in
Criterion \ref{Criterion-Temperature}, for the entire system and should not be
confused as a local quantity, which varies from region to region within the
system. This is true even if the system is inhomogeneous. Recall that we have
not imposed any requirement for the system to be homogeneous in our discussion
in Sec. \ref{Sec-UniqueMicrostates}. One may wonder if it makes any sense to
call $T(t)$ the temperature of the system even if it is inhomogeneous. It is
possible to think of an inhomogeneous system to be composed of a number of
homogeneous subsystems $\Sigma_{1},\Sigma_{2},\cdots$, each macroscopic in its
own right. In that case, we can assign a temperature $T_{1}(t),T_{2}%
(t),\cdots$ to $\Sigma_{1},\Sigma_{2},\cdots$, respectively. It is then
possible to relate $T(t)$ to $T_{1}(t),T_{2}(t),\cdots$. We have explicitly
shown this here by considering only two subsystems in Sect.
\ref{Sec-Applications} when they are of identical sizes, and Sect.
\ref{Sec-Tool-Narayan} when they are of different sizes, and treating the
system as $\Sigma_{\text{B}}$; see Eqs. (\ref{beta-A-Composite}) and
(\ref{T_eff}), respectively. For example, we can divide $\Sigma$ into four
subsystems $\Sigma_{1},\Sigma_{2},\Sigma_{3}$, and $\Sigma_{4}$ of equal
volumes and numbers of particles, but of different energies. We can assume
them in their own EQ macrostate $\mathcal{M}_{i}(E_{i},V/4,N/4)$ at
temperature $T_{i}$. Then, we will obtain for $\Sigma=\Sigma_{\text{B}}$
\[
\beta(t)=[\beta_{1}(t)+\beta_{2}(t)+\beta_{3}(t)+\beta_{4}(t)]/4,
\]
which will require three internal variables as shown in Eq.
(\ref{Internal Variable-2}). It is easy to generalize the above relation to
many subsystems and allowing the possibility of different sizes. We can also
allow for volumes to be different for different subsystems as was done in
deriving Eq. (\ref{System-Both Interactions}).

The possibility to study the formation of internal structures in $\Sigma_{i}$
in a NEQ\ $\Sigma$\ should prove very useful to understand what drives their
formation. A very simple example of this is the pattern formation of
Rayleigh-B\'{e}rnard cells and their competition \cite{Hohenberg} in a fluid
system. This pattern formation has received a lot of attention recently
\cite[for example]{Chatterjee}, where stable cells are studied in nonturbulent
convection in steady state. It is found that each cell can be described in one
of its EQ macrostate to a very good approximation with its own temperature
$T_{i}$. What our approach shows is that the stable convection here can also
be described by a thermodynamic constant (steady) temperature $T$ associated
with the steady macrostate of the entire fluid.

Having a global temperature for an inhomogeneous system does not mean that if
we insert a thermometer in it anywhere, we will measure $T$. This is because
the act of "inserting" a thermometer amounts to looking at the "internal"
structure of the system, so we will be probing it as $\Sigma_{\text{C}}$.
Thus, if we insert it in $\Sigma_{1}$, we will record $T_{1}$; and if we
insert it in $\Sigma_{2}$, we will record $T_{2}$, and so on. This should not
be a surprise. We refer the reader to an interesting discussion of this issue
in \cite{Muschik-2016}.

As far as fields such as the pressure that are associated with $dW$ are
concerned, they do not pose the same kind of problem as they are purely
mechanical. All one needs to do is to take their instantaneous averages over
microstate probabilities for any arbitrary macrostate; see, for example, Eq.
(\ref{GeneralizedMacroforce}) involving such an average. This is possible
because $\mathbf{W}$ is a parameter, which makes $\mathbf{F}_{\text{w}k}$
fluctuating quantities over $\mathfrak{m}_{k}$. This cannot be done for the
temperature as $E$ is not a parameter in the Hamiltonian. In this sense, we
are considering a NEQ version of the canonical ensemble in the MNEQT, which
makes $E_{k}$ fluctuating over $\mathfrak{m}_{k}$. Thus, $T$ plays the role of
a "parameter." For this reason, there is no way to define a temperature
$T_{k}$ for $\mathfrak{m}_{k}$ and then take its average. What we can do in
the MNEQT is to use the temperature of various subsystems to obtain $T$ as is
done in Eqs. (\ref{beta-A-Composite}) or (\ref{T_eff}).

We have shown that the definition of the irreversible macrowork $d_{\text{i}%
}W$ is always \emph{nonnegative} as required by the second law; see Eq.
(\ref{WorkInequality}). Various consequences of the second law are discussed
in Sect. \ref{Sec-IrreversibleInequalities}. We have shown that, once a model
for a system is given, we can identify the required number and nature of
internal variables as a computational scheme in Sects.
\ref{Sec-InternalVariables} and \ref{Sec-Applications}, and later sections in
the second half of the review. These applications provide a clear strategy,
once a model has been created, for computation for an arbitrary thermodynamic
process and should prove useful in the field.

We have mostly alluded to $\mathcal{M}_{\text{ieq}}$'s above to highlight the
importance of internal variables in $\mathfrak{S}_{\mathbf{Z}}$, and to
$\mathcal{M}_{\text{nieq}}$'s for memory-effects with respect to
$\mathcal{M}_{\text{ieq}}$'s in $\mathfrak{S}_{\mathbf{Z}}$. In the absence of
a reliable model, finding $\mathfrak{S}_{\mathbf{Z}}$ in many cases may not be
easy to do. Compared to this, the identification of the state space
$\mathfrak{S}_{\mathbf{X}}$ is almost trivial based on the experimental setup.
Therefore, it is much more convenient to work with $\mathfrak{S}_{\mathbf{X}}%
$, with respect to which all NEQ states possess memory. Thus, \emph{the novel
approach we develop here is extremely useful as it does not require knowing
the internal variables} as discussed in Sect. \ref{Sec-M_nieq}. However, for
completeness, we have developed the MNEQT in $\mathfrak{S}_{\mathbf{Z}}$,
which can be easily adapted to $\mathfrak{S}_{\mathbf{X}}$ by the procedure
outlined in Remark \ref{Remark-IEQ-ARB-Macrostate}.

It should be stressed, as noted in Remarks \ref{Remark-Regardless-of-Speed}
and \ref{Remark-dQ-dS-GeneralRelation} that both $dQ$ and $dS$ exist, and so
does their relation in Eq. (\ref{dQ-dS}), regardless of the speed of the
arbitrary process $\mathcal{P}_{\text{arb}}$. This makes the Clausius equality
extremely important and useful as there is no restriction on its validity. It
is a genuine equality even in the presence of irreversibility without any
restriction on the process. This should be contrasted with the conventional
form of Clausius's inequality in Eq. (\ref{ClausiusInequality}); the equality
here holds only in the absence of irreversibility.

The existence of a unique $T$ also appears in the microstate probabilities,
see Sect. \ref{Sec-MicrostateProbabilities} that can be used to determine
various fluctuations of interest. These probabilities for $\mathcal{M}%
_{\text{ieq}}$ also give a generalization of the EQ partition function to a
NEQ\ partition function in Eq. (\ref{NEQ-PF}). Because of the space
limitation, we did not cover its consequences.

\subsection{Applications}

We now come to the various applications of the MNEQT in the later half of the
review. The main lesson here is that several applications cannot be carried
out in the \r{M}NEQT. Apart from the many applications in Sect.
\ref{Sec-Applications} that we have already discussed above, we have applied
it to glasses when we derive the famous Tool-Narayanaswamy equation in Sect.
\ref{Sec-Tool-Narayan}. It is a phenomenological equation for which we provide
a theoretical justification within the MNEQT. We study an irreversible Carnot
cycle in Sect. \ref{Sec-CarnotCycle} and derive its efficiency in terms of the
entropy generation $\Delta_{\text{i}}S$ and show how it differs from the that
of a reversible Carnot cycle. We also show how to compute $\Delta_{\text{i}}S$
for a simple case in which each segment is irreversible but between EQ
macrostates; see Eq. (\ref{DeltaS_i-Carnot}).

The next important application is about friction and the Brownian motion in
Sect. \ref{Sec-Friction}. By considering the relative motion between $\Sigma$
and $\widetilde{\Sigma}$, we theoretically predict the well-know empirical
fact that friction is caused by the relative motion. We apply the approach to
a system of piston in a cylinder, a moving particle-spring system in a fluid,
and just a particle fluid system.

The last application is on free expansion in Sect. \ref{Sec-Free Expansion}.
Here, we consider classical and quantum expansion. In both cases, we make a
simple model of the process and show how it can used to determine
$\Delta_{\text{i}}S$ between not only two $\mathcal{M}_{\text{eq}}$
macrostates but also between two $\mathcal{M}_{\text{nieq}}$; the latter
cannot be determined in the \r{M}NEQT.

\subsection{Summary}

To summarize, we have given a detailed review of the MNEQT in an extended
state space that was initiated a while back
\cite{Gujrati-II,Gujrati-Heat-Work,Gujrati-Entropy2}. Its main attraction is
the variety of new applications, many of which cannot be investigated in the
\r{M}NEQT in which internal variables play no direct role. The approach is
applicable to a system in any arbitrary macrostate $\mathcal{M}_{\text{arb}}$
and is used to provide a unique but very sensible definition of the
temperature, which satisfies all of its important requirements. The useful
aspect of the statistical approach needed for the MNEQT is that it provides a
\emph{unique} definition of generalized macroheat and macrowork $dQ$ and $dW$,
respectively, that are independent contributions in the generalized first law
in Eq. (\ref{FirstLaw-SI}); both quantities are system intrinsic and obey the
conventional partitioning in Eq. (\ref{Y-partition}) valid for any process.
These macroquantities differ from the exchange macroheat and macrowork
$d_{\text{e}}Q$ and $d_{\text{e}}W$, respectively. Therefore, the MNEQT
directly considers the irreversible components $d_{\text{i}}Q$ and
$d_{\text{i}}W$ that originate from all \emph{internal} dissipation within the
system and satisfy an important identity $d_{\text{i}}Q\equiv d_{\text{i}}%
W>0$, see Corollary \ref{Cor-PositiveIrreversibleWork}, for any arbitrary
irreversible process. The irreversible macroquantities vanish for a reversible
process. The identification of a global and unique temperature $T$ is the most
significant aspect of the MNEQT in that it allows us to deal with $\Sigma$\ as
a blackbox so that we do not need to know its interior. This requires a
certain number of internal variables, which explains the extended state space.
We similarly define other fields like the pressure, etc. statistically in
terms of generalized "mechanical" forces; these also include generalized
forces for internal variables. All these definitions are instantaneous and are
not affected by how slow or fast any arbitrary process is. The latter only
determines the time window of relaxations of the internal processes, and the
choice of the state space. We believe that our novel approach provides a
first-ever definition of the temperature, pressure, etc. and of $dQ$ and $dW$
for any arbitrary macrostate, whether the system is isolated on in a medium.
Our approach is also valid to investigate nonequilibrium macrostates with
respect to $\mathfrak{S}_{\mathbf{X}}$, which brings memory effects in the
investigation. Thus, the approach is applicable in a wide variety of
situations, and fulfils Planck's dream.


\begin{thebibliography}{999}
\bibitem {Note}By definition, an observable is an extensive variable that can
be controlled from outside the system such as a medium with which the system
interacts. In contrast, an internal variable, to be introduced later, is an
extensive variable that cannot be controlled from outside as it is generated
internally within the system. \
\bibitem {Note-Landau}See Eq. (7.9) in \cite{Landau}, which shows how this
entropy formulation emerges in statistical physics. It is applicable to both
EQ and NEQ macrostates as is clear from Sec. 40 (see Eq. (40.7) in particular)
dealing with NEQ ideal gas.
\end{thebibliography}
\end{document}